\documentclass[%
reprint,
superscriptaddress,
nofootinbib,
 amsmath,amssymb,
 aps,
]{revtex4-2}

\usepackage{graphicx, subfigure, float}
\usepackage{dcolumn}
\usepackage{bm}
\usepackage{hyperref}
\usepackage{physics}
\usepackage{amsthm}
\usepackage{bbold}
\usepackage{thm-restate}
\usepackage{upgreek}
\usepackage{mathtools} 
\usepackage{xcolor}

\usepackage{amsthm} 
\usepackage{ulem}
\usepackage{dsfont}
\usepackage[most]{tcolorbox} 
\usepackage{array}

\usepackage{hyperref}
\usepackage{url}

\hypersetup{breaklinks=true}

\newtheorem{thm}{Theorem}
\newtheorem{lemma}{Lemma}

\newtheorem{defi}{Definition}
\newtheorem{prop}{Proposition}

\newtheoremstyle{noparen}%
  {\topsep}   
  {\topsep}   
  {\itshape}  
  {}          
  {\bfseries} 
  {.}         
  { }         
  {\thmname{#1}\thmnote{ #3}} 

\theoremstyle{noparen}
\newtheorem*{thm*}{Theorem}

  {\par\noindent\textbf{Remark: }\ignorespaces}%
  {\par}
\usepackage{siunitx}

\usepackage{xurl}   
\usepackage{hyperref}
\hypersetup{breaklinks=true}
\PassOptionsToPackage{hyphens}{url}\usepackage{hyperref}

\usepackage{xcolor}

\definecolor{verdeescuro}{rgb}{0.0, 0.7, 0.0}

\newcolumntype{C}{>{\centering\arraybackslash}X}
\newcolumntype{M}[1]{>{\centering\arraybackslash}m{#1}}

\begin{document}

\title{Emergent Quantum Dynamics as a Bayesian Inference Problem: A Critical Analysis}



\author{Thales B. S. F. Rodrigues}
\email{thalesbsfrodrigues@gmail.com}
\affiliation{Universidade Federal de Juiz de Fora, Departamento de Física, Juiz de Fora, MG, Brasil}
\affiliation{Universidade de São Paulo, Instituto de Matemática e Estatística, São Paulo, SP, Brasil}
\author{Lucas L. Brugger}
\email{lucasbrugger7@gmail.com}
\affiliation{Universidade Federal de Juiz de Fora, Departamento de Física, Juiz de Fora, MG, Brasil}
\author{Vinicius G. Valle}
\email{viniciusvalle2000@gmail.com}
\affiliation{Universidade Federal de Juiz de Fora, Departamento de Física, Juiz de Fora, MG, Brasil}

\author{Bruno F. Rizzuti}
\email{brunorizzuti@ufjf.br}
\affiliation{Universidade Federal de Juiz de Fora, Departamento de Física, Juiz de Fora, MG, Brasil}
\author{Cristhiano Duarte}
\email{cristhianoduarte@gmail.com}
\affiliation{Universidade Federal de Juiz de Fora, Departamento de Física, Juiz de Fora, MG, Brasil}
\affiliation{Institute for Quantum Studies, Chapman University, One University Drive, Orange, CA, 92866, USA}
\affiliation{Instituto de Física, Universidade Federal da Bahia, Campus de Ondina, Rua Barão de Geremoabo, s.n., Ondina, Salvador, BA 40210-340, Brazil}
\affiliation{Fundação Maurício Grabois, R. Rego Freitas, 192 - República, São Paulo - SP, 01220-010, Brazil}



\begin{abstract}

Coarse-grained descriptions can be used to account for physical processes in which information is lost or not entirely accessible. In this paper, we start by proposing a connection between effective, coarse-grained descriptions of quantum dynamics and the quantum conditional states formalism. In doing so, we address necessary and sufficient conditions for the existence of emergent dynamics from a subjective Bayesian point of view. Although our solution is (quasi-)optimal, the dynamics it determines are shown to be analytically limited---it solves the problem in a state-by-state case. Due to this limitation, we then implement semidefinite programming techniques to investigate the existence of effective dynamics in four paradigmatic scenarios. The existence of such an effective dynamics motivates the introduction of a new robustness measure that quantifies how much noise can be added to a microscopic dynamics without compromising its compatibility with a given coarse-grained description. Finally, we also show how one can analytically determine a valid emergent description in several examples. 

\end{abstract}

\maketitle
\section{Introduction} \label{Sec.Intro}

Quantum theory is a system of rules and norms adapted to describe a certain class of naturally occurring phenomena whose classical description is deemed impossible or simply unfit for purpose. Its predictive power in situations of uncertainty~\cite{HaaseEtAl16}, together with its inherently agent-centric character~\cite{BFZ16, Pitowsky05}, may lead us not only to use it, but to literally see it as a theory of probabilistic assignments more general and more resourceful than its classical counterpart~\cite{AruteEtAl19, CG19}. Seen from this angle, it is not entirely surprising that quantum theory accommodates both specific and classically counterintuitive features. If it is supposed to deal with a larger class of phenomena, it is potentially more resourceful and, consequently, has a greater explanatory power---one that rarely maintains intact our classical preconceptions.

The list of such perplexing, distinctly quantum aspects is not limited to foundational questions kept hidden miles away from any real-world application. Over the past decades, we have seen the emergence of several new quantum technologies whose underpinning mechanisms, functioning or even security are deeply rooted in those apparently counterintuitive (quantum) features such as multipartite Bell non-locality~\cite{BrunnerEtAl14}, quantum contextuality~\cite{BudroniEtAl22, Spekkens05}, entanglement~\cite{HorodeckiEtAl09}, no-cloning~\cite{WZ82}, no-broadcasting~\cite{BarnumEtAl07}, incompatibility~\cite{GuhneEtAl23}, quantum superposition~\cite{Nielsen_Chuang_2010}, quantum memory effects~\cite{RHP14} and etc. We transitioned from textbook cryptographic protocols~\cite{Nielsen_Chuang_2010} to near-commercial noisy intermediate-scale quantum gadgets~\cite {MQSSM20} in a matter of decades.  

Nonetheless, whether in its more foundational dimension, its more applied side, or at some point in between, the theory still has its own subtleties. Consider, for example, the situation firstly introduced by the authors in~\cite{DCBM17}. In their framework, adapted to deal with real-world scenarios~\cite{SM19}, a quantum system is prepared at a given time and measured by a faulty detector after undergoing a unitary evolution. Due to the lack of precision in the detection apparatus, modelled there by a coarse-graining map, one may want to describe the microscopic, unitary evolution by a macroscopic, effective (or emergent) map. The upshot is that the partial loss or scrambling of information in the scenario could make the macroscopic description simpler and more efficiently simulable than its microscopic counterpart, potentially accounting for quantum-to-classical transitions~\cite{DCBM17}.

What several authors have shown time and time again is that this particular coarse-graining scenario is but completely understood. It can be extended over multiple time steps~\cite{DATM20}, turned upside down in a more statistical treatment~\cite{VallejosEtAl22}, analysed through a plethora of different lenses~\cite{DATM20}, or given a more agentic-centric treatment~\cite{Duarte20}. It is precisely in this last standpoint that we locate this work. 

In the absence of a complete characterisation of the existence of a well-defined macroscopic, emergent map, we will further develop the idea that quantum theory is nothing but another theory of probabilistic assignments. In doing so, we demonstrate how the coarse-graining problem can be viewed as a problem of (quantum) Bayesian inference, and how this perspective enables us to obtain a state-dependent, macroscopic emergent dynamics whose principal ingredient is the Bayes' inversion of the coarse-graining map. Having done so, we critically examine this potential (inferential) solution and argue that it is applicable only when the strong constraints of the original scenario rule out any other possibility. It is a desperate resolution to cases where no other emergent dynamics can be found. 

This is where we turn to a computational investigation. Through semidefinite programming, we attempt to answer the following: how good is this proposed (inferential) Bayesian solution when other solutions are available? Put another way, how far, in the diamond norm, is this proposed solution from a general, non-state-dependent solution to the coarse-graining problem? To answer this question, we benchmark our inferential framework against four paradigmatic examples. 

For scenarios where one can analytically determine an effective dynamics, we introduce a new robustness measure. It quantifies how much noise can be added to a given unitary dynamics within a compatible coarse-graining scenario---within a scenario with a viable coarse-grained dynamics---without rendering it incompatible. This quantifier also served as inspiration for proposing a convex optimisation routine that, for a given convex combination parameter, tries to circumvent the lack of emergent dynamics in a coarse-graining scenario within this quantum Bayesian picture.   

There have been other efforts to address quantum-to-classical transitions and defective measurements through coarse-graining mechanisms in the past~\cite{KC07, Kabernik18}. Although they are interesting in their own right, in this work we wanted to establish a dialogue with and advance the standpoint originally proposed by the authors in~\cite{DCBM17}. Also, we emphasise that this contribution should be read in tandem with its complementary work~\cite{DATM20}. There, the authors scratch the surface of this problem in four mathematical perspectives. Because we also provide other entry points to the same problem, we believe our work aligns perfectly with what the authors of~\cite{DATM20} set out to do.

All in all, to make our work self-consistent, we review the coarse-graining problem in Appendix~\ref{SubSec.CGProblem}. The formalism of quantum conditional states, which allows us to advance the main standpoint of this work, is summarized in~\ref{SubSec.CQS}. In sec.~\ref{Sec.CGAsBayesianInference} we merge our approach—the conditional states formalism (CSF)—with the problem under investigation—the coarse-graining scenarios---viewing the latter from a Bayesian perspective. We also showcase to interesting coarse-graining scenarios, the fully classical one in in subsection~\ref{SubSec.SolutionClassicalClassical} and the measure-and-prepare one in subsection~\ref{SubSec.SolutionMandP}. The original (fully-quantum) coarse-graining scenario, along with our potential solution via the Petz map, is approached in Sec.~\ref{SubSec.FullyQuantumPetz}. 
We also present, in Sec.~\ref{Sec.CGSolutionNotPossible}, four distinct scenarios of the coarse-graining problem and analytically analyze them, based on an operational perspective, using what we refer to as laboratory space (lab space) variables. From an operational standpoint, these variables allow us to determine the conditions under which the effective dynamics can, in principle, be obtained.
In Sec.\ref{Sec.SDPAnalysis}, landing in a convex optimization perspective, we critically engage with our own proposed resolution to the problem under the light of semidefinite programming implementations, as well as evaluating the coarse-graining problem through the conditional states formalism as a whole. In Sec.~\ref{sec:NumericalResults}, we present numerical evaluations of our proposed Bayesian solution, together with the results obtained from the numerical implementation of the proposed semidefinite programs. Both analyses are conducted considering the same four concrete coarse-graining scenarios established in Sec~\ref{Sec.CGSolutionNotPossible}. We discuss our findings and point to further investigations in Sec. \ref{Sec.Conclusion}.

\section{The coarse-graining problem from a Bayesian perspective}\label{Sec.CGAsBayesianInference}

In Appendix~\ref{SubSec.CGProblem}, we introduce the coarse-graining problem from its usual standpoint: as a question of commutativity involving completely positive, trace-preserving maps. The possibility of finding an appropriate, effective quantum dynamics describing an underlying unitary evolution $\mathcal{U}_{t}$ subjected to lack or loss of information, determined by a coarse-graining map $\Lambda_{\text{CG}}$, is fully characterised in terms of the existence of a completely positive and trace-preserving (CPTP) map $\Gamma_{t}$ such that 
\begin{align}
    \Gamma_t \circ \Lambda_{\text{CG}} = \Lambda_{\text{CG}} \circ \mathcal{U}_{t}.
    \label{eq:CommutationRelationMain}
\end{align}
See the discussion surrounding Fig. \ref{fig:Coarse_graining_diagram}.

From that vantage point, the meaning of those maps, and also of the state space for that matter, is taken \textit{prima facie}. They are supposed to be naturally and uniquely determined by the faulty detector apparatuses---as an actual property of those apparatuses. Although we do not challenge this perspective, in this work, we begin to advance an alternative view---one that had already been alluded to by the authors of~\cite{DCBM17}, partially addressed in~\cite{DATM20}, and explored in~\cite{Duarte20}. We will assume that quantum theory is but another theory of probabilistic assignments \cite{BFZ16,LS14,Leifer06,LD22}, anchoring our operational standpoint in the formalism presented in Appendix~\ref{SubSec.CQS}. In doing so, it allows us to reframe the coarse-graining scenario as a problem in (subjective) Bayesian inference, as seen, illustratively, in Fig. \ref{fig:CGandCSF}.
\begin{figure}
    \centering
\includegraphics[width=0.80\linewidth]{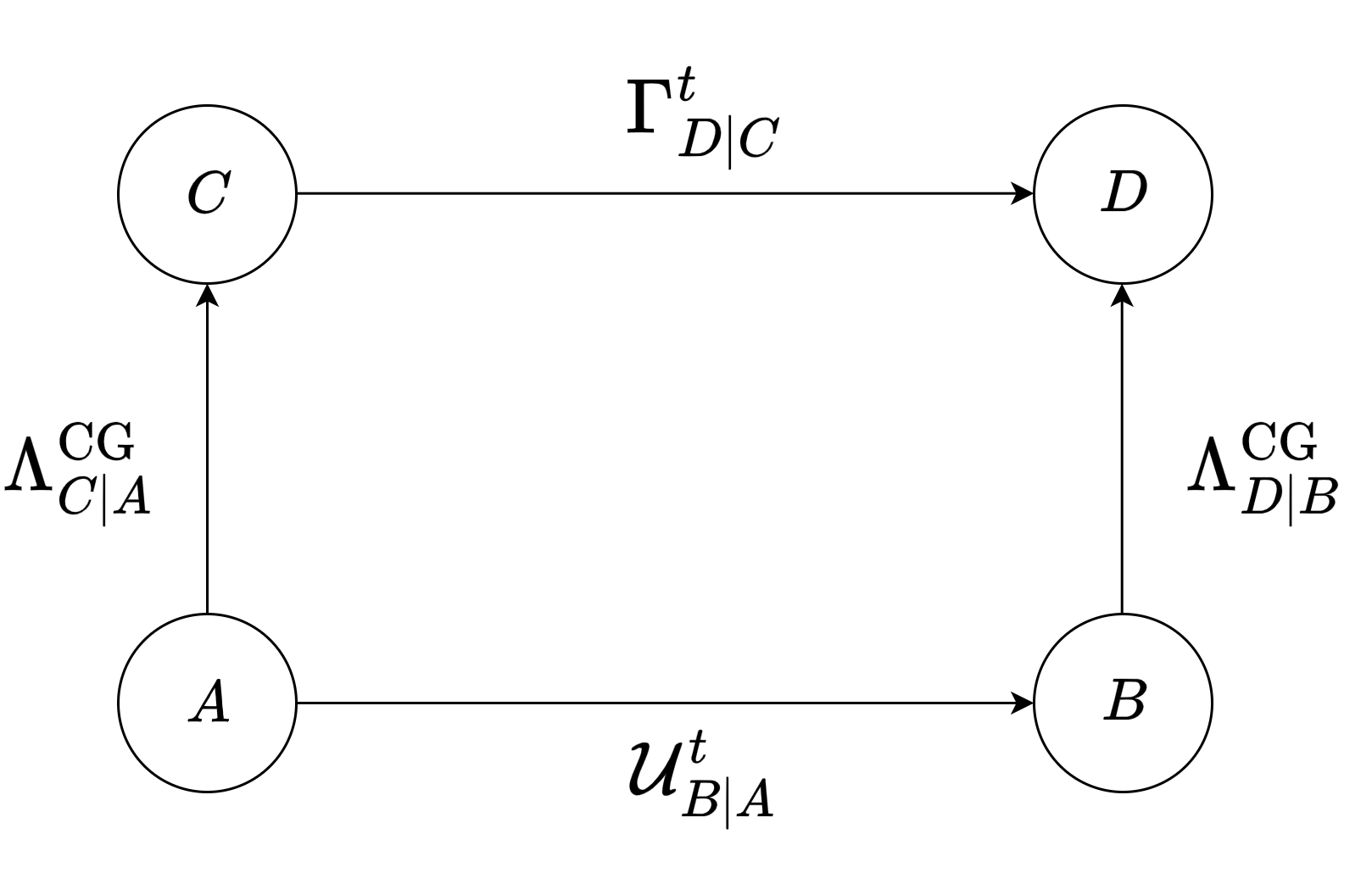} \caption{The coarse-graining problem diagrammatically. To say that the diagram commutes is equal to find a CPTP map $\Gamma_{D|C}^{t}$ such that $ \Gamma_{D|C}^{t} \circ \Lambda_{D|B}^{\text{CG}} = \Lambda_{D|B}^{\text{CG}} \circ \mathcal{U}_{B|A}^{t}$.}
    \label{fig:CGandCSF}
\end{figure}

This agent-centric turn forces us to consider quantum states as (subjectively determined) credal states~\cite{BFZ16} and CPTP maps as mechanisms of belief propagation~\cite{LS13,LS14}. The very use of the coarse-graining map, $\Lambda^{\text{CG}}_{C\vert A}(\Lambda^{\text{CG}}_{D\vert B})$, assumes a clearer interpretation in this picture: this map is now an expression of the experimentalist's (partial) knowledge of their detection apparatus; the emergent map $\Gamma^{t}_{D\vert C}$ is the agent's effective description, based on their knowledge; and the commutativity relation in eq.~\eqref{eq:CommutationRelationMain}, which now can be deployed in terms of the causal conditional states associated to each quantum channel and their respective belief propagation rules (Theorem \ref{teo:ChannelComposition}) is given by
\begin{align}\label{eq:CommutRelationBP} 
    \Tr_{C} [(\mathds{I}_A \otimes \varrho_{D\vert C})&(\varrho_{C\vert A} \otimes \mathds{I}_D)] = \nonumber \\ 
    &=\Tr_B[(\mathds{I}_A \otimes \varrho_{D\vert B})(\varrho_{B\vert A}\otimes \mathds{I}_D)].
\end{align}
In this view, eq.~\eqref{eq:CommutRelationBP} imposes an operational constraint on which types of emergent description the experimentalist is able to come up with. 

But there is more to this picture than those interpretational gains. As anticipated by the author in ref. \cite{Duarte20}, framing the coarse-graining problem as a Bayesian inference task also opens up the possibilities for inferential solutions. 

Although this perspective (and the CSF) provides significantly more powerful computational tools, that allow us to reach distant horizons where the text book formalism of quantum theory can barely touch, this advantage comes at a cost. As we will show in the next sections, the very structure of the Bayesian inversion of conditional states (or probabilities) renders the resulting solutions intrinsically state dependent. 

For example, in classical probability theory, obtaining the conditional probability $P(A \vert B)$ from $P(B \vert A)$ requires prior knowledge of a probability distribution $P(A)$. That is, by constructing the joint distribution and subsequently marginalizing it, one obtains the desired Bayesian inversion. An analogous structure holds for conditional states.

Thus, the necessity of specifying a prior state $\rho_A$, together with the explicit use of quantum Bayesian inversion, is what makes all the solutions, presented in the following sections, state dependent. This dependency can be understood as, for each state one assume as prior, one can find a compatible emergent dynamics, but without any guarantee of compatibility for any another distinct state.

In the following section \ref{Sec.CGSolutionPossible}, we present two simple cases—classical and measure-and-prepare—in which emergent dynamics can be directly achieved, although still paying the price of initial state dependency. In section \ref{Sec.Petzsolution}, we also analyze the fully quantum case, where the explicit dependence of the solution on the prior state becomes clear when we adopt the Petz recovery map \cite{Petz1986} in order to circumvent intrinsic limitations of the scenario.

\section{TWO EXAMPLES OF EMERGENT DYNAMICS - CLASSICAL AND MEASURE-AND-PREPARE}\label{Sec.CGSolutionPossible}

As we mentioned before, to obtain a solution to the coarse-graining problem determined by a coarse-graining map $\Lambda_{D|B}^{\text{CG}}$ and a microscopic unitary dynamics $\mathcal{U}_{B|A}^{t}$ means finding a CPTP map $\Gamma_{D|C}^{t}$ such that the relation \eqref{eq:CommutationRelationMain} holds true.
%
%
Stated like so, the whole problem may sound like a simple mathematical question. But that is not true. Even the canonical examples of the blurred and saturated detector and partial trace, motivating the scenarios addressed in ref.~\cite{DCBM17}, are highly disconcerting. Not only because it gives rise to situations where the diagram of Fig.~\ref{fig:CGandCSF} does not commute; rather, it may happen that $\Gamma^t_{D\vert C}$ is not even a function. In other words, the motivating examples themselves are more akin to a counterexamples than an examples. This is not to say that the case is pathological; on the contrary. As the authors of ref.~\cite{DCBM17} and ref.~\cite{DATM20} make clear, that `pathology' is more the rule than the exception. 

In this section, however, we set off to explore two (or four, as we will soon see) simple solutions to the problem, that highlights the deep connection of the Bayesian inference structure with the proposed solutions.

\subsection{The Fully Classical Case}\label{SubSec.SolutionClassicalClassical}

Coarse-graining scenarios are more varied than the one presented in Fig.~\ref{fig:CGandCSF}, where all regions are quantum. Although our main motivation is precisely the fully quantum one, in this subsection, we study the coarse-graining problem as a Bayesian inference problem when all regions are classical. The main question we address here is, given an underlying belief propagation and two others belief propagation that represent changes in an agent's views, can we determine an emergent belief propagation?

More formally, we assume that: the two microscopic regions $R$ and $S$ are classical; likewise, the macroscopic regions $X$ and $Y$ are also classical; the underlying microscopic evolution is a CPTP channel $\mathcal{E}_{S\vert R}: \mathcal{L}(\mathcal{H}_{R}) \rightarrow \mathcal{L}(\mathcal{H}_{S})$; similarly, the coarse-graining maps are $\mathcal{E}_{X\vert R}: \mathcal{L}(\mathcal{H}_{R}) \rightarrow \mathcal{L}(\mathcal{H}_{X})$ and $\mathcal{E}_{Y\vert S}: \mathcal{L}(\mathcal{H}_{S}) \rightarrow \mathcal{L}(\mathcal{H}_{Y})$; we want to obtain an emergent map $\Gamma = \mathcal{E}_{Y\vert X}: \mathcal{L}(\mathcal{H}_{X}) \rightarrow \mathcal{L}(\mathcal{H}_{Y})$ such that 
\begin{align}\label{Eq.CommutationClassicalCaseSettingNotation}
   \Gamma \circ \mathcal{E}_{X|R} = \mathcal{E}_{Y|S} \circ \mathcal{E}_{S|R}.
\end{align}
In this scenario, each Hilbert space comes equipped with a preferred basis $\{\ket{r}\}_{r \in \mathrm{Out}(R)}$, $\{\ket{s}\}_{s \in \mathrm{Out}(S)}$, $\{\ket{x}\}_{x \in \mathrm{Out}(X)}$, and $\{\ket{y}\}_{y \in \mathrm{Out}(Y)}$ where all operators are diagonal in that basis. Needless to say, $\mathrm{Out}(X)$ represents the set of possible values $X$ may assume. This arrangement is diagrammatically represented in Fig.~\ref{fig_4classcial_statesandchannels}. Each channel above is Choi-Jamio\l{}kowski isomorphic to a classical conditional state as below---see Appendix~\ref{Sec.Formalism}:
\begin{align}\label{Ex.classicalCSXRMainText}
    \rho_{X\vert R} = \sum_{x,r} P(X=x\vert R=r)\ket{r}\bra{r} \otimes \ket{x}\bra{x},
\end{align}
Belief propagation via $\rho_{X|R}$, or the action of its Choi-Jamio\l{}kowski isomorphic channel $\mathcal{E}_{X|R}$, is given by---also see Appendix~\ref{Sec.Formalism} (Lemma \ref{lemm:ChoiActing}): 
\begin{equation}\label{eq:ChoiActingMainText}
    \mathcal{E}_{X|R}(\sigma_{R}) = \Tr_R[\rho_{X|R}(\sigma_{R}^{T} \otimes \mathds{I}_X) ].
\end{equation}

Rewriting the coarse-graining question in this language is paramount to advance our standpoint, and it does so in two fronts. First, it shows how one can readily reframe the coarse-graining problem as an inference problem---where we subsume the conditional probabilities inside each conditional operator. Second, it hints to a possible general (Bayesian) solution: it suffices to revert the leftmost arrow and use the composition rule to define an appropriate emergent dynamics. We will analyse the efficacy of this solution later on in the paper. For now, we will just show how to construct it.  

\begin{figure}
    \centering
    \includegraphics[width=0.65\linewidth]{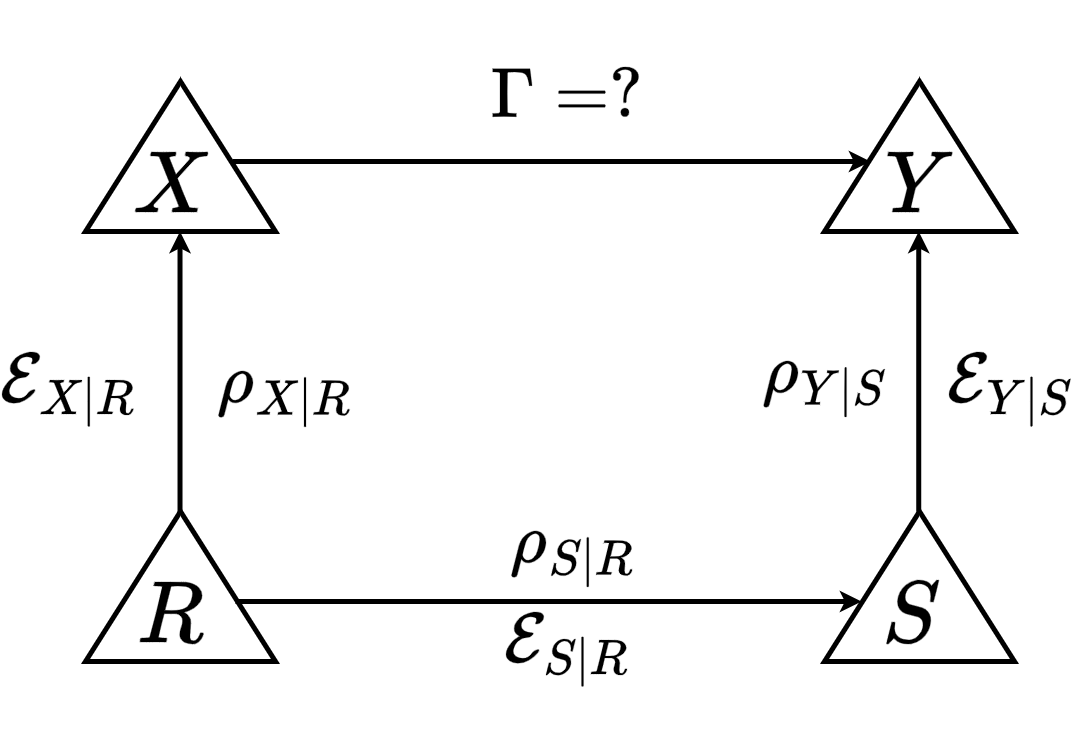}
    \caption{\textbf{Coarse-graining diagram for the full classical scenario}. The two lowermost classical regions are connected via a CPTP channel $\mathcal{E}_{S\vert R}$, while the uppermost classical regions are obtained through the channels $\mathcal{E}_{X\vert R}$ and $\mathcal{E}_{Y\vert S}$, analogous to the coarse-graining map $\Lambda_{\text{CG}}$. Together with the CPTP channels, we also have the representation of their respective Choi-Jamio\l kowski operators.}
    \label{fig_4classcial_statesandchannels}
\end{figure}

In fact, the solution to this case is somewhat direct. We need two ingredients: the quantum Bayes' rule~\cite{LS13} and the composition rule of Theorem~\ref{teo:ChannelComposition}. The former guarantees a well-defined conditional state $\rho_{R|X}$ from $\rho_{X|R}$---see Appendix~\ref{SubSec.AppSupMaterialFullyClassical}:
\begin{align}\label{Ex.classicalInversionXR}
    \rho_{R\vert X} = \rho_{X\vert R} \star \left(\rho_{R} \rho_{X}^{-1} \right).
\end{align}
The latter is to obtain the conditional state between the two upper classical regions of Fig.~\ref{fig_4classcial_statesandchannels} such that the diagram commutes. Combining the two we have
\begin{equation}\label{Ex.classicalCSXYMainText}
    \rho_{Y\vert X} = \text{Tr}_{RS} \left( \rho_{Y\vert S}\rho_{S\vert R}\rho_{R\vert X}\right).
\end{equation}

Equation~\eqref{Ex.classicalCSXYMainText} above is equivalent to obtaining the channel connecting both regions, as in Theorem~\ref{teo:ChannelComposition}, that is,
\begin{equation}\label{Ex.ClassicalCPTPcomp}
    \mathcal{E}_{Y\vert X} = \mathcal{E}_{Y\vert S} \circ \mathcal{E}_{S\vert R} \circ \mathcal{E}_{R\vert X}.
\end{equation}
In sum, we propose here that the emergent dynamics arising in the fully classical coarse-graining problem is given as
\begin{equation}
    \Gamma := \mathcal{E}_{Y\vert X} = \mathcal{E}_{Y\vert S} \circ \mathcal{E}_{S\vert R} \circ \mathcal{E}_{R\vert X}.
\end{equation}

We conclude this subsection by emphasizing that expression \eqref{Ex.classicalInversionXR} plays a central role in the construction of the desired solution. It not only enables the derivation of the effective dynamics \eqref{Ex.ClassicalCPTPcomp}, but also makes explicit its intrinsic state dependence. In particular, we were only able to obtain expression \eqref{Ex.ClassicalCPTPcomp} because we performed the Bayesian inversion of the conditional state $\rho_{X\vert R}$, which required a prior knowledge of the marginal state $\rho_{R}$.

\subsection{Measure-and-Prepare Channels}\label{SubSec.SolutionMandP}

Under the umbrella of an argument similar to that presented in subsection~\ref{SubSec.SolutionClassicalClassical}, namely that coarse-graining scenarios admit many variations, we explore here a particular—and somewhat peculiar—case.

Echoing that reasoning, in the next sections, we will argue that commutativity is not a general feature of the coarse-graining problem---at least it is not general in the fully quantum case. Consequently, in that light, one may be tempted to conclude that, except for a handful of particular combinations of time evolution and coarse-graining maps, there is no hope in finding an emergent map compatible with the fully-quantum diagram. 

We can leverage the fact that classical regions seem to allow for the emergence of a compatible macroscopic dynamics and solve our problem for yet another large class of coarse-graining maps. To do so, it suffices to squeeze in a classical region in between the upper-most and the lower-most levels of the diagrams. In other words, we will consider the case where the coarse-graining is a measure-and-prepare map. 

Thus, the scenario we address here is posed in the following way: let $\mathcal{E}_{C \vert A}^{MP}: \mathcal{L}(\mathcal{H}_{A}) \rightarrow \mathcal{L}(\mathcal{H}_{C})$ and $\mathcal{E}_{D \vert B}^{MP}: \mathcal{L}(\mathcal{H}_{B}) \rightarrow \mathcal{L}(\mathcal{H}_{D})$ be two measure-and-prepare channels and let $\mathcal{U}_{B \vert A}: \mathcal{L}(\mathcal{H}_{A}) \rightarrow \mathcal{L}(\mathcal{H}_{B})$ be a unitary channel. We thus seek for an emergent map $\Gamma^{MP}:\mathcal{L}(\mathcal{H}_{C}) \rightarrow \mathcal{L}(\mathcal{H}_{D})$ such that
\begin{equation}
    \Gamma^{MP} \circ \mathcal{E}_{C\vert A}^{MP} = \mathcal{E}_{D \vert B }^{MP}  \circ  \mathcal{U}_{B \vert A}.
\end{equation}

\begin{figure}
    \centering
    \includegraphics[width=0.60\linewidth]{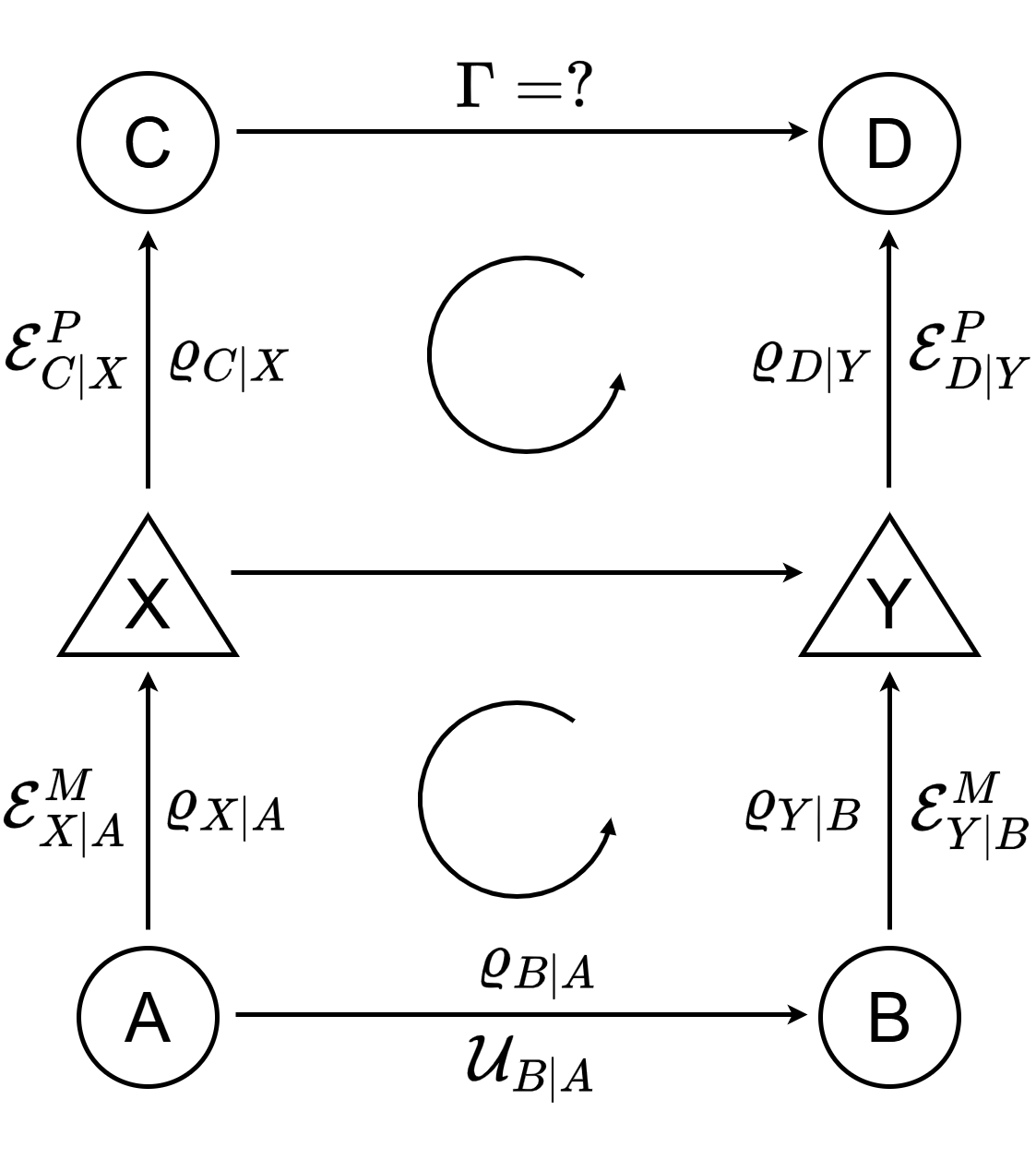}
    \caption{\textbf{Measure-and-prepare diagram.} The measure-and-prepare diagram can be decomposed into two constituent subdiagrams. In the lower subdiagram, the coarse-graining problem is formulated as a measurement process. In the upper subdiagram, it is formulated as an ensemble preparation process. We argue that both subdiagrams exhibit emergent dynamics when suitable measurement and preparation procedures are selected. Here, ``suitable'' refers to procedures that can be consistently interpreted as instances of coarse-graining; not all measurement or preparation protocols satisfy this criterion, but those that do can be naturally embedded within the coarse-graining framework considered in this work. In the left and right branches, we therefore see that measurements $\mathcal{E}_{X\vert A}^{M}$ and $\mathcal{E}_{Y\vert B}^{M}$ are performed on the quantum regions, leading to classical regions. Concomitantly, ensemble preparations $\mathcal{E}_{C\vert X}^{P}$ and $\mathcal{E}_{D\vert Y}^{P}$ are performed conditioned on the measurement outcomes, mapping the classical regions back to quantum regions.}
    \label{fig_measure_and_prepare}
\end{figure}

The diagram in Fig.~\ref{fig_measure_and_prepare} depicts this case. Moreover, upon closer inspection, we observe that the measure-and-prepare diagram is in fact composed of two smaller diagrams. The lower one represents the case in which coarse-graining is treated as a measurement channel, while the upper one represents the case in which coarse-graining is treated as an ensemble preparation channel. Not all procedures falls into this category, and care must be taken whenever we deal with this kind of construction, since not every measurement or preparations necessarily represents a coarse-graining in its very definition. What we emphasize here is that, when a suitable procedure is chosen, we argue that both diagrams exhibit an emergent dynamics in their own right. Such elaboration is made in Appendixes C and D.

The conditional states associated with the measure-and-prepare channels can be constructed as follows. We consider a measurement performed on a quantum region $A$, yielding a classical region $X$. Then, conditioned on the measurement outcomes, ensemble preparations are carried out, mapping the classical region back to a quantum region $C$. Consequently, the conditional state takes the form
\begin{equation}\label{Ex.measureandpreparestate}
\varrho_{C\vert A} = \text{Tr}_{X}\left(\varrho_{C\vert X} \varrho_{X \vert A} \right) = \sum_{x} \rho_{C}^{x} \otimes E_{x}^{A}.
\end{equation}
The construction on the right-hand side of diagram~\ref{fig_measure_and_prepare} is carried out in a similar manner, and a more detailed derivation of the calculations can be found in Appendix~\ref{SubSubSec:MeasureAndPrepare}.

To achieve our solution, that is, the effective dynamics $\Gamma$, we require three ingredients. The first two are as presented in the solution of subsection~\ref{SubSec.SolutionClassicalClassical}, namely the quantum Bayes’ rule and the composition rule of Theorem~\ref{teo:ChannelComposition}. The third ingredient is the measure-and-prepare channel, as introduced above in expression~\eqref{Ex.measureandpreparestate}. Analogously to the first case, the former guarantees that the conditional state $\varrho_{A\vert C}$ is well-defined, that is,
\begin{equation}\label{ex.bayesinainversionmeasureandprepare}
    \varrho_{A\vert C} = \varrho_{C \vert A} \star (\rho_{A}\rho_{C}^{-1}), 
\end{equation}
while the latter gives us that the uppermost conditional state connecting the quantum regions $C$ and $D$ is of the form
\begin{equation}
    \varrho_{D\vert C} = \text{Tr}_{AB}\left(\varrho_{D\vert B} \varrho_{B\vert A} \varrho_{A \vert C}\right).
\end{equation}
The third and final ingredient lead us to treat the coarse-graining problem in a measure-and-prepare perspective.

Thus, the channel that is isomorphic to it corresponds to the emergent dynamics we aimed to obtain. That is,
\begin{equation}\label{Ex.MEPsolution}
   \Gamma_{D\vert C} := \mathcal{E}_{D\vert C} =  \mathcal{E}_{D\vert B}^{MP} \circ \mathcal{U}_{B\vert A} \circ \mathcal{E}_{A\vert C}^{MP}.
\end{equation}
Such a construction is diagrammatically represented in Fig.~\ref{fig_measureandprep_emergent_dynamics}, where we can see arrows that directly connect the quantum regions while ignoring the squeezed classical region between them.
\begin{figure}
    \centering
    \includegraphics[width=0.75\linewidth]{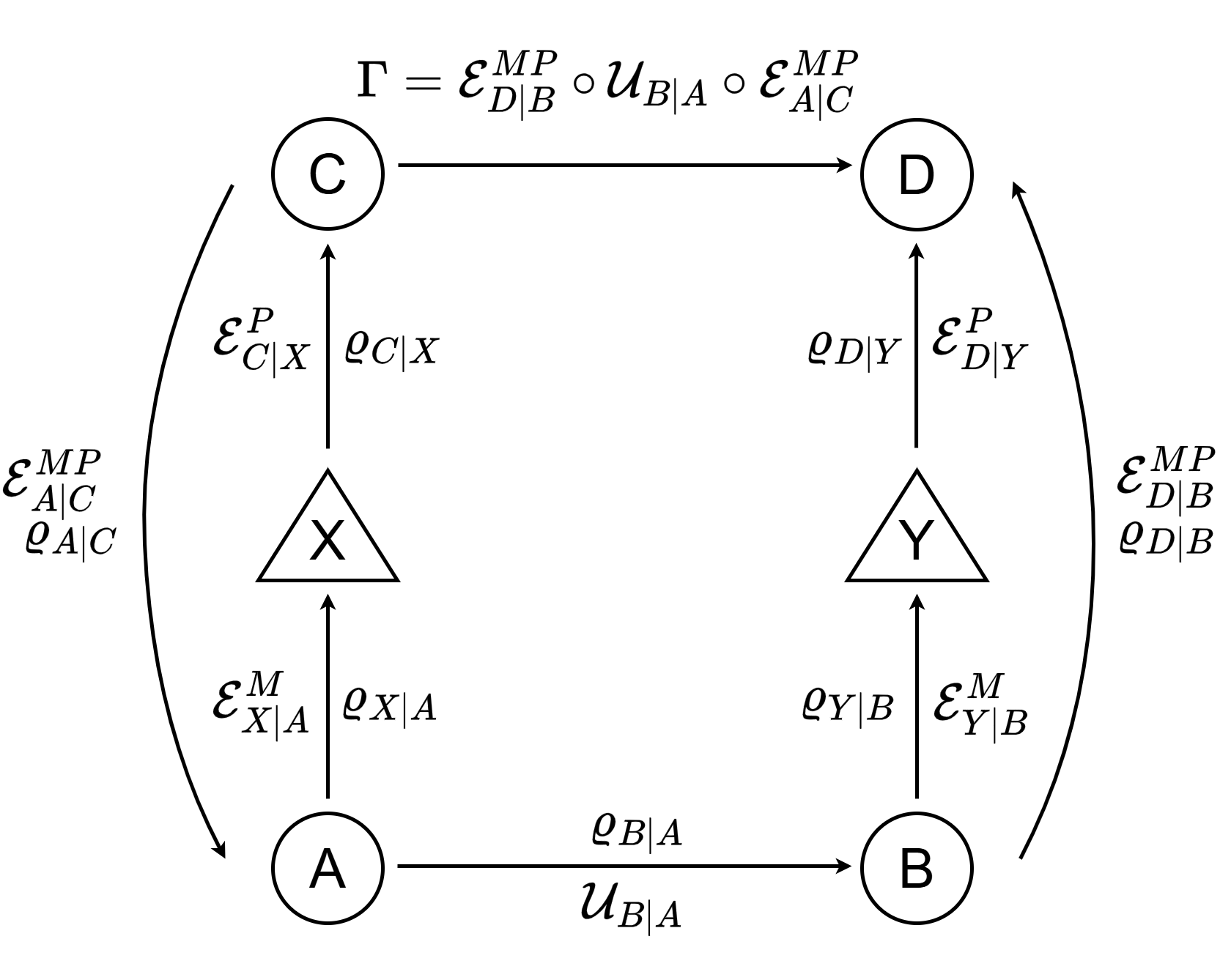}
    \caption{\textbf{Diagram representing the emergent dynamics in the measure-and-prepare scenario.} To obtain the emergent dynamics, the left arm of the diagram — which consists of a measurement followed by an ensemble preparation — gives rise to a conditional state referred to as the measure-and-prepare state. By doing the Bayesian inversion of this conditional state, following the same strategy as before, we are able to obtain the uppermost dynamics, which corresponds to our emergent channel.}
    \label{fig_measureandprep_emergent_dynamics}
\end{figure}

Once again, we establish the emergence of an effective dynamics. Nevertheless, the scenario addressed here is particularly interesting. In section~\ref{SubSec.FullyQuantumPetz}, we will argue that whenever a diagram involves four quantum regions, the problem becomes far from trivial to solve. In that context, our contribution required resorting to the Petz recovery map in order to overcome the limitations encountered.

Oddly enough, when considering measure-and-prepare conditional states, we can effectively hide the classical regions and focus only on the quantum ones. In doing so, we were able to identify a suitable effective dynamics, in direct contrast with the case of four quantum regions without any classical region squeezed between them.

However, the proposed solution is not without the restriction discussed in the preceding sections. Since it is once again embedded in the Bayesian inference structure, that is, expression \eqref{ex.bayesinainversionmeasureandprepare} is required to obtain the emergent dynamics \eqref{Ex.MEPsolution}, we again arrive at an effective dynamics that is intrinsically state dependent.

\section{Petz Recovery Maps as Possible Workaround}\label{Sec.Petzsolution}
\subsection{Introduction}\label{SubSec.PetzMotivation}

Consider the diagram of Fig.~\ref{fig:CGandCSF}. The major issue in the original model is that the first vertical arrow is not invertible.\footnote{And even if it had an inverse, the only case where the inverse is also a physical map would be when it is itself a unitary.} Had it not been the case, we could easily come up with a general solution to the problem, it would suffice to define $\Gamma^{t}:=\Lambda_{\text{CG}} \circ \mathcal{U}^{t} \circ \Lambda_{\text{CG}}^{-1}$, and the diagram would automatically commute. The CSF allows us to find a justified workaround that could resolve the problem in every situation. 

In contrast, the difference of the solution presented here and those given previously lies in the following. When dealing with classical or hybrid (including measure and prepare setups) conditional states, these states are invariant under the transition from the Choi isomorphism to the Jamio\l kowski isomorphism; that is, they retain the same form in both representations. This invariance no longer holds in the fully quantum case, as we will shortly see. Consequently, it becomes necessary to explicitly invoke the Petz map.  

\subsection{A possible workaround - Petz Maps}\label{SubSec.FullyQuantumPetz}

Inverting one of the arrows of diagrams of the form as represented in Fig.~\ref{fig:CGandCSF}, is not an easy—if not an impossible—task. However, this is not the case when we are inside the CSF, and we can explicitly use the Bayesian inversion rule, as we did in the previous sections.

Starting from the usual coarse-graining problem, with the aid of the Fig.~\ref{fig:CGandCSF}, we have that the conditional states isomorphic to their respective channels are,
\begin{align}
    &\varrho_{B\vert A} \cong \mathcal{U}_{B\vert A}^{t}, \label{Ex.isomorphoBA} \\ 
    &\varrho_{C\vert A} \cong \Lambda_{C\vert A}^{\text{CG}}, \label{Ex.isomorphoCA} \\ 
    &\varrho_{D\vert B} \cong \Lambda_{D\vert B}^{\text{CG}}. \label{Ex.isomorphDB}
\end{align}
Once more one resorts to Theorem \ref{teo:ChannelComposition}, obtaining
\begin{equation}\label{Ex.statesolutionfullyquantum}
    \varrho_{D \vert C} = \operatorname{Tr}_{AB}\left(\varrho_{D\vert B} \varrho_{B\vert A} \varrho_{A\vert C}\right),
\end{equation}
and its isomorphic channel
\begin{equation}\label{Ex.firstsolutionfullyquantum}
    \mathcal{E}_{D\vert C} =\Lambda_{D\vert B}^{\text{CG}} \circ \mathcal{U}_{B\vert A}^{t} \circ \Lambda_{A\vert C}^{\text{CG}}.
\end{equation}
The channel \eqref{Ex.firstsolutionfullyquantum} connects the regions $D$ and $C$, that is, the upper part of the Fig.~\ref{fig:CGandCSF}, and therefore is a candidate to be the desired emergent dynamics.

To obtain the conditional state \eqref{Ex.firstsolutionfullyquantum} directly associated with the above solution, and not different of the previous simpler scenarios, we've made use of the Bayesian inversion of the conditional state $\varrho_{C \vert A}$,
\begin{equation}\label{Ex.bayesinverfullyquantum}
\varrho_{A\vert C} = \varrho_{C\vert A} \star (\rho_{A}\rho_{C}^{-1}).
\end{equation}
Once again, this inversion comes at a cost. It not only relies in the own Bayesian structure—thereby rendering the solution state dependent—but also on the fact that the Jamio\l kowski isomorphic state is not positive in general (see Appendix~\ref{APX.PetzSolution} and ref.~\cite{LS13}). Due to this lack of complete positivity, the expression \eqref{Ex.bayesinverfullyquantum} does not yield a valid conditional state.

Nonetheless, within the CSF itself, this limitation can be circumvented. In particular, the emergence of the Petz recovery map \cite{Petz1986} from the channel associated with the inverted conditional state is what guarantees to us the feasibility of the solution. A more detailed discussion of this limitation and the emergence of the Petz map can be found in Appendix \ref{APX.PetzSolution} and, once again, in ref.~\cite{LS13}.

The Petz recovery map is of the form
\begin{equation} \label{eq:PetzMapfullyquatum}
    \mathcal{R}_{A\vert C} (\cdot) = \rho_A^{\frac{1}{2}} \{ (\Lambda^{\text{CG}}_{C \vert A})^{\dagger}[\rho_C^{-\frac{1}{2}} (\cdot) \rho_C^{-\frac{1}{2}}]\}\rho_A^{\frac{1}{2}},
\end{equation}
and this is precisely the channel $\mathcal{E}_{A\vert C}$ isomorphic to the state \eqref{Ex.bayesinverfullyquantum}. In this case, 
\begin{equation}
   \Lambda_{A\vert C}^{\text{CG}} = \mathcal{E}_{A\vert C}  = \mathcal{R}_{A\vert C} .
\end{equation}

The proposed emergent dynamics takes the form,
\begin{equation}\label{Ex.petzeffectivesolution}
    \Gamma^{t}_{D\vert C} \coloneqq \mathcal{E}_{D\vert C} = \Lambda_{D\vert B}^{\text{CG}} \circ \mathcal{U}_{B\vert A}^{t} \circ \mathcal{R}_{A\vert C}.
\end{equation}
Fig.~\ref{fig_fullyquantum_solution} shows our proposed solution along with the conditional states and the Petz map, as well as presenting the isomorphic channels for each of the conditional states.
\begin{figure}
    \centering
    \includegraphics[width=0.6\linewidth]{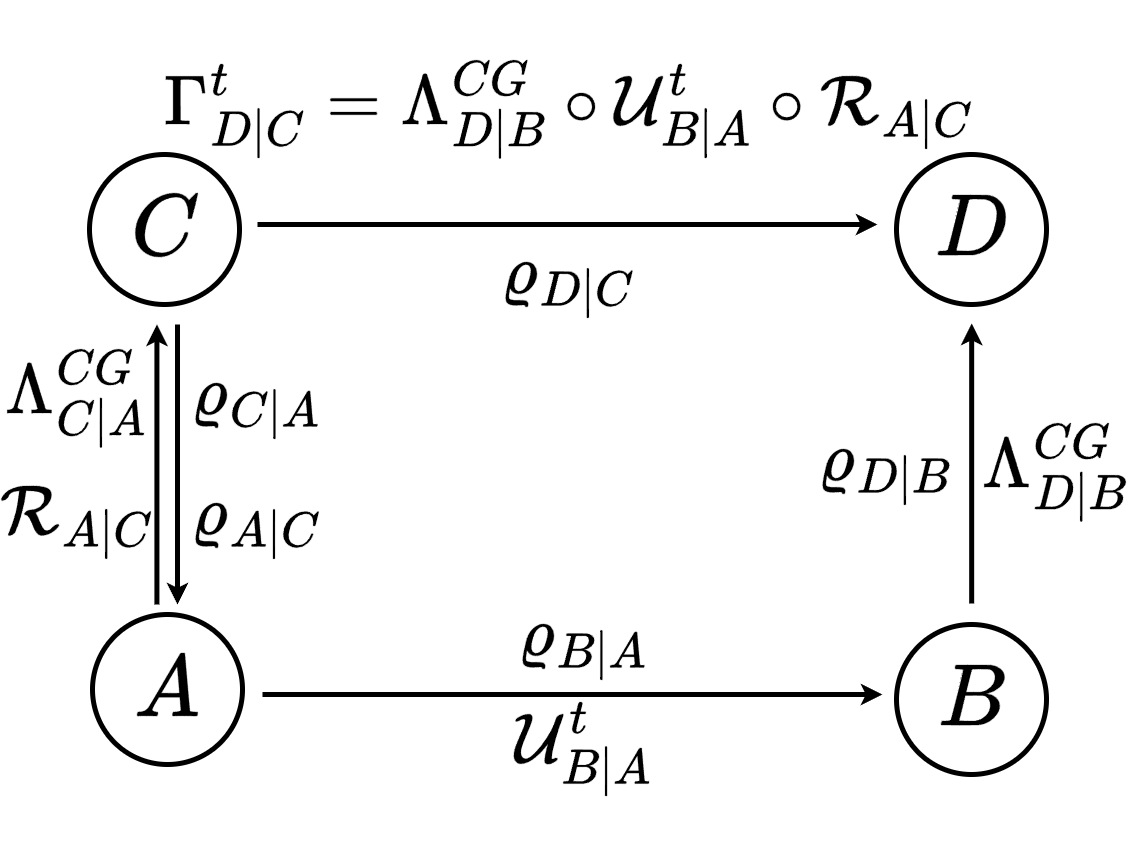}
    \caption{\textbf{Diagram representing the proposed solution to the emergent dynamics.} In the left arm of the diagram, we have the coarse-graining channel together with his isomorphic conditional state, along with the Petz map associated with the Bayesian inversion of the conditional state isomorph to the coarse-graining map. In the upper part of the diagram, we have the proposed emergent dynamics.}
    \label{fig_fullyquantum_solution}
\end{figure}

A closer inspection clarifies why the Petz map \eqref{eq:PetzMapfullyquatum} makes the state dependence of the solution \eqref{Ex.petzeffectivesolution} explicit. Suppose that, in the coarse-graining problem, the only objects available are the coarse-graining map and the unitary map. In this case, one must necessarily introduce a prior state $\rho_{A}$ such that
\begin{equation}\label{Ex.rho_cCG}
    \rho_{C} = \Lambda_{C\vert A}^{CG}(\rho_{A}).
\end{equation}
Moreover, the only way to recover $\rho_{A}$ from \eqref{eq:PetzMapfullyquatum} is to take the input state to be precisely \eqref{Ex.rho_cCG}, since $\Lambda_{C\vert A}^{CG}$ being a trace preserving map, implies that $(\Lambda_{C\vert A}^{CG})^\dagger$ is unital. For any other state in $\mathcal{D}(\mathcal{H}_{C})$, there is no guarantee that \eqref{eq:PetzMapfullyquatum} yields a valid preimage.

In conclusion, we present here yet another instance of the claim advanced in sec.~\ref{Sec.CGAsBayesianInference}. Whether due to the intrinsic Bayesian structure of the solution or, as in the present case, to the explicit appearance of the Petz map, the resulting effective dynamics is inevitably state dependent.

\subsection{How good is this potential solution?}\label{SubSec.FullyQuantumCriticialAnalysis}

The coarse-graining problem, as proposed here, is formulated as the search for an effective dynamics that respects the commutation relation~\eqref{eq:CommutationRelation} globally. However, the solution we propose, given by expression~\eqref{Ex.petzeffectivesolution}, renders the diagram commutative only pointwise, due to the direct state dependence on the Petz recovery map. As a consequence, we do not in fact obtain a global solution and cannot, strictly speaking, argue that our proposed solution $\Gamma_{D\vert C}^{t}$ renders the diagram to be commutative.

As we will see next in the upcoming numerical analysis, although we face the aforementioned limitations in the effective dynamics \eqref{Ex.petzeffectivesolution}, it will still be useful in different contexts, serving as the optimal solution when no other is attainable.

Before proceeding to the numerical analysis, we adopt another (still analytical) complementary approach, namely, verifying—within the setting of a two-qubit system—whether the emergent dynamics is attainable. 

Leveraging on what we will call laboratory space variables, we will obtain conditions that render the diagram to commute. These conditions will be confronted to the numerical analysis ahead. 

\section{The Fully Quantum Case: when a solution is not always possible}\label{Sec.CGSolutionNotPossible}
Now that we have investigated three situations where an appropriated emergent map was established, we can get back to original scenario in which all regions are quantum---where there is no preferred basis by which all density operators are simultaneously diagonalisable. In this fully quantum scenario it is not always the case that there is an emergent quantum dynamics that makes the diagram of Fig.~\ref{fig:Coarse_graining_diagram} commute. 
%

Although our analytical approach was designed to work in full generality, to verify it against concrete benchmarks, we will focus on studying and analyzing bipartite systems of two qubits. This way, we can determine interesting properties without having to care about the quirks of multipartite, high-dimensional quantum scenarios.
 To address this in detail, we begin by presenting two unitary dynamics which, when combined with the two paradigmatic coarse-graining maps $\Lambda_{\text{Tr}}$ and $\Lambda_{\text{BnS}}$ to be presented in a while, allow us to propose and analyze operationally four concrete coarse-graining scenarios.
In all cases, still resorting to the diagram of Fig.~\ref{fig:CGandCSF}, we will consider both regions $A$ and $B$ to be four-dimensional:  $\mathcal{H}_A \simeq \mathcal{H}_B \simeq \mathds{C}^2 \otimes \mathds{C}^2$, whereas the regions $C$ and $D$ will be two-dimensional, that is, $\mathcal{H}_C \simeq \mathcal{H}_D \simeq \mathds{C}^2$.

With the intention of bringing our concrete examples closer to some benchmark models in the literature, the first unitary evolution we choose to incorporate into the following coarse-graining scenarios is a quantum-channel version of the two-qubit SWAP gate. The SWAP operator, denoted in this contribution with respect to the computational basis by 
\begin{equation} \label{eq:SWAPGate}
    U_{\text{swap}} : = \begin{pmatrix}
        1 & 0 & 0 & 0\\
        0 & 0 & 1 & 0\\ 
        0 & 1 & 0 & 0 \\
        0 & 0 & 0 & 1
    \end{pmatrix},
\end{equation}
is a unitary operator that is widely recognized and employed in  quantum computing \cite{Nielsen_Chuang_2010}.
Although its action simply exchanges the quantum state of a two-qubit system, this gate has meaningful applications spread along the literature, one of them is its role as one of the building blocks of the SWAP test \cite{BBDEJRM97}: the best tool in quantum computation used to estimate how close two given quantum states are \cite{Nishimura2025}. 

Given that our concrete examples will be restricted to two-qubit systems on the lower branches of the coarse-graining diagram, our SWAP channel will assume, in its Kraus representation, the following form:
\begin{equation}\label{SWAPeq}
     \mathcal{U}^{swap}_{B\vert A}( \cdot) :=  U_{swap}
    (\cdot) 
  U^{\dagger}_{swap}.
\end{equation}

The second unitary evolution adopted here is based on an adapted version of the Hamiltonian used to model spin interactions of two neighboring particles in a quantum Ising model \cite{Sachdev199}, in a region where the external magnetic field vanishes. More specifically, we employ the Hamiltonian $H = -J \hbar\sigma_{z} \otimes \sigma_{z}$, for a given coupling constant $J$ in units of frequency, to model the interaction between the parts of a two-qubit system determined by the alignment (or anti-alignment) of its spin projections along the $z-$axis. As usual, $\sigma_z=\mbox{diag}(1,-1)$ is one of the Pauli matrices. An important feature of such a Hamiltonian is that the unitary operator arising  from it, denoted in this contribution by 
\begin{equation}\label{eq:U_z}
    U_{\sigma_z}: = e^{-i\frac{Ht}{\hbar}}  = \begin{pmatrix}
        e^{itJ} & 0 & 0 & 0\\
        0 & e^{-itJ} & 0 & 0\\ 
        0 & 0 & e^{-itJ} & 0\\
        0 & 0 & 0 & e^{itJ}\\
    \end{pmatrix},
\end{equation}
models a quantum dynamics in which entanglement is created (or annihilated) between the two parts of the system, depending intrinsically on the value of the time parameter $t$ \cite{Amaral2011MQ}. Thus, with the motivation of bringing to our coarse-graining scenarios a quantum channel that addresses quantum correlation aspects between the subsystems, we define our last unitary dynamics,  
\begin{equation}\label{eq:InteractionChannel}
\mathcal{U}^{\sigma_z}_{B\vert A}(\cdot): = U_{\sigma_z}(\cdot)\ U_{\sigma_z}^\dagger.
\end{equation}
We will name it $z-$interaction channel.

Setting aside both unitary evolutions above, we now look at the most general state in $\mathcal{H}_A \simeq \mathds{C}^4$ , which, according to Bloch's representation \cite{gamel_entangled_2016}, is given by,
\begin{align} 
    \rho_0=\frac{1}{4} \Big (\mathds{I}\otimes\mathds{I}+\sum^3_{i=1}r_i\sigma_i\otimes\mathds{I} &+\sum^3_{j=1}\mathds{I}\otimes s_j\sigma_j  \nonumber \\ \label{rho-rst} &+ \sum^3_{i,j=1}T_{ij}\sigma_i\otimes\sigma_j \Big ),
\end{align}
where $\Vec{\sigma} = (\sigma_1, \sigma_2, \sigma_3)$ is the collection of the Pauli matrices.

With this in mind, our analysis of the coarse-graining problem in this section closely follows the line of research developed in \cite{rizzuti2020operational, valle_towards_2024, grossi_one_2023}, known as the operational approach, in which the authors seek to reformulate and, at times, reconstruct mathematical objects—usually taken as “given”—directly from the physical world that surrounds us, thereby enabling an interpretation of the theory in terms of what are called laboratory variables.

Thus, leveraging on the operational approach, we name the ordered set $\{\vec{r}, \vec{s}, T \}$ as the lab space, due to the intrinsic connection that these variables have with the physical space, as discussed in detail in \cite{valle_towards_2024, grossi_one_2023}. In short, this set corresponds to the variables describing the qubit on one part of the bipartite system under consideration ($\vec{r}$), the qubit on the other ($\vec{s}$) and the correlation matrix, $T$, that encodes the entanglement properties between them, with elements $T_{ij}$.

Due to the fact that our states evolves under particular coarse-graining maps and unitary evolutions, it is in our interest to introduce, in Fig. \ref{fig_Gamma_tilde}, the counterpart mappings $\tilde{\Lambda}^{\text{CG}}_{C\vert A}$, $\tilde{\Lambda}^{\text{CG}}_{D\vert B}$, $\tilde{\mathcal{U}}_{B\vert A}^t$ and $\tilde{\Gamma}_{D\vert C}^t$ related to the 
lab space variables. They are in one-to-one correspondence with the former maps and translates the action of, say $\Lambda^{\text{CG}}_{C\vert A}$ in $\rho_0$ to $\tilde{\Lambda}^{\text{CG}}_{C\vert A}$ applied to the triple $\{\vec{r}, \vec{s}, T\}$ representing $\rho_0$.

Moreover, they will be useful describing whenever an emergent dynamics is attainable and under what conditions $\Gamma_{D\vert C}^t$ (or $\tilde{\Gamma}_{D\vert C}^t$) is indeed definable in the 4 different scenarios to be addressed.
Such results shall be compared with the numerical evaluations carried out.


All subsequent combinations of coarse-graining and unitary transformations that will be explored in the following subsections can be represented by the diagram in the Fig.~\ref{fig_Gamma_tilde}.

\begin{figure}
    \centering \includegraphics[width=1\linewidth]{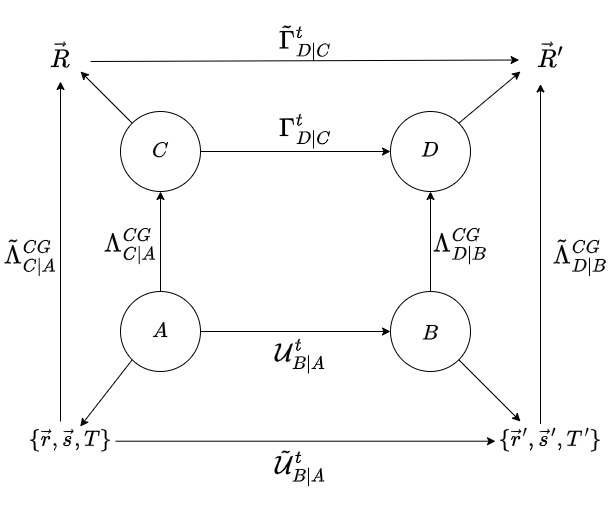}
    \caption{\textbf{Diagram representing the coarse-graining scenario with the addition of the mappings related to the lab space variables.} This diagram introduces, based on operational principles, the mappings $\tilde{\Gamma}_{D\vert C}^t$, $\tilde{\Lambda}^{\text{CG}}_{C\vert A}$, $\tilde{\Lambda}^{\text{CG}}_{D\vert B}$ and $\tilde{\mathcal{U}}_{B\vert A}^t.$ In essence, these new mappings are of the same nature as their counterparts; they merely serve as useful mathematical operational tool in describing our emergent dynamics. As for the the coarse-grained lab space vectors $\vec{R}$ and $\vec{R}'$, they describe their respective density operators in $\mathds{C}^2$ with $\tilde{\Gamma}_{D\vert C}^t(\vec{R}) = \vec{R}'$.}
    \label{fig_Gamma_tilde}
\end{figure}

\subsection{Scenario 1: Blurred and saturated detector and SWAP channel}\label{SubSec.FullyQuantumBnSSWAP}

The blurred and saturated detector basically consists of a device that cannot distinguish between different excited systems. Consider the case of one atom whose state is measured by means of fluorescence: when a laser is shined on it, we associate the state $\ket{1}$ if the light is scattered, and $\ket{0}$ otherwise. In this way, for a system composed of two neighboring atoms we have four possibilities: two excitations, $\ket{11}$; a single excitation, $\ket{10}$ or $\ket{01}$, corresponding to either the first or the second atom; and finally the case with no excitations, $\ket{00}$.

We obtain the blurred and saturated detector by assuming that our device easily saturates under the emission of a single atom and cannot distinguish the light scattered by each atom in the cases $\ket{10}$ and $\ket{01}$, effectively causing both blurring and saturation of the scattered light detected by its lens. 
Thus, the detectors would yield the same outcome whenever measuring the system in any of the states 
$\ket{11}$, $\ket{01}$ or $\ket{10}$.

This faulty detector is represented by the coarse-graining map $\Lambda_{\text{BnS}}:\mathcal{L}(\mathds{C}^4)\to\mathcal{L}(\mathds{C}^2)$, that acts on states through the Kraus representation,
with the following Kraus operators \cite{DCBM17},
\begin{align}
    K_1 &= \begin{pmatrix}
        1 & 0 & 0 & 0\\
        0 & 1/\sqrt{3} & 1/\sqrt{3} & 1/\sqrt{3}\\ 
    \end{pmatrix}; \nonumber \\ 
    K_2 &= \begin{pmatrix}
        0 & 0 & 0 & 0\\
        0 & 1/\sqrt{3} & 0 & -1/\sqrt{3}\\ 
    \end{pmatrix}; \nonumber\\ 
    K_3 &= \begin{pmatrix}
        0 & 0 & 0 & 0\\
        0 & 1/\sqrt{3} & -1/\sqrt{3} & 0\\ 
    \end{pmatrix};\nonumber\\ 
    K_4 &= \begin{pmatrix}
        0 & 0 & 0 & 0\\
        0 & 0 & 1/\sqrt{3} & -1/\sqrt{3}\\ 
    \end{pmatrix}. \nonumber
\end{align}

Acting with this coarse-graining on the general state given in \eqref{rho-rst}, we obtain:
\begin{equation}\label{BnSaction}
    \phi_0 = \Lambda_{\text{BnS}}(\rho_0) =  \frac{1}{2}\left(\mathds{I}+\vec{R}\cdot\Vec{\mathbf{\sigma}}\right),
\end{equation}
where  
\begin{align}
    R_1 &= \frac{r_1+s_1+t_{31}+t_{13}+t_{11}-t_{22}}{2\sqrt{3}}; \label{R1} \\ R_2 &= \frac{r_2+s_2+t_{32}+t_{23}+t_{12}+t_{21}}{2\sqrt{3}}; \label{R2} \\ R_3 &= \frac{r_3+s_3+t_{33}-1}{2}. \label{R3}
    \end{align}
The calculations presented above represent the arrow going from the quantum region $A$ to the quantum region $C$ in the diagram of Fig. \ref{fig_Gamma_tilde}. We still have to calculate the SWAP from $A$ to $B$, as well as the CG from $B$ to $D$, after the latter evolution. 

Upon acting on our state $\rho_0$ with the SWAP operator $\mathcal{U}_{swap}$, one obtains the corresponding result in the lab space variables,
\begin{equation}
    \{\vec{r}, \vec{s}, T \} \xrightarrow{\tilde{\mathcal{U}}_{swap}} \{\vec{s}, \vec{r}, T^\text{T} \}.
\end{equation}
That is, acting on the lab space with $\tilde{\mathcal{U}}_{swap}$ results in the exchange of the vectors $\Vec{r}$ and $\Vec{s}$ and the transposition of  the correlation matrix $T$.

Next, following the arrow from $B$ to $D$ in our diagram, by acting with the coarse-graining of the blurred and saturated detector on our evolved state under the SWAP, we obtain exactly the same result found in equation \eqref{BnSaction}.
The expressions \eqref{R1}-\eqref{R3} are clearly invariant under the exchange $\vec{r} \leftrightarrow \vec{s}$ as well as the transposition in the correlation sector. This result means that the only emergent map $\Gamma^{\text{BnS}}_{swap}$  attainable connecting the quantum regions $C$ and $D$ is represented by the identity on $\mathcal{L}(\mathds{C}^2)$.

This result is further reinforced when we examine the action of the SWAP operator on two qubits states, given by \begin{equation}\label{swap01}
    \ket{00}\to\ket{00}, \ \ \ket{10}\to\ket{01}, \ \ \ket{01}\to\ket{10}\  \ \text{e} \ \ \ket{11}\to\ket{11}. \nonumber
\end{equation} 
 Since our blurred and saturated detector does not distinguish between the states $\ket{01}$ and $\ket{10}$, and since the states $\ket{11}$ and $\ket{00}$ remain unchanged after the SWAP operation, we observe that the action of SWAP on our states does not modify the outcome produced by the blurred and saturated detector, thus corroborating our analysis of this scenario.
\subsection{Scenario 2: Blurred and saturated detector and $z-$interaction channel}\label{SubSec.FullyQuantumBnsSigmaz}

Focusing once more on the blurred and saturated detector, this time under a different unitary evolution represented by eq. \eqref{eq:InteractionChannel}. In contrast to the situation previously analyzed, we will see that it is precisely in this context that the example of the blurred and saturated detector reveals its limitations. 
The issue arises from the fact that, under certain conditions, the $\Gamma$ mapping may not even be consistently definable. As the analysis will indicate, its very existence depends on constraints linked to the initial state $\rho_0$---a fact that, when taken together with the types of emergence addressed in \cite{carroll2024emergencepossiblymean}, makes evident the restricted domain of applicability (for this type of coarse-graining description) within which an emergent dynamics can be defined. 

We now apply the $z-$interaction channel  \eqref{eq:InteractionChannel} to our state and then apply the coarse-graining once again.

For a general $4\times4$ matrix $A$ with elements $(A)_{ab} = \alpha_{ab}$, with $a,b=1,...,4$,
the unitary channel induces the transformation,
\begin{equation}\label{sigma_z-on-rho}
    U_{\sigma_z}A\ U_{\sigma_z}^\dagger = \begin{pmatrix}
        \alpha_{11} & e^{2it}\alpha_{12} & e^{2it}\alpha_{13} & \alpha_{14}\\
        e^{-2it}\alpha_{21} & \alpha_{22} & \alpha_{23} & e^{-2it}\alpha_{24}\\ 
        e^{-2it}\alpha_{31} & \alpha_{32} & \alpha_{33} & e^{-2it}\alpha_{34}\\
        \alpha_{41} & e^{2it}\alpha_{42} & e^{2it}\alpha_{43} & \alpha_{44}\\
    \end{pmatrix}.
\end{equation}

We apply this result to the initial state $\rho_0$, followed by the CG map obtaining,
\begin{equation} \label{phi_0linha}
    \phi_0' = \frac{1}{4}\begin{pmatrix}
        1+r_3+s_3+t_{33} & \frac{\beta_1}{\sqrt{3}} \\
        \frac{\beta_2}{\sqrt{3}} & 3-r_3-s_3-t_{33} \\ 
    \end{pmatrix},
\end{equation}
where
\begin{align}
    \beta_1 =& e^{2it}(r_1+s_1+t_{13}+t_{31}-i(r_2+s_2+t_{23}+t_{32})) \nonumber\\ &+t_{11}-t_{22}-i(t_{12}+t_{21}) ;\\ 
    \beta_1 =& e^{-2it}(r_1+s_1+t_{13}+t_{31}+i(r_2+s_2+t_{23}+t_{32}))\nonumber\\ &+t_{11}-t_{22}+i(t_{12}+t_{21}).
\end{align}

We now analyze and compare both states $\phi_0$, found in equation \eqref{BnSaction} of the previous scenario, and $\phi_0'$ of eq. \eqref{phi_0linha}. To do this, since only the variables $t_{11}, \  t_{22}, \  t_{12}$ and $ t_{21}$ are not multiplied by the factors $e^{2it}$ and $e^{-2it}$, it is in our interest to decompose the coarse-grained lab space vector $\Vec{R}$ present in the equation \eqref{BnSaction} into two parts,
\begin{equation}
    \Vec{R} = \vec{R}_a + \vec{R}_b,
\end{equation}
with 
\begin{align}
\vec{R}_a =& \frac{1}{2\sqrt{3}}[(r_1+s_1+t_{13}+t_{31})\hat{i}+(r_2+s_2+t_{23}+t_{32})\hat{j}]+\nonumber \\ &\frac{1}{2}(r_3+s_3+t_{23}-1)\hat{k}
\end{align}
and
\begin{equation}
    \vec{R}_b = \frac{1}{2\sqrt{3}}[(t_{11}-t_{22})\hat{i}+(t_{12}+t_{21})\hat{j}].
\end{equation}
Thus, we find that $\tilde{\Gamma}^{\text{BnS}}_{\sigma_z}$ acts on the vector $\Vec{R}$ in the following way,
\begin{equation}
    \tilde{\Gamma}^{\text{BnS}}_{\sigma_z}(\Vec{R}) = \mathcal{R}(t)\vec{R}_a + \vec{R}_b,
\end{equation}
where 
\begin{equation}
    \mathcal{R}(t) = \begin{pmatrix}
        \cos2t & \sin 2t & 0 \\
        -\sin 2t & \cos2t & 0 \\ 
        0 & 0 & 1 \\
    \end{pmatrix}.
\end{equation}

The action of $\tilde{\Gamma}^{\text{BnS}}_{\sigma_z}$ on our vector $\vec{R}$ is but a $z$-axis rotation matrix $\mathcal{R}(t)$ to specific components of this vector. In the particular case where $\vec{R}_b$ is zero, the operation can be represented purely by a single proper rotation. In this case where $\vec{R}_b = \vec{0}$, one can immediately write down the explicit form of the emergent dynamics, based on the intrinsic connection between $SU(2)$ and $SO(3)$ groups \cite{valle_rotations_2026}
\begin{equation}
\Gamma^{\text{BnS}}_{\sigma_z}(\cdot) = \mathcal{U}_z(t) (\cdot) \ \mathcal{U}^{*}_z(t).
\end{equation}
Here, $\mathcal{U}_z(t)$ is the $SU(2)$ element corresponding to the  $\mathcal{R}(t) \in SO(3)$ rotation,
\begin{equation}
    \mathcal{U}_z(t) = \begin{pmatrix}
        e^{it} & 0 \\
        0 & e^{-it}
    \end{pmatrix}.
\end{equation}

Clearly, this ideal scenario does not hold in general, as the components of $\vec{R}_b$ are not necessarily null. Consequently, a solution that ensures the commutation of our diagram, for all microscopic states, is not guaranteed in all cases. The feasibility of such a solution depends critically on the parameters $t_{11}, \, t_{22},\,  t_{12}$ and $t_{21}$, namely, $t_{11} = t_{22}$ and $t_{12} = -t_{21}$ (making $\vec{R}_b = \vec{0}$), which are determined by the matrix that encodes the entanglement properties of our state $\rho_0$. Hence, these parameters play a fundamental role in dictating the commutativity of the diagram we are aiming for.

Some recent works concerning thermodynamical descriptions of coarse-graining scenarios are interested in the reversibility of the emergent dynamics \cite{VallejosEtAl22}. Interestingly enough, the condition $\vec{R}_b = \vec{0}$, guaranteeing the very definition of our emergent dynamics, is also the condition that assures the invertibility of $\tilde{\Gamma}^{\text{BnS}}_{\sigma_z}$. This topic, together with the analysis of thermodynamic quantum processes via the conditional states formalism, will be explored elsewhere.

\subsection{Scenario 3: Partial Trace and SWAP channel}\label{SubSec.FullyQuantumTrSWAP}

One of the most fundamental and widely studied examples of coarse-graining maps is the partial trace. It is clear that it's definition naturally embodies the idea of disregarding inaccessible or unobserved degrees of freedom, that is, when a quantum system is described by a density operator $\rho_{SE}$ on the Hilbert space $\mathcal{H}_{S} \otimes \mathcal{H}_E$, the operation $\Tr_{E}(\rho_{SE})=\rho_S$ 
effectively discards any information stored in subsystem E. This is the unique map the consistently describes observable quantities for subsystems of a large system \cite{Nielsen_Chuang_2010}.

From an operational standpoint, this corresponds to a situation in which the observer or detector is unable—or simply not allowed—to access degrees of freedom associated with $E$. In this sense, the partial trace constitutes a paradigmatic example of a completely positive, trace-preserving (CPTP) map that embodies the essence of the coarse-graining procedure: it maps the full microscopic description of the joint system onto a reduced, macroscopic (or accessible) description, thereby suppressing correlations and coherences that lie beyond the resolution of the observation.

It would be particularly interesting to examine how the partial trace affects the state before and after the action of our chosen unitaries, given that this loss of information is inherently irreversible, since the full state $\rho_{SE}$ cannot, in general, be reconstructed from the reduced state $\rho_S$.

Starting with the simplest case, represented by the SWAP, as discussed in the scenario presented in \ref{SubSec.FullyQuantumBnSSWAP}, its action in our state $\rho_0$ has, as one of its consequences, the exchange of the vector $\Vec{r}$ with the vector $\Vec{s}$ in eq.  \eqref{rho-rst}. Thus, if we perform the partial trace over one of the bipartite system before applying the SWAP, we obtain, from \eqref{rho-rst}, 
\begin{equation}
\phi = \frac{1}{2}\left (\mathds{I}+ \vec{r}\cdot \vec{\sigma}\right ),
\end{equation}
whereas performing the same operation after the SWAP yields
\begin{equation}
\phi' = \frac{1}{2}\left (\mathds{I}+ \vec{s}\cdot \vec{\sigma}\right ).
\end{equation}
In the case of the $\tilde{\Gamma}^{Tr}_{swap}$ mapping, the situation becomes even more restrictive as compared to the above scenarios. Since the partial trace entails a complete loss of information about the second part on the bipartite system, no consistent definition of the $\tilde{\Gamma}^{Tr}_{swap}$ map can be established except in the trivial case where the local Bloch vectors $\vec{r}$ and $\vec{s}$ coincide. When this happens, we find that the $\Gamma^{Tr}_{swap}$ is equal to the identity map meaning that no evolution is predicted by the allegedly emergent dynamics. Indeed, for distinct initial states, with $\{\vec{r},\vec{s_1},T\}$ and $\{\vec{r},\vec{s_2},T\}$, not even a $\tilde{\Gamma}^{Tr}_{swap}$ map would be admissible, since one would be forced to have $\tilde{\Gamma}^{Tr}_{swap}(\vec{r})=\vec{s_1}$ and $\tilde{\Gamma}^{Tr}_{swap}(\vec{r})=\vec{s_2}$. We finally point out that this line of reasoning is completely independent of the correlation matrix $T$, since the coarse-graining rules it out.

\subsection{Scenario 4: Partial Trace and $z-$interaction channel}\label{SubSec.FullyQuantumTrSigmaz}

Finally, in our last scenario, since we already know how the $z-$interaction channel acts on our state, see eq. \eqref{sigma_z-on-rho}, it suffices to take the partial trace over $E$ of the $4\times 4$ matrix representing our our time-evolved state $\rho_t$.

We obtain
\begin{equation}
   \Tr_E (\rho_t)= \frac{1}{2}\left(\mathds{I}+\vec{r}\,'\cdot \vec{\sigma}\right),
\end{equation}
where
\begin{align}\label{43}
    \tilde{\Gamma}^{Tr}_{\sigma_z} (\vec{r}) = \vec{r}\,'= \mathcal{D}(t)\vec{r} + \sin 2t \vec{\tau}.
\end{align}
We are denoting $\mathcal{D}(t):= \mbox{diag}(\cos2t, \cos 2t,1)$ and $\vec{\tau}:= (t_{23}, -t_{13}, 0)$.

In this particular scenario, the dynamics is given by a linear transformation induced by $\mathcal{D}(t)$, together with a non-linear oscillating second term in \eqref{43}. This second term depends on specific entries of the correlation matrix that encode the entanglement structure of the composite state---in this case, the components $t_{23}$ and $t_{13}$. A linear (and reversible) dynamics is attained whenever 
\begin{equation}\label{t23=t13=0}
t_{23}=t_{13}=0.    
\end{equation}
If the condition above is satisfied, then we can explicitly write down the emergent dynamics. It is a time-varying phase flip. The Kraus operators are given by $\kappa_1 = \cos t \mathds{I}$ and $\kappa_2 = \sin t \sigma_z$.

Hence, we once again observe that the same underlying mechanism governs both situations, dependence on $T$, reinforcing the crucial role of these parameters in determining the existence and reversibility of  $\tilde{\Gamma}^{Tr}_{\sigma_z}$.

Our results concerning each scenario, with the necessary condition for the emergence of the corresponding emergent dynamics are summarized 
in the table of Fig \ref{fig:emergent_dynamics}.
\begin{figure}
\centering
\resizebox{\columnwidth}{!}{$
\begin{array}{c|c|c}
 \Gamma& \Lambda_{BnS} & \Lambda_{Tr} \\
\hline
\mathcal{U}_{swap}(\cdot) &
\begin{array}{c}
\text{Identity on } \mathcal{L}(\mathds{C}^2)
\end{array}
&
\begin{array}{c}
\text{Identity on } \mathcal{L}(\mathds{C}^2) \\
\text{whenever } \vec{r}=\vec{s}
\end{array}
\\
\hline
\mathcal{U}_{\sigma_z}(\cdot) &
\begin{array}{c}
\Gamma^{\text{BnS}}_{\sigma_z}(\cdot) = \mathcal{U}_z(t) (\cdot) \ \mathcal{U}^{*}_z(t) \\ \text{ whenever } \vec{R}_b=\vec{0}
\end{array}
&
\begin{array}{c}
\Gamma^{Tr}_{\sigma_z}(\cdot) = \kappa_1(\cdot)\kappa^{\dagger}_1 + \kappa_2(\cdot)\kappa^{\dagger}_2  \\
\text{ whenever } \vec{\tau} = \vec{0} 
\end{array}
\end{array}
$}
\caption{Explicit emergent dynamics for different coarse-graining problems.}
\label{fig:emergent_dynamics}
\end{figure}

From the analysis of the scenarios presented above, it becomes clear that the emergent dynamics of the coarse-graining problem under consideration possess a distinct and strong dependence on the diagram utilized, i.e., on the particular combination of coarse-graining and unitary evolutions being addressed. We can, however, identify certain similarities among these scenarios. For instance, in both cases where the $z$-interaction channel is used, we observe a clear dependence on the initial state $\rho_0$, specifically on the variables of the matrix $T$, that encodes the entanglement of $\rho_0$, a situation that will be numerically explored in Sec.~\ref{sec:NumericalResults}.

Similarly, when analyzing the scenarios involving the SWAP operator, in both cases, we have the identity map on $\mathcal{L}(\mathds{C}^2)$ as the resulting emergent dynamics, albeit under vastly different constraints. While there are no restrictions on the initial state in the case of the blurred and saturated detector, the partial trace scenario requires, for the existence of the emergent dynamics, the condition $\vec{r} = \vec{s}$ in Eq.~\eqref{rho-rst}.

It is evident that the analysis presented here is restricted to the specific coarse-graining discussed in this section, and that the similarities observed among these scenarios alone are by no means sufficient to establish a general pattern across different combinations of the coarse-graining problem. That being said, a general characterization of the coarse-graining problems will be the subject of future investigations on this topic.

For now, we will focus on a different method of analysis. Through a semidefinite programming approach and a computational perspective, we reinforce the analysis taken in this section, as well as explore the role of our proposed solution $\Gamma^{\text{Petz}}_{D \vert C, \rho_A}$ (Sec. \ref{SubSec.FullyQuantumPetz}) in these four paradigmatic coarse-graining scenarios.

\section{Semidefinite programming Analysis}\label{Sec.SDPAnalysis}

As discussed in  previous sections, the conditional states formalism, when adopted as an operational standpoint, gives to the coarse-graining problem a Bayesian inference character. Besides this, as noted in the discussions around Fig.\ref{fig:CGandCSF} and stated at eq. \eqref{eq:CommutRelationBP}, it becomes possible to look at the conventional commutativity relation \eqref{eq:CommutationRelation} as a function of the causal conditional states associated with each quantum channel exhibited in the diagram of Fig. \ref{fig:CGandCSF} and in the others shown along the section \ref{Sec.CGSolutionPossible}. However, since, by Definition \ref{def:CausalState}, such states are not positive semidefinite operators in general, this prevents us from directly approaching, at first glance, the coarse-graining scenario within this framework of conditional states through the implementations of semidefinite programs. 

Because of this, in order to establish the appropriate setting to analyze the problem from a semidefinite programming perspective, we rely on the positivity of the acausal conditional states associated with each quantum dynamics and employ these states to study coarse-graining scenarios. We also will adopt eq.  \eqref{eq:compositionRuleChoi} as the basis for expressing the conventional commutativity relation \eqref{eq:CommutationRelation} in terms of the respective Choi-isomorphic operators of each quantum dynamics. 

In this sense, we advise the reader that this representation slightly differs from the approach adopted throughout sections \ref{Sec.CGSolutionPossible} and \ref{Sec.CGSolutionNotPossible} for describing the action of quantum channels. Since acausal and causal conditional states are connected via a partial transposition (Definition \ref{def:CausalState}), we point out that the use of acausal conditional states will not alter this conventional point of view of a coarse-graining diagram, given that we will still be considering the causal connection in a quantum dynamics and focusing solely, for convenience, on their Choi-isomorphic operators.

On this basis, we engage in some questions that arise from our proposed  Bayesian, state-dependent solution to the coarse-graining problem, addressing its limitations and exploring additional aspects of this Bayesian point of view of the coarse-graining problem through semidefinite programs (SDPs) implementations. Such implementations will anchor the numerical benchmarking, conducted in more detail in the next section, of our proposed Bayesian solution when one employs it in the coarse-graining scenarios discussed along Sec. \ref{Sec.CGSolutionNotPossible}.


Before delving into this convex optimization treatment, it is reasonable for the purposes of this work to first introduce the diamond norm \cite{kitaev97}, represented as $\vert \vert \cdot\vert \vert _{\diamond}$. The diamond norm is an operational tool to measure the distinguishability between two quantum channels, which has efficient SDP formulations \cite{Skrzypczyk23,BS10, WatrousSDPnorms}. Aligned with that, we directly embed our following analysis with its operational SDP version---borrowed from ref. \cite{Skrzypczyk23} and adapted to the notations of the conditional states formalism for characterizing quantum channels. That is, this leads us to employ the following SDP formulation of the diamond norm:
\begin{align}\label{eq:DiamondNormSDP}
    &\text{Given }\rho_{B\vert A} \simeq \Phi_{B \vert A}\\
    & \vert \vert  \Phi_{B \vert A}\vert \vert _{\diamond} := \text{Minimise} \quad \epsilon \notag \\
    \notag\\ 
    & \text{subject to} \quad \rho_{B\vert A} = Z_{AB} - X_{AB};\notag\\
    \notag\\
    & \quad \quad \quad \quad \quad Z_{AB} \geq 0, \quad X_{AB} \geq0;\notag \\ 
    \notag\\
    & \quad \quad \quad \quad \quad \epsilon \mathds{I}_A \geq \Tr_B(Z_{AB} + X_{AB}), \notag
\end{align}
where we will denote SDP auxiliary variables, such as $X_{AB},Z_{AB} \in\mathcal{L}(\mathcal{H}_{AB})$, with the same joint notation adopted in previous sections to refer to their respective Hilbert spaces. Moreover, based on the above formulation of the diamond norm, we define a set $\mathcal{B}^{\epsilon, \Phi}_{\diamond}$---referred to here as the \textit{diamond norm ball} of radius $\epsilon$ around a given channel $ \Phi_{B \vert A}$---consisting of quantum channels $ \mathcal{N}_{B \vert A}$ that are $\epsilon$-close to $ \Phi_{B \vert A}$ as
\begin{align}\label{eq:diamondNormBall}
\mathcal{B}^{\epsilon,\Phi_{B \vert A}}_{\diamond}
:=\;& \{ \mathcal{N}_{B \vert A}: \|\Phi_{B \vert A} - \mathcal{N}_{B \vert A}\|_{\diamond} \leq \epsilon \} \\
=\;&  \left\{ \mathcal{N}_{B \vert A}:
\begin{aligned}
& \rho_{B \vert A} - \sigma_{B \vert A} = Z_{AB} - X_{AB}, \\
& X_{AB} \geq 0,\quad Z_{AB} \geq 0, \\
& \epsilon\,\mathds{I}_A \geq \Tr_B(X_{AB} + Z_{AB}), \\
& \sigma_{B \vert A} \geq 0,\quad \Tr_B(\sigma_{B \vert A}) = \mathds{I}_A
\end{aligned}
\right\}.
\notag
\end{align}

Now, we are in a position to analyze the limitations of our proposed Bayesian solution $\Gamma^{\text{Petz}}_{D\vert C, \rho_A}$ for the (fully-quantum) coarse-graining problem. Since the $\Gamma^{\text{Petz}}_{D|C, \rho_A}$ depends intrinsically on an initial state $\rho_A$ for its construction, it restricts our task to finding an emergent dynamics that makes the diagram in Fig. \ref{fig:CGandCSF} commute only in a state-by-state manner, at least for $\rho_A$ itself. Therefore, in order to remove the $\rho_A$ dependence out of the commutativity relation while maintaining it in the  $\Gamma^{\text{Petz}}_{D|C, \rho_A}$ construction, we start to face this state-dependence limitation through the following crucial inquiry: Is there another quantum dynamics $\Gamma^t_{D\vert C}$, closest in the diamond norm to $\Gamma^{\text{Petz}}_{D\vert C, \rho_A}$, which makes the diagram of Fig.\ref{fig:CGandCSF} commute for any initial state? A direct and conventional SDP formulation of such a question assumes the following form:
\begin{align}\label{eq:AuxSDP2}
     &\text{Given} \quad \Gamma^{\text{Petz}}_{D\vert C, \rho_A},\, \mathcal{U}^t_{B\vert A}, \, \Lambda^{\text{CG}}_{C\vert A}(\Lambda^{\text{CG}}_{D\vert B})\\
    &  \text{Minimise} \quad \epsilon \notag \\
    \notag\\ 
    & \text{subject to}\quad \Gamma^t_{D\vert C} \circ \Lambda_{C\vert A}^{\text{CG}} = \Lambda_{D\vert B}^{\text{CG}} \circ \mathcal{U}^t_{B\vert A} ;\notag\\
    \notag\\
    & \quad \quad \quad \quad \quad \vert \vert  \Gamma^{\text{Petz}}_{D\vert C, \rho_A} - \Gamma^t_{D\vert C}\vert \vert _{\diamond} \leq \epsilon;\notag \\ 
    \notag\\
    & \quad \quad \quad \quad \quad \Gamma^t_{D\vert C} \quad \text{is CPTP}.  \notag
\end{align}
Alternatively, we can rewrite the SDP above entirely within the conditional states formalism. This is achieved by replacing the diamond norm constraint by its SDP characterization---expressed in accordance with the notations of the SDP in \eqref{eq:DiamondNormSDP}---and by representing all quantum channels and the commutativity relation (eq.   \eqref{eq:CommutRelationBP}) through their respective conditional states. The resulting, equivalent SDP reads as  
\begin{tcolorbox}
\subsection*{\textbf{Searching for state-independent solutions close to $\Gamma^{\text{Petz}}_{D|C, \rho_A}$}}
    \vspace{-1.0cm}
\begin{align}\label{eq:SDP2}
     &\text{Given} \quad \rho^{\text{Petz}}_{D\vert C, \rho_A}, \rho_{B\vert A}, \rho_{C\vert A}(\rho_{D\vert B})\\
    &  \text{Minimise} \quad \epsilon \notag \\
    \notag\\ 
    & \text{subject to}\quad\Tr_{C} [(\mathds{I}_A \otimes \rho_{D\vert C})(\rho_{C\vert A}^{T_C} \otimes \mathds{I}_D)] = \notag\\ 
    & \quad \quad \quad \quad \quad\Tr_B[(\mathds{I}_A \otimes \rho_{D\vert B})(\rho_{B\vert A}^{T_B} \otimes \mathds{I}_D)];\notag\\
    \notag\\
    & \quad \quad \quad \quad \quad \rho^{\text{Petz}}_{D\vert C} - \rho_{D\vert C} = Z_{CD} - X_{CD};\notag \\ 
    \notag\\
    & \quad \quad \quad \quad \quad Z_{CD} \geq 0, \quad X_{CD} \geq 0;\notag \\ 
    \notag \\
     & \quad \quad \quad \quad \quad \epsilon \mathds{I}_C \geq \Tr_D(Z_{CD} + X_{CD}); \notag\\
     \notag \\
      & \quad \quad \quad \quad \quad  \rho_{D\vert C} \geq 0, \quad \Tr_D(\rho_{D\vert C}) = \mathds{I}_C \notag,
\end{align}
\end{tcolorbox} \noindent
thereby emphasizing the usefulness of the conditional states formalism in approaching the coarse-graining problem. We claim that, through the above SDP, we can start to understand how distant the proposed state-dependent solution is from another effective solution---if the latter actually exists.

To find an emergent dynamics compatible to a given microscopic dynamics and coarse-grained description, regardless of whether it is close to $\Gamma^{\text{Petz}}_{D\vert C, \rho_A}$, we can reformulate the above problem. It suffices to remove the diamond distance constraints and raise it to a feasibility SDP:
\begin{tcolorbox}
\subsection*{\textbf{Searching for state-independent solutions}}
    \vspace{-1cm}
\begin{align}\label{eq:SDP3}
     &\text{Given} \quad \rho_{B\vert A}, \rho_{C\vert A}(\rho_{D\vert B})\\
    &  \text{Find} \quad \rho_{D\vert C} \notag \\
    \notag\\ 
    & \text{subject to}\quad\Tr_{C} [(\mathds{I}_A \otimes \rho_{D\vert C})(\rho_{C\vert A}^{T_C} \otimes \mathds{I}_D)] = \notag\\ 
    & \quad \quad \quad \quad \quad\Tr_B[(\mathds{I}_A \otimes \rho_{D\vert B})(\rho_{B\vert A}^{T_B} \otimes \mathds{I}_D)];\notag\\
    \notag\\
    & \quad \quad \quad \quad \quad  \rho_{D\vert C} \geq 0, \quad \Tr_D(\rho_{D\vert C}) = \mathds{I}_C \notag.
\end{align}
\end{tcolorbox}
The SDPs presented above concern the existence of an emergent dynamics, related or not to our proposed solution $\Gamma^{\text{Petz}}_{D|C, \rho_A}$. The questions raised by the SDPs \eqref{eq:SDP2} and \eqref{eq:SDP3}---and the answers thereon---provide an initial direction by which we start to understand: (i) where our proposed solution is located in relation to a definitive solution, that is, an emergent dynamics which commutes the diagram for any initial state, and (ii) in what scenarios there exists such definitive solution, circumventing in that way the state-by-state limitation of our ones. From now on, we put our proposed Bayesian solution aside and concentrate our efforts to investigate additional aspects of the coarse-graining problem in this Bayesian inference picture here sketched.

To do so, we need first to introduce the concept of \textit{robustness}. Historically, the \textit{robustness measure}, proposed in ref. \cite{GR1999}, is a special case of an entanglement magnitude. Roughly speaking, it quantifies how much noise can be added to an entangled state in order to erase all entanglement contained in that state. Despite such a concept being first concerned about the quantum correlation properties of states, there are several other applications and variants of such a quantity proposed in the literature \cite{BRM21,JWW21,FWTB2020,YLLL22,HN03}. Hence, inspired by such studies and in line with the concept of the (non-)compatibility between the quantum channels employed in a coarse-graining scenario (Definition \ref{def:CompatibilityCG}), we propose a new robustness measure that quantifies how much noise can be added to a unitary evolution---compatible with a given coarse-graining map---while still remaining compatible with the corresponding coarse-grained description. In light of the above,  

\begin{defi}[CG-compatibility robustness]
 Let $\mathcal{U}^t_{B|A}$ be a unitary dynamics compatible with a given coarse-graining map $\Lambda^{\text{CG}}_{C|A}(\Lambda^{\text{CG}}_{D|B})$, that is, there exists a CPTP emergent dynamics $\Gamma^t_{D|C}$ such that $\Gamma^t_{D|C} \circ \Lambda^{\text{CG}}_{C|A} = \Lambda^{\text{CG}}_{D|B}\circ \mathcal{U}^t_{B|A}$. We define the \textit{CG-compatibility robustness} of $\mathcal{U}^t_{B|A}$ as 
\begin{equation}\label{eq:DefiRobustness}
r_c(\mathcal{U}^t_{B\vert A}) : = 1 - \Tilde{r}_c(\mathcal{U}_{B\vert A}),
\end{equation}
where the quantity $\Tilde{r}_c(\mathcal{U}_{B|A})$ is given explicitly by
\begin{align}\label{eq:DefAuxRobustness}
    & \Tilde{r}_c(\mathcal{U}^t_{B\vert A}) := \\ 
    &\text{min}\{\gamma \enspace \vert \enspace \gamma \mathcal{U}^t_{B\vert A} +  (1 - \gamma)\Phi^t_{B\vert A} =: \mathcal{N}^t_{B\vert A}\in \mathcal{C}(A,B),\notag\\
    & \Phi^t_{B\vert A} \in \mathcal{C}(A,B), \exists \Omega^t_{D\vert C} \in \mathcal{C}(C,D): \notag\\ 
    & \Omega^t_{D\vert C} \circ \Lambda^{\text{CG}}_{C\vert A} = \Lambda^{\text{CG}}_{D\vert B}\circ \mathcal{N}^t_{B\vert A}\} \notag.
\end{align}
\end{defi}

In the definition above, the notation $\mathcal{C}(A,B)$ represents the ser of all CPTP maps from $\mathcal{L}(H_A)$ to $\mathcal{L}(H_B)$. 

Intuitively, $r_c(\mathcal{U}^t_{B\vert A})$ quantifies how \textit{robust} the microscopic dynamics $\mathcal{U}^t_{B\vert A}$ is---for a given coarse-grained description---in allowing the existence of another emergent dynamics $\Omega^t_{D\vert C}$, even in the presence of any microscopic noise. In other words, the value assumed by $r_c(\mathcal{U}^t_{B|A})$ refers to the maximal addition of noise (dictated by $\Phi^t_{B|A}$) to a compatible unitary dynamics $\mathcal{U}^t_{B|A}$, with a fixed coarse-grained description $\Lambda^{\text{CG}}_{C|A}(\Lambda^{\text{CG}}_{D|B})$, before spoiling the compatibility. In this sense, the smaller $\Tilde{r}_c(\mathcal{U}^t_{B|A})$ is, the more $\mathcal{U}_{B|A}^t$ can be perturbed by the noise $\Phi^t_{B|A}$ while remaining compatible with the coarse-grained description $\Lambda^{\text{CG}}_{C|A}(\Lambda^{\text{CG}}_{D|B})$, and thus the greater its $CG$--compatibility robustness $r_c(\mathcal{U}_{B|A})$ is.
 
 We can look at the respective conditional state associated with each aforementioned quantum channel in order to propose the following SDP characterization to the auxiliary quantity $\Tilde{r}_c(\mathcal{U}_{B|A})$:
\begin{align}\label{eq:SDP4}
     &\text{Given} \quad \rho_{B\vert A}, \phi_{B\vert A}, \rho_{C\vert A}(\rho_{D\vert B}) \\
    & \Tilde{r}_c(\mathcal{U}^t_{B\vert A}) =  \text{Minimise} \quad \gamma \notag \\
    \notag\\ 
    & \text{Subject to}\quad \sigma_{B\vert A} = \gamma \rho_{B\vert A} + (1 -\gamma)\phi_{B\vert A}; \notag\\
    \notag \\
    &\quad \quad \quad \quad \quad\Tr_{C} [(\mathds{I}_A \otimes \omega_{D\vert C})(\rho_{C\vert A}^{T_C} \otimes \mathds{I}_D)] = \notag\\ 
    & \quad \quad \quad \quad \quad\Tr_B[(\mathds{I}_A \otimes \rho_{D\vert B})(\sigma_{B\vert A}^{T_B} \otimes \mathds{I}_D)];
    \notag\\
    \notag\\
    & \quad \quad \quad \quad \quad \sigma_{B\vert A} \geq 0, \quad \Tr_B(\sigma_{B\vert A}) = \mathds{I}_A ;
    \notag \\ 
    \notag\\
    & \quad \quad \quad \quad \quad \omega_{D\vert C} \geq 0, \quad \Tr_D(\omega_{D\vert C}) = \mathds{I}_C, \notag 
    \end{align}
and, as consequence, efficiently compute $r_c(\mathcal{U}_{B|A}^t)$ from it.

Note that, since we are interested in cases where the noisy, resultant microscopic dynamics $\mathcal{N}^t_{B\vert A}$ is CPTP, the very same constraints in the SDP which guarantee that its Choi-isomorphic state $\sigma_{B\vert A}$ is a valid acausal conditional state imply naturally that $r_c(\mathcal{U}^t_{B\vert A}) \in[0,1]$. Moreover, when $r_{c}(\mathcal{U}^t_{B\vert A})=0$, we say that the microscopic dynamics $\mathcal{U}^t_{B\vert A}$ is not robust---relative to a given compatible coarse-grained description---against any noise $\Phi^t_{B\vert A}$. 

Inspired by the discussions about the CG-compatibility robustness above, we take now a subtly different path. There is one last question that we want to address in this contribution: given a coarse-grained description, dictated by a coarse-graining map $\Lambda^{\text{CG}}_{C|A}(\Lambda^{\text{CG}}_{D|B})$, and an initial, non-compatible unitary dynamics $\mathcal{U}^t_{B|A}$, does there exist any microscopic dynamics $\Psi^t_{B|A} \in \mathcal{C}(A,B)$ such that, for some $\gamma \in [0,1]$, the noised, resultant dynamics
\begin{equation}\label{eq:ResultantMicroDynamics}
    \mathcal{J}^t_{B|A} := \gamma \mathcal{U}^t_{B|A} + (1-\gamma) \Psi^t_{B|A}
\end{equation}
is compatible with an emergent dynamics $\Theta^t_{D|C} \in \mathcal{C}(D,C)$, satisfying $\Theta^t_{D|C} \circ \Lambda^{\text{CG}}_{C|A} = \Lambda^{\text{CG}}_{D|B} \circ \mathcal{J}^t_{B|A}$? Put another way, is it possible to render a non-compatible unitary dynamics compatible? Such a question, when one lets $\rho_{B\vert A}$ and  $\rho_{C\vert A}(\rho_{D\vert B})$ be the conditional states associated with $\mathcal{U}^t_{B\vert A}$ and $\Lambda^{\text{CG}}_{C\vert A}(\Lambda^{\text{CG}}_{D\vert B})$, respectively, could be reformulated as a feasibility problem, as shown in the SDP below.
\begin{tcolorbox}
\subsection*{\textbf{Robustness inspired}}
    \vspace{-1cm}
\begin{align}\label{eq:SDP5}
     &\text{Given} \quad \rho_{B\vert A},\rho_{C\vert A}(\rho_{D\vert B}), \gamma \\
    & \text{Find} \quad \psi_{B\vert A} \notag \\
    \notag\\ 
    & \text{subject to}\quad \sigma_{B\vert A} = \gamma \rho_{B\vert A} + (1- \gamma)\psi_{B\vert A}; \notag\\
    \notag \\
    &\quad \quad \quad \quad \quad\Tr_{C} [(\mathds{I}_A \otimes \theta_{D\vert C})(\rho_{C\vert A}^{T_C} \otimes \mathds{I}_D)] = \notag\\ 
    & \quad \quad \quad \quad \quad\Tr_B[(\mathds{I}_A \otimes \rho_{D\vert B})(\sigma_{B\vert A}^{T_B} \otimes \mathds{I}_D)];
    \notag\\
    \notag\\
    & \quad \quad \quad \quad \quad \psi_{B\vert A} \geq 0, \quad \Tr_B(\psi_{B\vert A}) = \mathds{I}_A ;
    \notag \\ 
    \notag\\
    & \quad \quad \quad \quad \quad \theta_{D\vert C} \geq 0, \quad \Tr_D(\theta_{D\vert C}) = \mathds{I}_C. \notag 
\end{align}
\end{tcolorbox}
In other words, the SDP above provides a yes/no answer to the following question: if the initial microscopic dynamics of a quantum system is non-compatible with a given coarse-grained description, does there exist any CPTP map such that their convex combination becomes compatible?

This question seems to be deeply in the core of the coarse-graining problem, because even if we are unable to describe the underlying microscopic dynamics of a quantum system through a less complex, effective dynamics---promoted by a given coarse-grained description---by solving the optimization problem stated in \eqref{eq:SDP5}, one is able to check whether there exists some quantum dynamics whose combination with the initial ones makes it possible to operationally obtain such an effective description.

To numerically address the SDPs above, we implemented these optimization programs using the Python language and the CVXPY library \cite{diamond2016cvxpy}, employing MOSEK \cite{mosek} as the solver. 
With the aim of encouraging further research in this direction and promoting reproducibility of our numerical findings, the data and implementation codes used here to generate our numerical results are available on GitHub.\footnote{\url{https://github.com/thalesbsfr/CG-SDP-numerical-implementations}} 

\begin{figure*}[ht]
    \centering
    \subfigure[]{\label{subfig:ColorMapA}\includegraphics[width=0.24\textwidth]{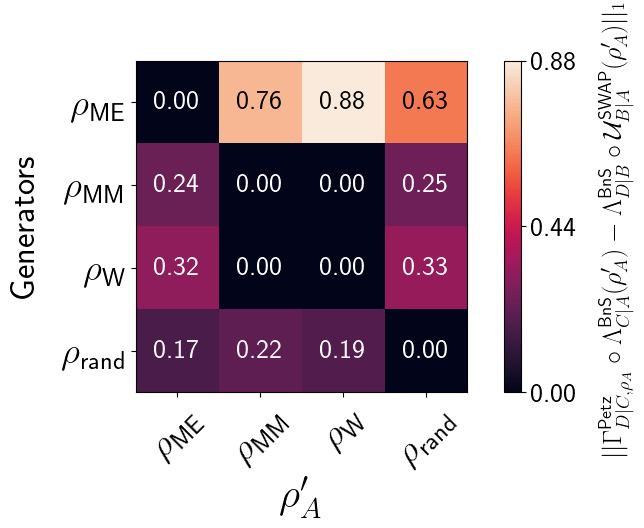}} 
    \subfigure[]{\label{subfig:ColorMapB}\includegraphics[width=0.24\textwidth]{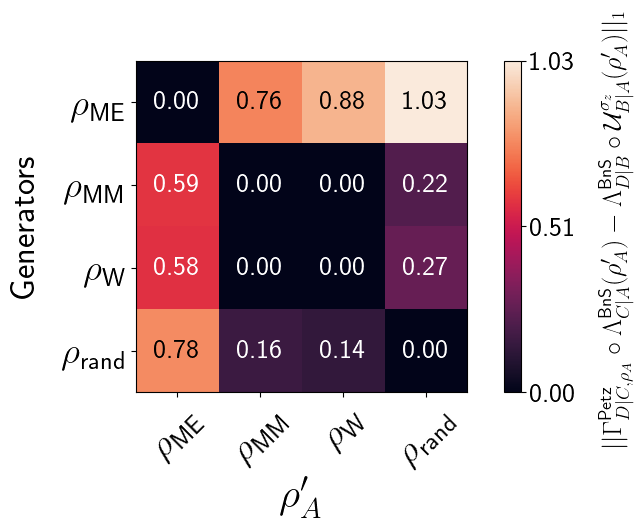}} 
    \subfigure[]{\label{subfig:ColorMapC}\includegraphics[width=0.24\textwidth]{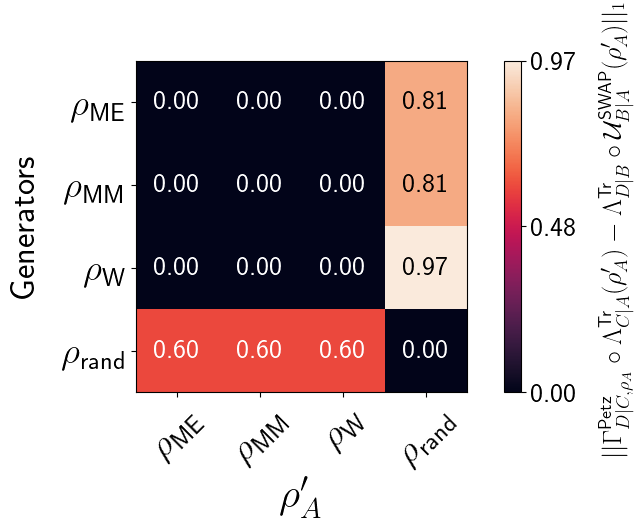}}
    \subfigure[]{\label{subfig:ColorMapD}\includegraphics[width=0.24\textwidth]{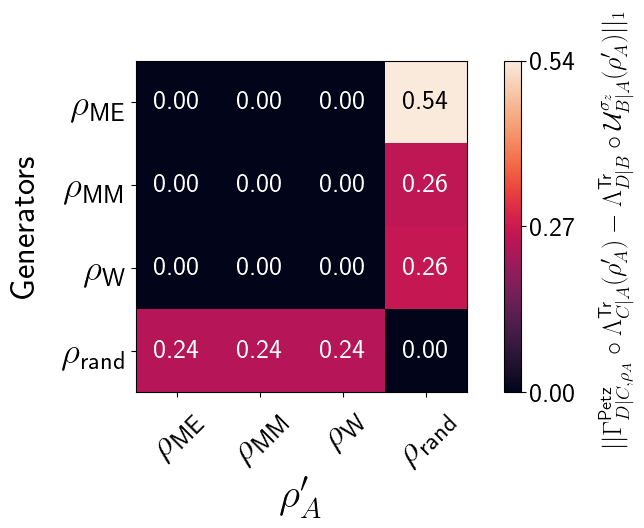}}
    \caption{Performance of the solution $\Gamma^{\text{Petz}}_{D|C, \rho_A}$ in four different coarse-graining scenarios for four different generators.
    The states on the rows on the left represents the generators, while those on the bottom of the columns represents the $\rho'_{A}$ states used in the norm of the commutativity relation in \eqref{eq:traceNormCommuting}. The performance of the solution $\Gamma^{\text{Petz}}_{D|C, \rho_A}$ in commuting the diagram, quantified by the norm in eq.~\eqref{eq:traceNormCommuting}, is represented through a color-scale mapping. Image (a) shows the results for our scenario 1, while images (b), (c) and (d), for the scenarios 2, 3 and 4, respectively. One can see that, in general, even in the worst scenario (based on the mean value of all results) in this picture for the commutativity using the proposed $\Gamma^{\text{Petz}}_{D|C, \rho_A}$ solution, scenario 2, image (b), there are states different from the generator in which the diagram commutes.}
    \label{fig:ColorMaps}
\end{figure*}

\section{Numerical results} \label{sec:NumericalResults}

\subsection{Benchmarking diagram commutativity} 
\label{subsec:DiagramCommutativity}

As a starting point, it was shown along this contribution that our proposed solution $\Gamma^{\text{Petz}}_{D\vert C, \rho_A}$ (for the fully-quantum coarse-graining scenario) has an intrinsic dependence on an initial state $\rho_A$. Moreover, as also discussed, for a given microscopic dynamics and coarse-grained description, it satisfies the commutativity relation at least for the same $\rho_A$ adopted in its own construction---its \textit{generator}. Nevertheless, fixing an initial generator $\rho_A$ in the $\Gamma^{\text{Petz}}_{D|C, \rho_A}$ construction one might pragmatically ask: how robust is the $\Gamma^{\text{Petz}}_{D|C, \rho_A}$ solution in making the diagram of Fig. \ref{fig:CGandCSF} commute for other states in the region A, not necessarily equal to its generator? In other words, letting $\rho'_A \in \mathcal{D}(\mathcal{H}_A)$ be any state of the initial region $A$, how close, in the trace norm, is 
\begin{equation} \label{eq:traceNormCommuting}
    \vert \vert \Gamma^{\text{Petz}}_{D\vert C, \rho_A} \circ \Lambda^{\text{CG}}_{C\vert A}(\rho'_A) - \Lambda^{\text{CG}}_{D\vert B} \circ \mathcal{U}^t_{B\vert A}(\rho'_A)\vert \vert _1  
\end{equation}
to zero, for a given configuration of unitary dynamics $\mathcal{U}^t_{B\vert A}$ and coarse-graining map $\Lambda^{\text{CG}}_{C\vert A}(\Lambda^{\text{CG}}_{D\vert B})$? 

Initially, we study four generators: (1) $\rho_{\text{ME}}:=  \ket{\Phi^+}\bra{\Phi^+} \in \mathcal{D} (\mathds{C}^2 \otimes \mathds{C}^2)$, representing the normalized maximally entangled state; (2) $\rho_{\text{MM}}: = \frac{1}{4}\mathds{I}_{4 \times 4}$, the maximally mixed state; (3) a state $\rho_{\text{rand}}$ generated randomly using a uniform random sampling method of the NumPy library \cite{NumPy2020}; and (4) the two-qubit Werner state
\begin{equation} \label{eq:rhoWerner}
    \rho_{\text{W}} := \lambda \ket{\Psi^-}\bra{\Psi^-} + \frac{1-\lambda}{4} \mathds{I}_{4\times 4},
\end{equation}
for $\ket{\Psi^-}:= \frac{1}{\sqrt{2}}(\ket{01} - \ket{10})$ and $\lambda = 1/3$. Figures \ref{subfig:ColorMapA}, \ref{subfig:ColorMapB}, \ref{subfig:ColorMapC} and \ref{subfig:ColorMapD} contain our initial results. Because scenarios 2 and 4 are time dependent, we decided to show the results initially for $t=\SI{1.0}{\second}$---other time steps will be investigated later in this work. Moreover, in order to simplify the calculations, we have set $J=\SI{1.0}{\second^{-1}}$ in the unitary $U_{\sigma_z}$ (eq.   \eqref{eq:U_z} of the $z$-interaction channel throughout the following numerical evaluations).
\begin{figure*}[t!]
    \centering
    \subfigure[]{\label{subfig:HistogramsA}\includegraphics[width=0.24\textwidth]{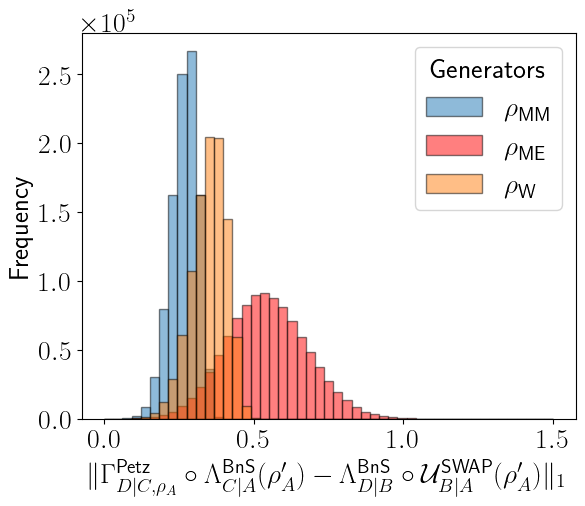}} 
    \subfigure[]{\label{subfig:HistogramsB}\includegraphics[width=0.24\textwidth]{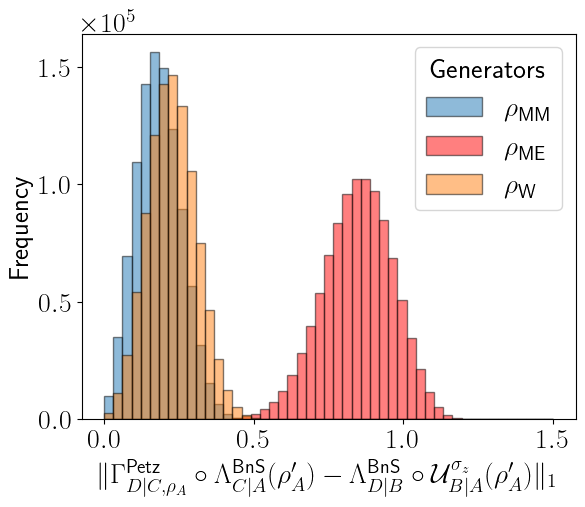}} 
    \subfigure[]{\label{subfig:HistogramsC}\includegraphics[width=0.24\textwidth]{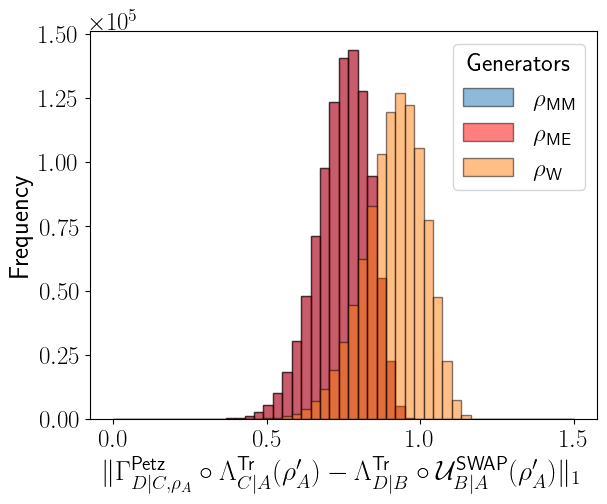}}
    \subfigure[]{\label{subfig:HistogramsD}\includegraphics[width=0.24\textwidth]{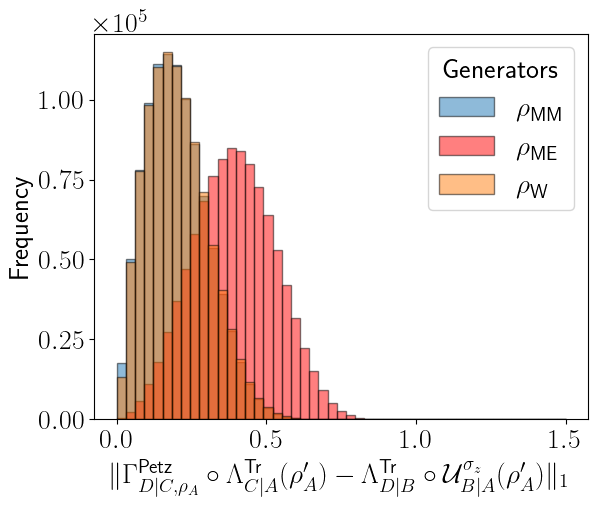}}
    \caption{Histograms showing the commutativity performance of $\Gamma^{\text{Petz}}_{D|C, \rho_A}$ (eq.   \eqref{eq:traceNormCommuting}) for a sample of $10^6$ random states in the four scenarios. We can see that, for scenarios 1 and 2, shown in images (a) and (b), respectively, the generator in which the $\Gamma^{\text{Petz}}_{D|C, \rho_A}$ achieves the worst results was $\rho_{\text{ME}}$, while regards the others two, their respective histograms presents a similar behavior, with $\rho_{\text{MM}}$ producing better results in general. In the last histograms---shown in images (c) and (d), corresponding to scenarios 3 and 4, respectively---a slightly different behavior emerges in comparison to the previously discussed histograms.  In scenario 3, since the $\Gamma_{D \vert C , \rho_A}^{\text{Petz}}$ generated by $\rho_{\text{ME}}$ and by $\rho_{\text{MM}}$ are identical, both $\rho_{\text{ME}}$ and $\rho_{\text{MM}}$ yield similar results, whereas $\rho_{\text{W}}$ performs worst in general. On the other hand, in scenario 4, we can see that $\rho_{\text{MM}}$ and $\rho_{\text{W}}$ achieves almost equal, best results and the $\rho_{\text{ME}}$ the worst one. Overall, one can notice that, in scenario 4, all distributions reach results closest to the zero in comparison to the others, while the scenario 3, on average, the most distant. Also, on average, one can see that, the generator in which produces the worst results was $\rho_{\text{ME}}$.}
\label{fig:Histograms}
\end{figure*}

From the results shown in Fig. \ref{fig:ColorMaps}, we observe that, in all four scenarios, deploying the solution $\Gamma^{\text{Petz}}_{D|C, \rho_A}$ as the emergent dynamics makes it possible to satisfy the commutativity relation for states different from its generators, as is the case, in all scenarios, when $\rho_{\text{W}}$ is the generator and $\rho_{\text{MM}}$ goes into the commutativity relation (or vice versa). In particular, for the scenario 4, Fig. \ref{subfig:ColorMapD}, it also suggests that, for a given generator, the $\Gamma^{\text{Petz}}_{D\vert C, \rho_A}$ makes the diagram commute not only for the generator itself, but additionally, there exists a set of states which satisfy the desired commutativity relation, achieving zero in \eqref{eq:traceNormCommuting}. This fact, concerning a pragmatic point of view of the coarse-graining problem, emphasizes numerically that, even in some specific scenarios where there does not exist an emergent dynamics that commutes the diagram for all initial states, i.e., which satisfies eq.   \eqref{eq:CommutationRelation}, or its respective causal conditional state in \eqref{eq:CommutRelationBP}, our proposed Bayesian solution $\Gamma^{\text{Petz}}_{D\vert C, \rho_A}$ can be employed for this task at least for a set of initial states. 

For sustaining this last numerical statement, we also engage with several evaluations of the expression \eqref{eq:traceNormCommuting}  through a sample of $10^6$ different random states as $\rho'_A$.
Since $\rho_{\text{rand}}$, as shown in Fig. \ref{fig:ColorMaps}, is the generator for which $\Gamma^{\text{Petz}}_{D|C, \rho_A}$ yields the worst performance in commuting the diagram for states different from itself, we restrict our subsequent calculations of the expression \eqref{eq:traceNormCommuting} to the other three generators. The results are summarized in the histograms of Fig.~\ref{fig:Histograms}. 

Notably, from the above histograms one can see that there are some scenarios, such as scenario 2 and 4 shown in Fig.~\ref{subfig:HistogramsB} and Fig.~\ref{subfig:HistogramsD}, respectively, in which our proposed solution $\Gamma^{\text{Petz}}_{D\vert C, \rho_A}$, with $\rho_{\text{MM}}$ and $\rho_{\text{W}}$ as generators, achieves results very close to zero for a considerable number of random states in the first bins of such histograms. Highlighting, in this way, our initial numerical statement that, from a pragmatic point of view, our solution can be employed to commute the diagram in such scenarios, not only for the generator itself, but also for a set of distinct states of the initial region. A more detailed view of the best results in the bin closest to zero in the histograms of scenarios 2 and 4 (Fig.~\ref{subfig:HistogramsB} and Fig.~\ref{subfig:HistogramsD}, respectively) is shown in Fig.~\ref{fig:Boxplots}.
\begin{figure}[ht]
    \centering
    \subfigure[]{\label{subfig:BP_setup2}\includegraphics[width=0.35\textwidth]{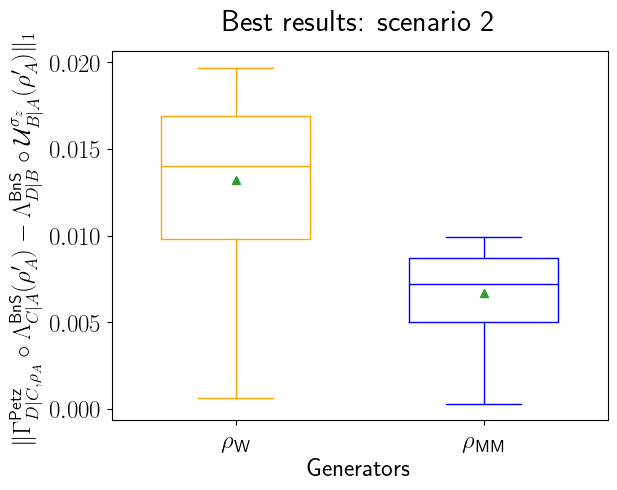}} 
    \subfigure[]{\label{subfig:BP_setup4}\includegraphics[width=0.35\textwidth]{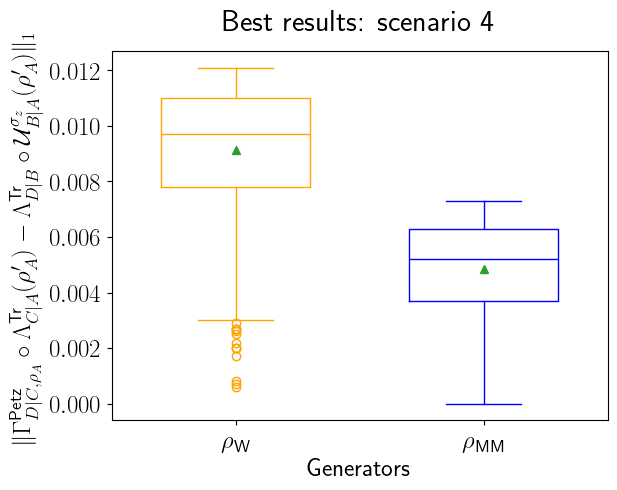}} 
    \caption{Boxplots of the best $10^3$ results obtained in scenarios 2 and 4. Each box represents the interquartile range, the central line marks the median, the green triangles indicates the mean values, and outliers as individuals circles. One can notice that, in general, the $\rho_{\text{MM}}$ as generator produces better results in comparison to the $\rho_{\text{W}}$. Moreover, for scenario 4, blue boxplot in image (b), it achieves the minimum result (zero) for $28$ different random states within these $10^3$ analyzed.}
\label{fig:Boxplots}
\end{figure}

\begin{figure*}[ht]
    \centering
    \subfigure[]{\label{subfig:Time_setup2}\includegraphics[width=0.45\textwidth]{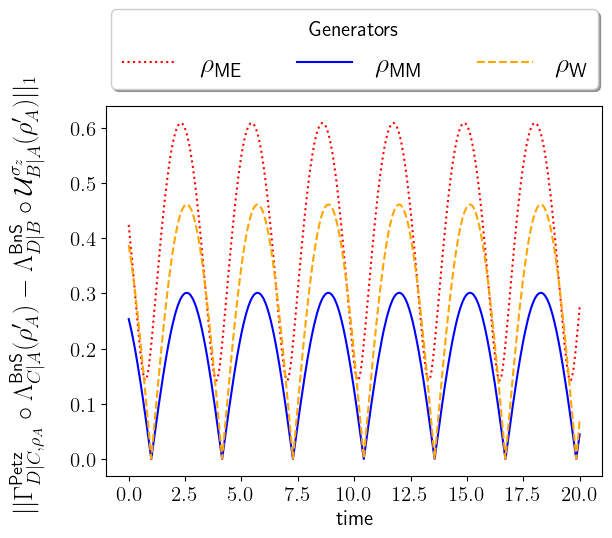}} 
    \subfigure[]{\label{subfig:Time_setup4}\includegraphics[width=0.45\textwidth]{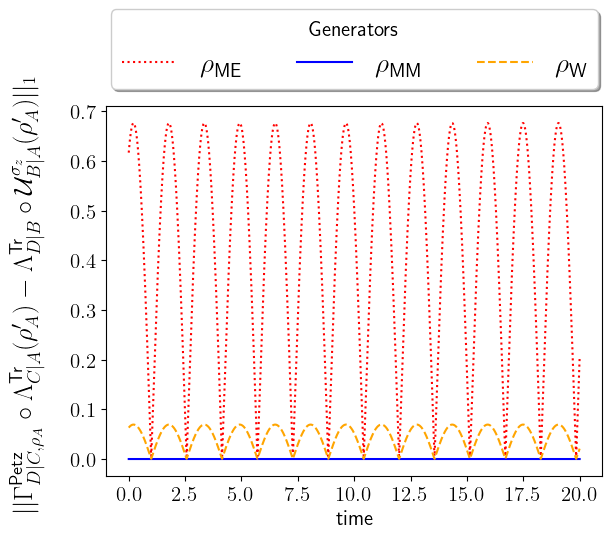}} 
    \caption{Time varying effects on the commutativity performance of $\Gamma^{\text{Petz}}_{D|C, \rho_A}$. Image (a) shows the results for scenario 2 and image (b) for the scenario 4. We can see, for scenario 2, changing the time $t$ of the microscopic dynamics $\mathcal{U}^{\sigma_z}_{B\vert A}$ can improve the results obtained for $\rho_{\text{ME}}$ generator; however, it neither approaches nor reaches zero as equal the results of the other generators. In general, it sustains the same behavior shown in Fig.~\ref{subfig:HistogramsB}, with the $\rho_{\text{MM}}$ achieving the best results, follows by the $\rho_{\text{W}}$  and, despite in some time-steps its results has improved, $\rho_{\text{ME}}$ remains the worst generator for this scenario. Regarding the scenario 4, the differences between the performance of each generator become highly evident through this time varying evaluation. $\rho_{\text{ME}}$, which achieves the worst results in Fig. \ref{subfig:HistogramsD}, here, for some timestep, it achieves results close to zero, in contrast to the results observed on image (a). Moreover, notably, the results for $\rho_{\text{W}}$ for all timestep, even in some specific ones it gets a little worse, they revolved around values very close to zero. For the $\rho_{\text{MM}}$ as generator, we can see a very peculiar behavior in comparison to the others. The selected state, which commutes the diagram of scenario 4 for $t = \SI{1.0}{\second}$ in the microscopic dynamics $\mathcal{U}^{\sigma}_{B\vert A}$, remains commuting the diagram, i.e., achieving zero in the expression \eqref{eq:traceNormCommuting} for all the timesteps in this analysis. Shining light to the usefulness of the $\rho_{\text{MM}}$ as $\Gamma^{\text{Petz}}_{D\vert C, \rho_A}$ generator in this coarse-graining scenario.}
\label{fig:Timeplots}
\end{figure*}
As mentioned previously, the results regarding scenarios 2 and 4 of Figures \ref{fig:ColorMaps}, \ref{fig:Histograms} and \ref{fig:Boxplots} were generated by setting $t = \SI{1.0}{\second}$ in the unitary dynamics $\mathcal{U}^{\sigma_z}_{B\vert A}$. Since such dynamics, as discussed around eq.   \eqref{eq:InteractionChannel}, can create or annihilate correlations between the two parts of the initial system depending on the value of the parameter $t$, it is reasonable to also analyze numerically what this fact could imply to diagram commutativity involving our interested scenarios. Taking the random state which generates the best result in Fig. \ref{subfig:HistogramsB} and \ref{subfig:HistogramsD} for each generator, one can see, in Fig. \ref{fig:Timeplots}, how varying the time parameter $t$ of the microscopic dynamics directly affects the commutativity of the diagram of scenarios 2 and 4.

From the results shown in Fig. \ref{fig:Timeplots}, one can see that, for scenario 2 (Fig. \ref{subfig:Time_setup2}), the time dependence of the microscopic dynamics $\mathcal{U}^{\sigma_z}_{B\vert A}$ does really interfere in the commutativity performance of all generators. In particular, it is worth mentioning that, in comparison to the results depicted in Fig. \ref{subfig:HistogramsB} for $t = \SI{1.0}{\second}$, there are some values of $t$ where the results produced with all generators were improved. Furthermore, regarding the results of scenario 4 (Fig. \ref{subfig:Time_setup4}), we would like to highlight that, performing the same analysis with the other $28$ random states in which the diagram of scenario 4 commutes when one uses the $\rho_{\text{MM}}$ as the $\Gamma^{\text{Petz}}_{D\vert C, \rho_A}$ generator and sets $t = \SI{1.0}{\second}$ in $\mathcal{U}^{\sigma_z}_{B\vert A}$, it also presents the same behavior and result observed in the blue-continuous line of Fig. \ref{subfig:Time_setup4}. Reinforcing, in that way, the usefulness of the $\Gamma^{\text{Petz}}_{D\vert C, \rho_A}$ generated by $\rho_{\text{MM}}$ in commuting the diagram not only for the generator itself, but also for a set of random states, regardless of the value of $t$ in the microscopic dynamics. We also notice that all those $28$ states satisfy the condition \eqref{t23=t13=0}, corroborating our analytical analysis. 

To conclude this numerical investigation of the diagram commutativity for our proposed $\Gamma^{\text{Petz}}_{D\vert C,\rho_A}$, we engage in a last but not least inquiry within this paradigm. The $\rho_{\text{W}}$ generator here employed, in the way it was defined in eq.   \eqref{eq:rhoWerner}, is a separable state for $\lambda \in [-1/3,1/3]$ and entangled for $\lambda \in(1/3,1]$ \cite{Werner1989}. Given that, in its very definition, we have stated $ \lambda = 1/3$ and conducted all of our above numerical evaluation with it being a separable state. One might ask what changing in the value of the parameter $\lambda$ could imply for the performance of $\rho_{\text{W}}$ as a generator. Taking each random state which produces the best result in each of the four scenarios\footnote{As a remark, for generated this analysis with emphasis on the effect of changing the value of the parameter $\lambda$ in the $\rho_{\text{W}}$ generator, we consider $t = \SI{1.0}{\second}$ on the microscopic dynamics $\mathcal{U}^{\sigma_z}_{B\vert A}$  of scenarios 2 and 4.} depicted in Fig. \ref{fig:Histograms} with $\rho_{\text{W}}$ as generator, an enlightenment upon such a question can be seen in Fig. \ref{fig:werner_graph}. 

\begin{figure}[ht]
    \centering
    \includegraphics[width=0.85\linewidth]{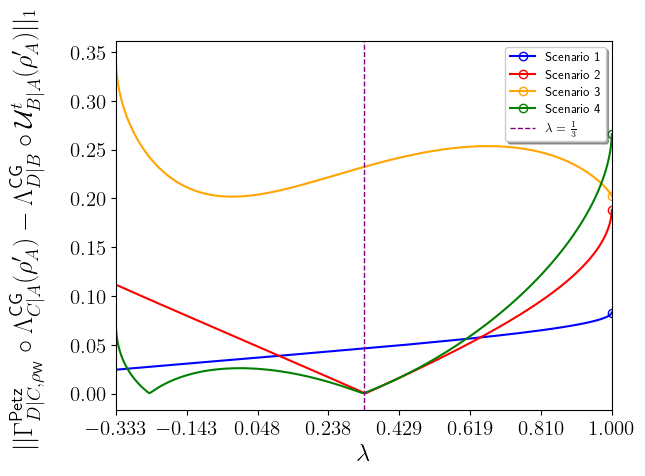}
    \caption{Performance of $\Gamma^{\text{Petz}}_{D|C, \rho_{\text{W}}}$ as a function of the parameter $\lambda$ in $\rho_{\text{W}}$. For each scenario, the random state $\rho'_A$ that goes into \eqref{eq:traceNormCommuting} was the one that achieves the best results in Fig. \ref{fig:Histograms} histograms of $\rho_{\text{W}}$ generator for $\lambda = 1/3$.  One can see that, with this vertical purple dashed line indicating where $\lambda =1/3$, how changing this parameter affects the commutativity performance of $\rho_{\text{W}}$ in all four scenarios. In particular, regarding the results of scenario 2 and scenario 4, one can see, notably, a inflection point and a worsening in the results displayed in their respective line graphs, right at $\lambda = 1/3$. Since for $\lambda > 1/3$ the $\rho_{\text{W}}$ becomes entangled and the worst result for these scenarios occurred when $\lambda \approx 1$\footnote{Given that, for $\lambda=1$, $\Lambda_{\text{BnS}}(\rho_{\text{W}})$ is not invertible and, because of this, we cannot construct the $\Gamma^{\text{Petz}}_{D\vert C,\rho_A}$, the open circles 
    highlight that we only consider $\lambda$ values close to $1$ at the end of $x-$axis.} ($\rho_{\text{W}}$ approximately maximally entangled), drawing a parallel with the results of the maximally entangled generator $\rho_{\text{ME}}$ shown in Fig. \ref{subfig:HistogramsB} and \ref{subfig:HistogramsD}, this might be an evidence that: the more entangled the $\Gamma^{\text{Petz}}_{D\vert C, \rho_A}$ generator is, the worse is the commutativity performance of $\Gamma^{\text{Petz}}_{D\vert C, \rho_A}$ in such scenarios.} 
    \label{fig:werner_graph}
\end{figure}

\subsection{SDP implementations} \label{SubSec:NumericalSDPs}
 After establishing the performance evaluation of our proposed Bayesian solution, we now turn our attention to some additional numerical inquiries raised by the semidefinite programs introduced throughout Sec. \ref{Sec.SDPAnalysis}. Although our SDPs were designed to work in full generality, to validate them against concrete benchmarks. We focus on the analysis of bipartite systems composed of two-qubits. In this way, we can determine interesting properties without having to care about the quirks of multipartite, high-dimensional quantum scenarios.

It was shown by the results displayed on Fig. \ref{fig:Histograms} and along the discussions around Fig. \ref{fig:Boxplots} that, for a given coarse-graining scenario, the choice of the $\Gamma^{\text{Petz}}_{D\vert C, \rho_A}$ generators interferes directly in its commutativity performance. Besides this, we can claim from the aforementioned result that, for scenario 4, within a sample of $10^6$ random states, the $\Gamma^{\text{Petz}}_{D\vert C, \rho_A}$ generated by $\rho_{\text{MM}}$ can be employed as a candidate of emergent dynamics for only a set of $28$ distinct states plus the generator itself. As said before, all of them satisfy the condition in eq. \eqref{t23=t13=0} In order to circumvent such limitations while pursuing a more general treatment of commutativity, we numerically solve the proposed SDP \eqref{eq:SDP2} for the four scenarios considered and the three generators depicted in Fig. \ref{fig:Histograms}. Among them, the only scenario in which it was possible to find an optimal emergent dynamics, $\Gamma^{\text{opt}}_{D\vert C}$, that commutates the diagram for all initial states was in scenario 1. We report in Table \ref{tab:1}, for each generator, the optimal value achieved of $\epsilon$ which minimizes such an optimization program.
\begin{table}[H]
\centering
\begin{tabular}{| c| c| }
Generator & $\vert \vert  \Gamma^{\text{Petz}}_{D\vert C, \rho_A} - \Gamma_{D\vert C}^{\text{opt}} \vert \vert _{\diamond}$ \\ 
\hline
 $\rho_{\text{ME}}$ & $1.66$ \\
 $\rho_{\text{MM}}$ & $0.42$  \\
 $\rho_{\text{W}}$ & $0.55$ \\
\end{tabular}
\caption{Optimal values achieved by solving SDP \eqref{eq:SDP2} for the three generators. The scenario 1 was the only scenario in which such a solution was feasible.}
\label{tab:1}
\end{table}

Additionally, in agreement with the above results, by solving numerically the feasibility SDP described in eq. \eqref{eq:SDP3} for the four scenarios here approached, the only one for which the program was solvable was setup 1; with the optimal (feasible) solution given explicitly by
\begin{equation} \label{eq:ChoiId}
\rho^{\text{opt}}_{D\vert C} = 
\begin{pmatrix}
    1 & 0 & 0 & 1 \\
    0& 0& 0& 0 \\
    0& 0& 0& 0 \\
    1& 0& 0 &1 
\end{pmatrix}.
\end{equation}
Since this solution is the Choi-isomorphic operator to the identity channel, we can conclude that, among all the scenarios evaluated, the scenario in which the microscopic dynamics is compatible with the coarse-grained description is the one composed of the blurred and saturated detector and the SWAP channel. With the emergent, macroscopic dynamics given solely by the identity channel, this fact, in agreement with the operational construction shown along Sec. \ref{SubSec.FullyQuantumBnSSWAP}, shows the consistency of our proposed SDP to efficiently finding such a general solution---when it exists---for the coarse-graining problem in the quantum Bayesian formalism here employed. Furthermore, tracing a parallel with the results displayed in Sec. \ref{subsec:DiagramCommutativity}, we would like to emphasize that, even in the scenarios where there is no general emergent dynamics---by the fact that the microscopic dynamics is not compatible with the coarse-grained description, such as scenarios 2, 3 and 4---our proposed Bayesian solution $\Gamma^{\text{Petz}}_{D\vert C, \rho_A}$, with the right choice of its generator, can be employed pragmatically to solve the problem---at least for the generator itself and a limited number of initial states. 

We now showcase the robustness aspect of such microscopic (compatible) dynamics. More specifically, given that from the above discussions we know that the SWAP channel is compatible with the coarse-grained description represented by the blurred and saturated detector map, we restrict our attention to the discussions around the proposed robustness measure and its SDP characterization displayed in eq. \eqref{eq:SDP4}. Essentially, since the microscopic dynamics $\mathcal{U}^{\sigma_z}_{B|A}$ is not compatible with the $\Lambda_{\text{BnS}}$ map, we initially take it to be the added noise dynamics. By numerically solving the SDP enunciated in eq.  \eqref{eq:SDP4} with its respective causal conditional states, we obtain, for the $CG$-compatibility robustness of the SWAP channel with respect to $\Lambda_{\text{BnS}}$ and this noise, that $r_{c}(\mathcal{U}^{\text{swap}}_{B|A}) \approx 0$. In other words, up to numerical imprecision, the compatibility of $\mathcal{U}^{\text{swap}}_{B|A}$ with the coarse-grained description $\Lambda_{\text{BnS}}$ is not robust against the noise generated by $\mathcal{U}^{\sigma_z}_{B|A}$.  

Moreover, conducting the same analysis except for changing the noise dynamics by: (1) a random unitary evolution $\mathcal{U}^{rand (1)}_{B\vert A} (\cdot): = U^{(1)} (\cdot) U^{(1) \dagger} $ where $U^{(1)}$ is some random element\footnote{The random unitary matrices used in this analysis were generated using SciPy's~\cite{2020SciPy-NMeth} \href{https://docs.scipy.org/doc/scipy/reference/generated/scipy.stats.unitary_group.html}{\texttt{scipy.stats.unitary\_group.rvs}} method.} of the group $U(4)$; and (2) a random product unitary evolution $\mathcal{U}^{rand (2)}_{B\vert A} (\cdot): = U^{(2)} (\cdot) U^{(2) \dagger} $ where $U^{(2)}:= U_A \otimes U_B$ for some random elements $U_A, U_B$ of the group $U(2)$, the results obtained are in agreement with the one shown previously; with the $CG$-compatibility robustness of the SWAP channel, up to numerical imprecision, being equal to zero. Therefore, at least for these three types of noise, we can conclude that although the SWAP channel is compatible with the coarse-grained description raised by the blurred and saturated detector map, its compatibility is not robust.

This fact alone suggests how difficult the existence of a macroscopic, emergent dynamics is in the presence of any undesirable microscopic noise. It also suggests that, since we have chosen the noisy, resulting microscopic dynamics to be CPTP, even in this noisy scenario one is able to use the $\Gamma^{\text{Petz}}_{D|C, \rho_A}$ solution to solve the problem at least for the generator itself---and for a limited number of initial states. Thus, remarking a valuable property of the $\Gamma^{\text{Petz}}_{D|C, \rho_A}$ solution.

Implementing the robustness-inspired SDP proposed in eq.  \eqref{eq:SDP5} and taking into account the previous verifications conducted above, we now investigate numerically whether there exists a microscopic dynamics which allows the existence of a general emergent dynamics for the non-compatible scenarios 2, 3, and 4. In particular, we are willing to answer, whether there exists a microscopic dynamics such that its convex combination with the initial one becomes compatible with that coarse-grained description.

Letting our standpoint be the discussions around eq.   \eqref{eq:ResultantMicroDynamics}, with the convex combination parameter $\gamma$ assuming $1.5 \times10^3$ numerical values evenly spaced within the interval $[0,1]$, we solve numerically the feasibility SDP proposed in \eqref{eq:SDP5} for the three aforementioned scenarios. In Table \ref{tab:2}, we report the sub-intervals of $\gamma$ in which, by solving such feasibility SDP, it was possible to find: (1) a (optimal) feasible microscopic dynamics $\Psi^{t}_{B\vert A}$ which, together with the respective initial (non-compatible) microscopic dynamics $\mathcal{U}^t_{B\vert A}$ of the scenario, allows us to construct the resultant dynamics $\mathcal{J}^t_{B\vert A}$ (eq.   \eqref{eq:ResultantMicroDynamics}), compatible with the respective coarse-grained description; and (2) an effective macroscopic dynamics $\Theta^t_{D\vert C}$ that satisfies the commutativity relation \eqref{eq:CommutRelationBP} for this resultant configuration.
\begin{table}[H]
\centering
\begin{tabular}{| c| c| }
Scenarios & Solution ($\Theta^{t}_{D\vert C}$) attainable \\ 
\hline
  2 ($\mathcal{U}^t_{B\vert A} = \mathcal{U}^{\sigma_z}_{B\vert A}$) & $0\leq\gamma \leq 0.557$ \\
  3 ($\mathcal{U}^t_{B\vert A} = \mathcal{U}^{swap}_{B\vert A}$)& $ 0\leq\gamma \leq 0.249$ \\
  4 ($\mathcal{U}^t_{B\vert A} = \mathcal{U}^{\sigma_z}_{B\vert A}$)& $ 0\leq\gamma \leq 0.524$ \\
\end{tabular}
\caption{Sub-intervals of $\gamma$ in which we can promote compatibility in the three non-compatible scenarios.}
\label{tab:2}
\end{table}
From these results, we can claim that, by running the proposed SDP \eqref{eq:SDP5} for a given non-compatible coarse-graining scenario, one is able---at least for a sub-interval of $\gamma$ values---to find a resultant microscopic dynamics compatible with a given coarse-grained description. Furthermore, from a pragmatic viewpoint, since the SDP \eqref{eq:SDP5} also returns an effective macroscopic dynamics for that resultant scenario, together with the results on the performance of $\Gamma_{D \vert C, \rho_A}^{\text{Petz}}$ discussed in Sec. \ref{subsec:DiagramCommutativity}, two conclusions can be drawn. First, the problem can be solved for all initial states whenever, for each non-compatible scenario considered, the parameter $\gamma$ lies within the sub-intervals reported in Table \ref{tab:2}. Second, for values of $\gamma$ outside these intervals, given that the resultant microscopic dynamics remains CPTP, one is able to solve the problem in a state-by-state manner by using our proposed $\Gamma^{\text{Petz}}_{D\vert C, \rho_A}$, at least for its generator or for a limited number of underlying initial states.

\section{Conclusion}\label{Sec.Conclusion}

Coarse-grained descriptions are ubiquitous in the sciences. Being virtually unescapable, they represent rational attempts to make sense of complex phenomena when an account of all the involved degrees of freedom is impossible or simply unreasonable---either because there are far too many or because we are limited to interacting with only a few of them. Working under this perspective, in this paper, we further explored the emergence of a compatible coarse-grained dynamics when information is lost or mixed (or a combination of the two) along the way. The situation we investigated is best represented in the diagram of Fig.~\ref{fig:Coarse_graining_diagram}.

Under this framework, given a microscopic dynamics $\mathcal{U}_{t}$ and a coarse-graining map $\Lambda_{\text{CG}}$, thought as representing the physical process of loss of information, questioning whether there exists an emergent CPTP map $\Gamma_t$, such that the following identity holds true 
\begin{align}
    \Gamma_{t} \circ \Lambda_{\text{CG}} = \Lambda_{\text{CG}} \circ \mathcal{U}_{t}
\end{align}
has already been shown to be a highly non-trivial problem~\cite{DCBM17, DATM20}. What we saw in this work is that such a non-triviality potentially arises from the correlations present between the coarse-grained system itself and the traced-out degrees of freedom. Depending on $\Lambda_{\text{CG}}$ and what the microscopic evolution $\mathcal{U}_{t}$ does with the system, there might not be enough information left to make the diagram commute. Nonetheless, this problem certainly fades away when $\Lambda_{\text{CG}}$ is unitary, so that one can essentially reverse the direction of the leftmost arrow in the diagram of Fig.~\ref{fig:Coarse_graining_diagram}, therefore pointing to a potential solution. In other words, we could define an emergent map as below:
\begin{align}
    \Gamma_{t}:= \Lambda_{\text{CG}} \circ \mathcal{U}_{t} \circ \Lambda_{\text{CG}}^{-1}
\end{align}
Even though this reversal does not work in general, whenever $\Lambda_{\text{CG}}$ is not unitary, one could re-frame the original coarse-graining problem into a quantum (Bayesian) inference scenario and ask for the best guess for $\Gamma_{t}$, given an apriori microscopic state $\rho_{A}^{\mbox{priori}}$. Under these lenses, a reversal for $\Lambda_{\text{CG}}$ is well-defined and it is given by its associated Petz recovery map. By doing so, an effective emergent dynamic can be defined---see section~\ref{Sec.CGAsBayesianInference}. This potential solution does come with its own inherent costs, though. It explicitly depends (non-linearly) on the a priori state $\rho_{A}^{\mbox{priori}}$---leaving open the door for the question, `what is the optimal a priori state?'---and may not be an optimal solution overall: as we argued in Sec.~\ref{Sec.SDPAnalysis}, in scenarios where a solution does exist, when comparing the Petz-inspired emergent map with the actual solution, the former does not recover the latter. In other words, the Petz reversal map solution should be viewed as an educated guess for those situations where no other solution is possible. There might be a possible refinement of this line of research. Recently, the authors of ref.~\cite{JM23} investigated a scenario of state retrieval beyond the paradigm of Bayesian inference. Although one of the tenets of our standpoint was to advance a Bayesian subjectivist agenda for interpreting the state vector, we believe that moving beyond this paradigm may yield better results for physically motivated reversals that produce emergent, coarse-grained dynamics. We will explore this line of thought in a forthcoming work.

Beyond the inference reframing, one could determine analytically the possibility of an emergent map consistent with the diagram of Fig.~\ref{fig:Coarse_graining_diagram}. That is precisely what we did for several scenarios. Interestingly, whenever one of the horizontal levels of the diagram is classical, it is possible to find an adequate solution---in other words, hybrid diagrams where the coarse-graining is nothing but a preparation or a measurement, or even a combination of both, admit a compatible emergent coarse-grained dynamics. As we mentioned earlier, this suggests the fundamental role that correlations play in our framework. The same role can be seen when we addressed the problem in a component-wise manner for four paradigmatic scenarios: (a) blurred and saturated detector with a SWAP channel; (b) blurred and saturated detector with a z-interaction channel; (c) partial trace with a SWAP channel; and (d) partial trace with a z-interaction channel. In all those examples, we were able to pinpoint necessary conditions for the emergence of well-defined macroscopic dynamics. 

Lastly, motivated by the state-dependence limitation of our proposed Bayesian solution, we conducted a numerical evaluation in order to benchmark its performance in four concrete coarse-graining scenarios. Using the trace distance of the difference between the two sides of the conventional commutativity relation (eq.   \eqref{eq:traceNormCommuting}), we investigated the capacity of our solution in solve the problem, for a given coarse-graining scenario, in terms of different generator states in its construction. Our findings reveal a strong dependence in the choice of the generator state, with the solution performing better in the four scenarios here approached when it was generated from the maximally mixed state. Moreover, we observed that the solution can be employed, for a given coarse-graining scenario, not only for the a priori state itself but also for a limited number of random states.
Note that throughout this work, we have picked up states randomly according to the uniform measure. In theory, we might have used a different form of sorting out random states, but since we are not investigating asymptotic behaviors (where concentration effects start to take place), nor are we circumscribed by an operational task demanding a more adapted measure (where measures should be compatible with allowed transformations), we have opted for a simpler random process.

Since the evaluated Bayesian solution depends explicitly on the Petz recovery map, we claim that this finding might suggest a strong connection with the best performance observed for the Petz recovery map, as shown previously in ref. \cite{LFN22}, when it is generated from the maximally mixed state. This might constitute evidence, in parallel with that previous study, that the best choice of reference state for the Petz recovery map---not only for reverting the action of the dephasing and depolarizing channels, but also for the blurred and saturated detector coarse-graning map and for the partial trace---is the maximally mixed state. Additionally, we observed that, when varying the parameter of the Werner state taken as the reference state of the Petz recovery map in our solution, the more entangled this state is, the worse the results in reverting the coarse-graining maps. Suggesting, in that way, a fruitful direction for future investigations involving the role played by the reference state's quantum correlations---the generator of our solution, and consequently of the Petz recovery map---in coarse-graining scenarios beyond those addressed here.
 
 Furthermore, in our semidefinite programming investigation, we have proposed novel convex-optimization strategies to overcome the state-dependence limitation of our Bayesian solution, such as searching for general, state-independent solutions, as shown in SDPs \eqref{eq:AuxSDP2}-\eqref{eq:SDP3}. We also address additional aspects of the coarse-graining problem from the 
 perspective of the conditional state formalism as a whole. By viewing the compatibility of quantum dynamics in a coarse-graining scenario as a resource, we propose a new robustness measure (Definition \ref{eq:DefiRobustness} to quantify how much noise an initial compatible unitary dynamics can tolerate before its compatibility with a given coarse-grained description vanish. Additionally, drawing inspiration from it, we formulate a semidefinite program (SDP \eqref{eq:SDP5}) that searches a compatible microscopic dynamics---and, consequently, a compatible macroscopic, emergent dynamics---for initially non-compatible coarse-graining scenarios. It should be remarked that this constitute a step toward a comprehensive computational study---in tandem with the ones approached in \cite{DATM20}---of the essence of the conventional coarse-graining problem, as stated in ref. \cite{DCBM17}, within the quantum Bayesian framework for quantum theory. In addition to this, since we have benchmarked and implemented numerically our proposed programs for only four concrete coarse-graining scenarios, we hope that they might serve as a basis for further research in these directions, especially in evaluating our proposed optimization routines across other coarse-graining scenarios.

To conclude, we want to emphasise that this paper should be read in tandem with an earlier work by some of the authors~\cite{DATM20}. There, the authors tried to investigate the emergence of a well-defined coarse-grained dynamics via four mathematical perspectives. Our work not only deepens but also adds to the pool of tools the authors explored in that paper. While there, the authors were preoccupied with solving the coarse-graining problem in full generality; here, we set ourselves the task of reframing the whole question as an inference problem and analysing classes of cases that could shed a brighter light on the existence of a compatible emergent dynamics. It was precisely this more modest case study that motivated us to define quantifiers, such as the CG-compatibility robustness measure.

\begin{acknowledgments}
The authors thank Carlos Humberto and Lucas Porto for all the thought-provoking discussions. CD thanks the hospitality of the Institute of Quantum Studies and the camaraderie of Matthew Leifer, whose infinite patience was fundamental for this work. L L. Brugger thanks the support from the Fundação de Amparo à Pesquisa do Estado de Minas Gerais (FAPEMIG) for his PhD grant and also thanks the Programa de Bolsas de Pós-Graduação da Universidade Federal de Juiz de Fora (PPGF -UFJF) for his master grant. This study was financed in part by the Coordenação de Aperfeiçoamento de Pessoal de Nível Superior – Brasil (CAPES) – Finance Code 001. Thales B. S. F. Rodrigues thanks CAPES for his master's grant, as well as the PPGF - UFJF for its hospitality. This study was financed, in part, by the São Paulo Research Foundation (FAPESP), Brasil. Process Number 2025/27776-6. V. G. Valle thanks the Conselho Nacional de Desenvolvimento Científico e Tecnológico (CNPq) for financial support through his master’s grant and the PPGF–UFJF for its support and resources. This work was supported by Grant 63209 from the John Templeton Foundation. The opinions expressed in this publication are those of the authors and do not necessarily reflect the views of the John Templeton Foundation. This work was supported by CNPq through two grants from the Conhecimento Brasil Program (Lines 1 and 2). 
\end{acknowledgments}

%

\begin{widetext}
\appendix

\section{Coarse-Graining Mechanism and Quantum Conditional States: a brief overview of the formalism}\label{Sec.Formalism}
\setcounter{thm}{0}

This work should be seen as part of a larger project, which advances a particular standpoint: that quantum theory can be seen as a theory of (subjective) probabilistic assignment, one that appropriately emphasises the role played by the concept of agency. To begin our exploration of this path, we investigate the emergence of macroscopic quantum dynamics through the lens of (quantum) Bayesian inference. Two main ingredients are, then, necessary: the coarse-graining scenario of ref. \cite{DCBM17} and the conditional states formalism we borrow from \cite{LS13}. Each is addressed in much more detail elsewhere~\cite{LS14,DATM20}. Still, we briefly review either topic so that our work is self-consistent and the notation and nomenclature are accordingly uniformised. 
\subsection{The Coarse-Graining Problem}\label{SubSec.CGProblem}
From a numerical perspective, a full description of many-body quantum systems seems to be a computationally impossible task~\cite{Feynman1982, PerezGarcia2007}. The complexity and the number of parameters of such descriptions grows exponentially with their degrees of freedom, and a typical simulation of a system containing $6\times10^{23}$ particles, which takes into account all particle-particle interactions and does not leverage any approximation method (e.g., mean-field approximations and Monte Carlo methods), would take an entire lifetime to be executed on a classical computer. 

One way to address this limitation within the quantum information paradigm is to work with effective descriptions, as in refs.~\cite{DCBM17, COVM21, Duarte20,DATM20}. There, a coarse-grained dynamics is determined by a coarse-graining map and the underlying microscopic evolution. Such a coarse-graining is a completely positive and trace-preserving map~\cite{DCBM17}, not only responsible for reducing the dimensionality of the system---representing the incapacity to access all degrees of freedom of the microscopic system---but also models more complex effects, arising, for example, from faulty measurement devices~\cite{DCBM17}. This motivates the following definition:

\begin{defi}[Coarse-graining map] Let $\mathcal{H}_D \simeq \mathds{C}^{D}$  and $\mathcal{H}_d \simeq \mathds{C}^{d}$ be Hilbert spaces assigned, respectively, to a $D-$dimensional and a $d-$dimensional quantum systems. Let $\mathcal{L}({\mathcal{H}_D})$ and $\mathcal{L}({\mathcal{H}_d})$ be the sets of linear operators acting on $\mathcal{H}_D$ and $\mathcal{H}_d$, respectively. A coarse-graining map $\Lambda_{\text{CG}}$ is defined to be a CPTP linear map: 
    \begin{equation}\label{EDef.coarsegraining}
        \Lambda_{\text{CG}}: \mathcal{L}({\mathcal{H}_D}) \to \mathcal{L}({\mathcal{H}_d}),  
    \end{equation}
such that $D> d$.
\end{defi}
The idea is that by defining a coarse-graining map $\Lambda_{\text{CG}}$ in this way, one can obtain a macroscopic, less complex, yet informative description of quantum systems subjected to noise and information loss. 

For the underlying microscopic dynamics, it is more than reasonable to assume that the system evolves unitarily \cite{Nielsen_Chuang_2010}. So, there is a unitary quantum channel 
\begin{align}
    \mathcal{U}_t : \mathcal{D}{(\mathcal{H}_D)} \to \mathcal{D}{(\mathcal{H}_D)}
\end{align}
governing the system's microscopic evolution. As usual,  $\mathcal{D}(\mathcal{H}_D) = \{ \rho \in \mathcal{L}(\mathcal{H}_D ) \vert  \rho \geq 0, \Tr(\rho) = 1\}$ is the state space over $\mathcal{H}_D$ formed by positive-semidefinite operators of trace equals one. 

\begin{figure}
    \centering
    \includegraphics[width=0.30\linewidth]{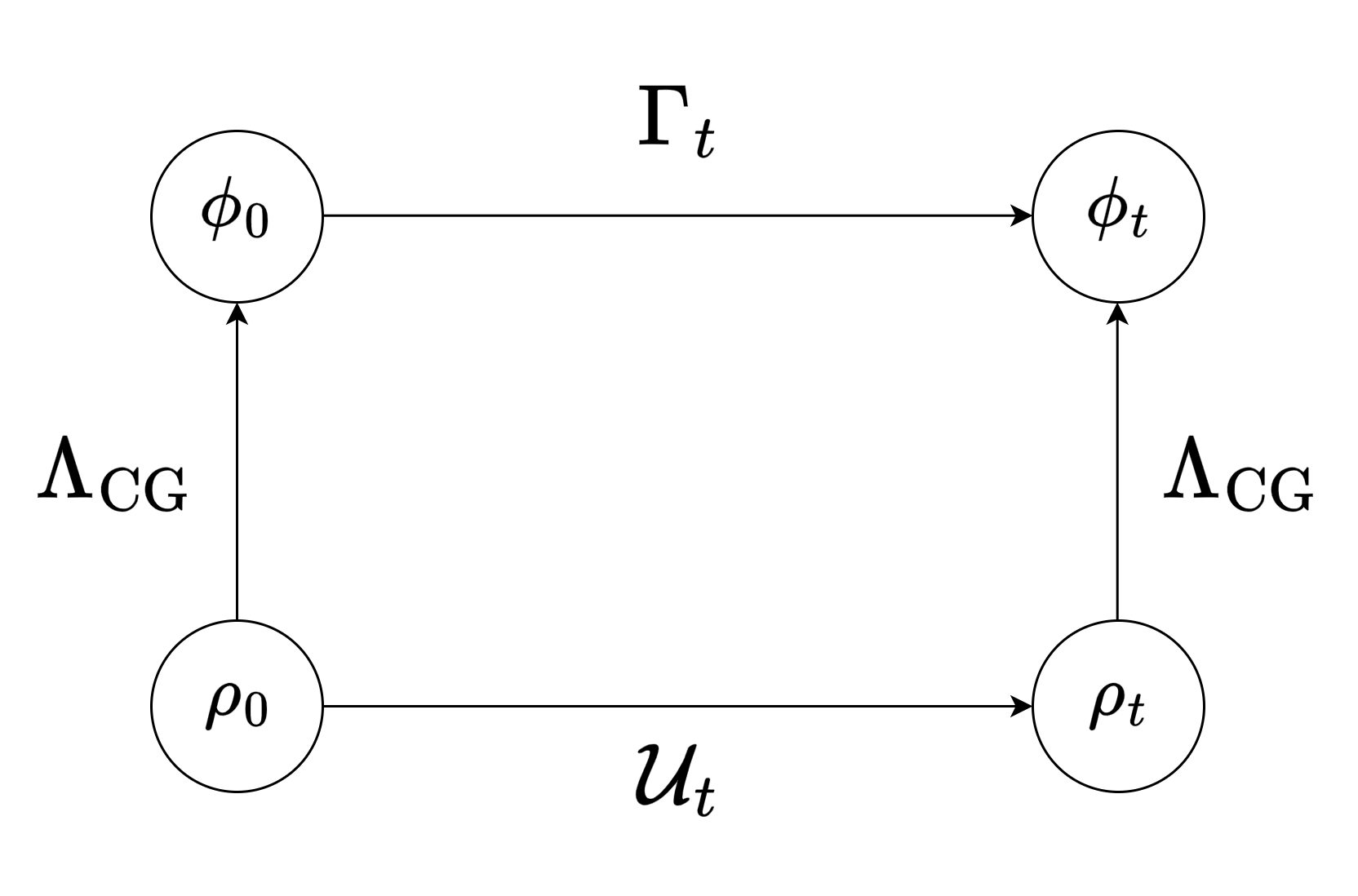}
    \caption{\textbf{The coarse-graining problem diagram.} The horizontal arrow at the bottom represents the unitary evolution of a closed many-body quantum system. Vertical arrows represent information loss, noise, or lack of access to all degrees of freedom, while the horizontal arrow at the top represent, whenever it exists, the emergent, less complex and macroscopic dynamics. }
    \label{fig:Coarse_graining_diagram}
\end{figure}

In this framework, the central question is: given a closed microscopic quantum system, whose state space is described by $\mathcal{D}(\mathcal{H}_D)$, what is the effective, emergent dynamics 
\begin{align}
    \Gamma_t : \mathcal{D}(\mathcal{H}_d) \to \mathcal{D}(\mathcal{H}_d)
\end{align}
dictated by the action of a coarse-graining map $\Lambda_{\text{CG}}$ and the unitary dynamics $\mathcal{U}_t$? Such a question, primarily presented in ref. \cite{DCBM17}, is known as the \textit{coarse-graining problem}. This problem, which has a neat diagrammatic representation, seee Fig.~\ref{fig:Coarse_graining_diagram}, is mathematically expressed by the attempt to satisfy the following relation: 
\begin{equation}\label{eq:CommutationRelation}
    \Gamma_{t} \circ \Lambda_{\text{CG}} = \Lambda_{\text{CG}} \circ \mathcal{U}_t,
\end{equation}
that in itself is another way to say that the diagram of Fig.~\ref{fig:Coarse_graining_diagram} commutes. One of the key aspects of the coarse-graining problem, which lies behind the commutativity relation seen in Eq.~\eqref{eq:CommutationRelation}, is the fact that, if such emergent dynamics $\Gamma_t$ exist---a property that does not hold in general--- one can describe the evolution of the underlying many-body quantum system through a less complex and accessible framework, rather than the $\mathcal{U}_t$ itself. In other words, looking again at the diagram of Fig. \ref{fig:Coarse_graining_diagram}, this regards obtaining the effective quantum state $\phi_t$ either by $\Lambda_{\text{CG}} \circ \mathcal{U}_t (\rho_0) = \Lambda_{\text{CG}}(\rho_t)$ or by $\Gamma_t \circ \Lambda_{\text{CG}}(\rho_0) = \Gamma_t(\phi_0)$, interchangeably. 

Upon this exposition, we define the notion of compatibility in a coarse-graining scenario.
\begin{defi}[Compatibility in a coarse-graining scenario]\label{def:CompatibilityCG}
Let $\mathcal{H}_D$ be a $D-$dimensional Hilbert space and let $\mathcal{H}_d$ be a $d-$dimensional Hilbert space, with $D > d$. Let $\mathcal{U}_t: \mathcal{L}(\mathcal{H}_D) \to \mathcal{L}(\mathcal{H}_D)$ be a unitary dynamics and let $\Lambda_{\text{CG}}: \mathcal{L}(\mathcal{H}_D) \to \mathcal{L}(\mathcal{H}_d)$ be a coarse-graining map. We say that $\mathcal{U}_t$ is compatible with $\Lambda_{\text{CG}}$ (or, equivalently, that the coarse-graining scenario itself is compatible) if there exists a CPTP emergent dynamics $\Gamma_t: \mathcal{L}(\mathcal{H}_d) \to \mathcal{L}(\mathcal{H}_d)$ such that 
    \begin{align}
    \Gamma_t \circ \Lambda_{\text{CG}} = \Lambda_{\text{CG}} \circ \mathcal{U}_t.
    \end{align}
\end{defi}\noindent
As we have mentioned before, this commutativity relation is the core of the coarse-graining problem advanced in \cite{DCBM17}. 

More details on the coarse-graining question, as well as some specific examples and applications of the formalism, can be found in refs. \cite{DCBM17, CM20,CF19}. We also explore four paradigmatic scenarios later on in the present work. Readers who need more examples may want to skip to subsections~\ref{SubSec.FullyQuantumBnSSWAP}-\ref{SubSec.FullyQuantumTrSigmaz}. We pass on to discuss the other ingredient we will need in our approach: the conditional quantum states.

\subsection{Conditional States Formalism}\label{SubSec.CQS}

\subsubsection{Channel-Operator Connection}\label{SubSubSec.ChannelOperatorConnection}
 
This section contains a brief overview of an adapted version of the conditional states formalism developed in ref.~\cite{LS13}. The main difference between the two formulations is just a matter of convenience---it is essentially a decision of where to write down a partial trace. Because we wanted to investigate the coarse-graining problem with semidefinite programming, we then choose to work with the Choi isomorphism~\cite{Choi:1975nug}, as in definition~\ref{def:ChoiIso} below, instead of the Jamio\l{}kowski's found in \cite{LS13, JAMIOLKOWSKI1972275}.

\begin{defi}[Choi Isomorphism] \label{def:ChoiIso} Let $\mathcal{H}_A$ and $\mathcal{H}_B$ be two finite-dimensional Hilbert spaces. Let
\begin{equation}
    \mathcal{E}_{B|A}: \mathcal{L}(\mathcal{H}_A) \to \mathcal{L}(\mathcal{H}_B) 
\end{equation}
be a linear map. The Choi image of $\mathcal{E}_{B|A}$ is given by an operator $\rho_{B|A} \in \mathcal{L}(\mathcal{H}_A \otimes \mathcal{H}_B)$ such that 
\begin{align}
    \rho_{B|A} &:= (\text{id}_A \otimes \mathcal{E}_{B|A}) \ket{\Phi^+}\bra{\Phi^+} \nonumber \\ 
    &= (\text{id}_A \otimes \mathcal{E}_{B|A})(\sum_{i,j = 0}^{d-1} \ket{i}\bra{j} \otimes \ket{i}\bra{j}) , 
\end{align}
where $d := dim(\mathcal{H}_A)$,
$\ket{\Phi^+} = \sum_{i=0}^{d-1}\ket{i} \otimes \ket{i}$ is an unnormalized maximally entangled state of $\mathcal{H}_A \otimes \mathcal{H}_A$ with respect to some preferred orthonormal basis $\{\ket{k}\}_{k=0}^{d-1}$ of $\mathcal{H}_A$, and $\text{id}_A$ is the identity map acting on $\mathcal{L}(\mathcal{H}_A)$. Also, for some $\sigma_{A} \in \mathcal{L}({\mathcal{H}_A})$, the action of $\mathcal{E}_{B|A}$ on $\mathcal{L}(\mathcal{H}_A)$ is recovered by 
\begin{equation}\label{eq:ChoiActing}
    \mathcal{E}_{B|A}(\sigma_{A}) = \Tr_A[\rho_{B|A}(\sigma_{A}^{T} \otimes \mathds{I}_B) ] \in \mathcal{L}{(\mathcal{H}_B)}.
\end{equation}
\end{defi}
The isomorphic connection between the Choi state $\rho_{B|A}$ and $\mathcal{E}_{B|A}$ can be promptly checked, and we do so in the Appendix~\ref{Sec.AppProof}. One relevant application of the Choi(-Jamio\l{}kowski) isomorphism, which is one of the cornerstones of the conditional states formalism, is that it maps any quantum channel to a positive semidefinite operator. This fact is highlighted in the following theorem, whose proof is also displayed in Appendix \ref{Sec.AppProof}. 
\begin{thm}\label{teo:ConnectsCPTPtoPSD}
    Let $\mathcal{E}_{B|A}: \mathcal{L}(\mathcal{H}_A) \to \mathcal{L}(\mathcal{H}_B)$ be a linear map and let $\rho_{B|A} \in \mathcal{L}(\mathcal{H}_A \otimes \mathcal{H}_B)$ be its Choi-isomorphic operator. Then, it follows that $\rho_{B|A}$ satisfies: 
    \begin{enumerate}
        \item $\rho_{B|A} \geq0 ;$
        \item $\Tr_B(\rho_{B|A}) = \mathds{I}_A$,
    \end{enumerate}
    if, and only if, $\mathcal{E}_{B|A}$ is CPTP.
\end{thm}

Note that the conditional notation we adopted for $\mathcal{E}_{B|A}$ and $\rho_{B|A}$ captures the idea of a generalised conditional probability~\cite{LS13, LS14}. The Choi image, $\rho_{B|A}$, is a semidefinite positive operator, and the trace over its outputs' space is the appropriate identity. These indices, as it will become clearer in a minute, are also more than simply notational shorthands for conditioning. They also denote the character of the Hilbert space, or the region, these operators are associated with. Simply put, an \textit{elementary region} is designed to denote any portion of the space-time where an agent might possibly make a single intervention on a quantum system through an experiment, either by making a measurement or by preparing a state. A \textit{region} (classical or quantum) denotes a collection of elementary regions. To each elementary region $A$ we associate a Hilbert space $\mathcal{H}_{A}$. To each composed region, say $AB$, we associate a Hilbert space $H_{AB}:=\mathcal{H}_{A} \otimes \mathcal{H}_{B}$, where  $\mathcal{H}_{A}$ is associated to $A$ and $\mathcal{H}_{B}$ to $B$. 

This focus on \textit{elementary regions} and \textit{regions} is an effort by the authors of refs.~\cite{LS13,LS14} to obtain a causally neutral theory of Bayesian quantum inference. From their standpoint, bipartite scenarios and physical processes governed by a CPTP map are supposed to be treated on the same footing, with a composed region $AB$ and an associated Hilbert space $H_{AB}:=\mathcal{H}_{A} \otimes \mathcal{H}_{B}$. Recall that usually these two scenarios are treated discreptantly. Bipartite states are positive semidefinte trace class objects belonging to $\mathcal{L}(\mathcal{H}_{A} \otimes \mathcal{H}_{B})$, while channels are completely positive, trace-preserving linear maps like $\mathcal{E}_{B|A}: \mathcal{L}(\mathcal{H}_{A}) \rightarrow \mathcal{L}(\mathcal{H}_{B})$. Even with this focus, the conditional states formalism has its inherent limitations. What it succeeds, though, is on establishing a formal parallel between concepts of quantum theory with classical probability formulas---see refs.~\cite{LS13,LS14}. In that sense, a classical random variable $X$ (or any other uppercase Latin letter near the end of the alphabet) finds its counterpart within the formalism through a quantum region that encodes its values. These regions are called classical regions and their associated Hilbert spaces come equipped with a preferred basis $\{\ket{x}\}_{x \in \text{Out}(X)}$ where every density operator in question is diagonal on that basis.

\subsubsection{Belief Propagation}\label{SubSubSec.BeliefPropagation}

We are now in a position to delve into the Bayesian conditioning aspect of the formalism, which, as shown in the subsequent section, enables us to reframe the coarse-graining problem as a Bayesian inference problem. To do so, we will adopt the `given' notation---drawing a parallel to the one commonly used in the context of classical conditional probabilities. So, when a linear map is written as $\mathcal{M}_{B|A}$, it means that
\begin{equation}
    \mathcal{M}_{B\vert A} : \mathcal{L}(\mathcal{H}_A) \to \mathcal{L}(\mathcal{H}_B),
\end{equation}
 and, extending by analogy, the operator isomorphic to it will be denoted by $\omega_{B\vert A} \in \mathcal{L}(\mathcal{H}_{AB})$. With this notation in hand, we are ready to define a certain class of conditional operators as follows. 
 
 \begin{defi} \label{def:acausalStates}
 An acausal conditional state $\rho_{B\vert A}\in \mathcal{L}(\mathcal{H}_{A} \otimes\mathcal{H}_{B})$ is an operator satisfying:
 \begin{enumerate}
     \item $\rho_{B\vert A} \geq 0$;
     \item  $\Tr_B(\rho_{B\vert A}) = \mathds{I}_A$,
 \end{enumerate}
 where $\mathds{I}_A$ is the identity operator acting on $\mathcal{H}_A$.
 \end{defi}
Note that, in the light of Theorem \ref{teo:ConnectsCPTPtoPSD}, any operator $\rho_{B\vert A}\in \mathcal{L}(\mathcal{H}_{AB})$, Choi-isomorphic to a linear map $\mathcal{E}_{B\vert A}$, is an acausal conditional state if and only if the map $\mathcal{E}_{B\vert A}$ is CPTP. Because of its meaning, we state it as a proposition, but it is essentially the same as Theorem~\ref{teo:ConnectsCPTPtoPSD} combined with the definition above

\begin{prop}\label{Thm.AcausalCPTP}
    Let $\mathcal{E}_{B|A}: \mathcal{L}(\mathcal{H}_A) \to \mathcal{L}(\mathcal{H}_B)$ be a linear map and let $\rho_{B|A} \in \mathcal{L}(\mathcal{H}_A \otimes \mathcal{H}_B)$ be its Choi-isomorphic operator. Then, the following are equivalent:
    \begin{enumerate}
        \item $\rho_{B|A}$ is an acausal state;
        \item $\mathcal{E}_{B|A}$ is CPTP.
    \end{enumerate}
\end{prop}

So far, we have defined what an acausal conditional state is. They are those operators in a one-to-one correspondence with CPTP maps via the Choi isomorphism. Note that in ref.~\cite{LS13} the authors choose to use the Jamio\l{}kowski isomorphism instead; and this is where they get their nomenclature from. Now, recall that given two (classical) random variables $X$ and $Y$, we can define a conditional probability distribution $P(X\vert Y)$, independently of the causal relation that holds between them. In standard probability theory, we do not need to establish whether there is a mechanism connecting the two variables; we simply assume a joint sample space in which they can be measured simultaneously. Thus, in line with treating quantum theory as a neutral theory of Bayesian inference, we need to define a second object, so we can deal with causal and acausal relationships alike.  

\begin{defi}\label{def:CausalState}
A causal conditional state of $B$ given $A$ is an operator $\varrho_{B\vert A} \in \mathcal{L}(\mathcal{H}_{AB})$ that can be denoted as
 \begin{equation}
     \varrho_{B\vert A} = \rho_{B\vert A}^{T_A},
 \end{equation}
for some acausal conditional sate $\rho_{B\vert A} \in \mathcal{L}(\mathcal{H}_{AB})$. $T_A$ represents the partial transpose with respect to some basis on $\mathcal{H}_A$.
\end{defi}

Based on this definition, one can directly recognize that, if $\rho_{B\vert A}$ is the Choi-image operator of a given quantum dynamics $\mathcal{E}_{B\vert A}$, then the causal conditional state $\varrho_{B\vert A}$ is the Jamio\l{}kowski-isomorphic \cite{JAMIOLKOWSKI1972275} operator to it. Furthermore, due to the partial transposition operation, the causal conditional state $\varrho_{B\vert A}$ defined upon a acausal conditional state $\rho_{B|A}$ is generally not a positive operator \cite{LS13}. This fact permeates one of the major limitations of this formalism, and is discussed in more detail in the main text. 

\textbf{Remark.} As we mentioned before, we defined an acausal conditional state with the Choi isomorphism. In~\cite{LS13,LS14} the authors have opted to define an acausal state via the Jamio\l{}kowski isomorphism. Essentially, the difference between our conditional states and theirs is simply where one sticks the partial transposition. We agree that, under their definition, the physical intuitions behind `causal' and `acausal' scenarios are clearer. However, because we are investigating coarse-graining scenarios using semi-definite programming, we thought it would be wiser to stick to objects in the cone of semi-definite matrices---a similar move was also done in~\cite{DuarteEtAl25}. If the reader wants to get rid of this potentially confusing issue, they can safely replace `acausal conditional states' with `Choi-conditional states' and  `causal conditional states' with `Jamio\l{}kowski-conditional states', and the results would be exactly the same. In this sense, given an arbitrary linear map $\mathcal{N}_{B\vert A}: \mathcal{L}(\mathcal{H}_A) \to  \mathcal{L}(\mathcal{H}_B)$ and the Jamio\l{}kowski-isomorphic operator $\varrho_{B\vert A}$ associated to it, the former is a quantum channel if and only if $\varrho^{T_A}_{B\vert A} \geq 0$ and $\Tr_{B}(\varrho_{B\vert A}^{T_A}) = \mathds{I}_A$, i.e., when $\varrho_{B\vert A}^{T_A}$ is a valid acausal conditional state.


  

From the inferential standpoint we want to advance, because $\varrho_{B\vert A}$ is supposed to carry the \textit{causal} connection between the regions $A$ and $B$, it can be used to propagate an agent's beliefs across these two regions. Basically, this \textit{belief propagation} takes states from region $A$ to region $B$ as below:
%
\begin{equation}
    \mathcal{N}_{B\vert A} (\rho_A) = \Tr_A[\varrho_{B\vert A}(\rho_A\otimes \mathds{I}_B)] \in \mathcal{D}(\mathcal{H}_B).
\end{equation}
Note that this expression differs slightly from eq.   \eqref{eq:ChoiActing}. Here we are now looking at the action of $\mathcal{N}_{B\vert A}$ in terms of its Jamio\l{}kowski-isomorphic operator rather than the Choi-isomorphic one---so, as anticipated, we miss the partial transposition. Fig.~\ref{Fig2} summarizes this discussion, which also advances hybrid operators and its relation to measurements and ensemble preparation--- the main topic of items \textit{(c)} and \textit{(d)} below.

\begin{figure}[h]
\centering
\subfigure[]{
     \includegraphics[scale=0.3]{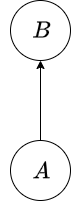}\label{subfig:2a}}\hspace{0.5cm}
     \hspace{1.0cm}
\subfigure[]{
    \includegraphics[scale=0.3]{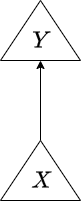}\label{subfig:2b}}\\
\subfigure[]{
    \includegraphics[scale=0.3]{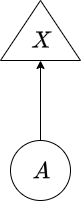}\label{subfig:2c}}\hspace{0.5cm}
     \hspace{1.0cm}
\subfigure[]{
    \includegraphics[scale=0.3]{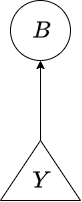}\label{subfig:2d}}
\caption{\textbf{Conditional States Formalism in Diagrams.} Classical regions are denoted by triangles, and quantum regions by circles. Arrows represent causal connections and the directions of causal influence. (a) \textit{General quantum dynamics:} the region $A$ represents the input of a quantum channel $\mathcal{N}_{B\vert A}$ and $B$ represents the output; (b) \textit{Classical Measurements:} this process takes the random variable $X$ as input as gives $Y$ as output. The conditional state $\rho_{Y\vert X} \in \mathcal{L}(\mathcal{H}_X \otimes \mathcal{H}_Y)$ carries the conditional probability distribution $P(Y\vert X)$; (c) \textit{Quantum measurement:} the hybrid operator $\rho_{X\vert A} \in{\mathcal{L}(\mathcal{H}_A \otimes \mathcal{H}_X)}$, designed to propagates beliefs in this case, carries a POVM $\{E_x^A\}_{x\in \text{Out(X)}} \subset \mathcal{L}(\mathcal{H}_A)$ in its composition. The Born rule is recovered as an instance of belief propagation; (d) \textit{Ensemble preparation:} a hybrid operator $\rho_{B\vert Y}\in \mathcal{L}(\mathcal{H}_Y \otimes \mathcal{H}_B)$ carries a set $\{\rho_{y}^B\}_{y\in \text{Out}(Y)} \subset \mathcal{D}(\mathcal{H}_B)$ in its definition. Belief propagation, in this case, gives the ensemble average state $\rho_{B} = \sum_y P(Y=y) \rho_y^B\in \mathcal{D}(\mathcal{H}_B)$.} 
\label{Fig2}
\end{figure}

\subsubsection{Hybrid States: Measurements}\label{SubSubSec.HybridStatesMeasurements}

When delving into a broad class of well-established references in quantum theory, the reader will inevitably encounter extensive discussions emphasizing that the process of measuring a quantum system lies at the very core of the theory. Although many philosophical questions arise from these discussions, in standard quantum theory textbooks~\cite{Nielsen_Chuang_2010,WatrousLectureNotes,cohen,LS13}, a measurement $\mathcal{M}$ with classical outcomes $x \in \mathrm{Out}(\mathcal{M})$ is represented by a POVM $\{E_{x}^{A}\}_{x \in \mathrm{Out}(\mathcal{M})}$ acting on the state space $\mathcal{H}_{A}$, such that
\begin{equation}
    \sum_{x\in Out(\mathcal{M})} E_{x}^{A} = \mathds{1}_{A}.
\end{equation}

Focusing on the measurement process itself—namely, the classical outcomes obtained through the measurement procedure and the corresponding probabilities of obtaining them—we observe that it can be represented by a channel that takes a quantum state as its input and outputs a classical state. In this way, the measurement process is represented by a channel that maps a quantum region to a classical one, formally given by $\mathcal{E}_{X\vert A}: \mathcal{L}(\mathcal{H}_{A}) \rightarrow \mathcal{L}(\mathcal{H}_{X})$, and naturally withing the conditional states formalism has a representation that we call as a \textit{hybrid conditional state} associated with the measurement process, and is defined as,
\begin{equation}\label{Ex.DefMesState}
    \varrho_{X\vert A} = \sum_{x\in Out(\mathcal{M})} \ket{x}\bra{x} \otimes E_{x}^{A}.
\end{equation}
Furthermore, the hybrid state given by expression~\eqref{Ex.DefMesState} satisfies both the definitions of acausal and causal states. A detailed demonstration and discussion of this fact are provided in ref.~\cite{LS13}.

\subsubsection{Hybrid States: Preparation}\label{SubSubSec.HybridStatesPreparations}

Besides measurement, another physical process extensively addressed in quantum theory is the ensemble preparation \cite{Nielsen_Chuang_2010,WatrousLectureNotes,cohen,LS13}. In this case, the process consists of constructing quantum states conditioned on classical measurement outcomes. Analogously, the ensemble preparation procedure maps classical states to quantum states and can therefore be interpreted as a channel from a classical region to a quantum region. Essentially, this channel is given by $\mathcal{E}_{A\vert X}: \mathcal{L}(\mathcal{H}_{X}) \rightarrow \mathcal{L}(\mathcal{H}_{A})$. Therefore, this physical process arises naturally within the conditional states formalism, and this is the second type of hybrid conditional states that we will be working with. Thus, this hybrid state is defined \cite{LS13} as follows,
\begin{equation}\label{Def.measurementHybrid}
    \varrho_{A\vert X} = \sum_{x\in \mathrm{Out}(\mathcal{M})} \ket{x}\bra{x} \otimes \rho^{A}_{x},
\end{equation}
where $\{\ket{x}\}_{x\in \mathrm{Out}(\mathcal{M})}$ is a preferred basis to the classical region and $\{\rho^{A}_{x} \in \mathcal{L}(\mathcal{H}_A)\}$ is a set of positive and normalized operators indexed by the values of $x$. The hybrid state, defined in the expression \eqref{Def.measurementHybrid}, also satisfies both definitions of causal and acausal conditional state. And once more, a detailed demonstration and discussion of this fact are provided in ref.~\cite{LS13}.

\subsubsection{Conditional States: Measure and Prepare}\label{SubSubSec:MeasureAndPrepare}

As previously presented, the conditional states isomorphic to the channels are then given by,
\begin{align}\label{Ex.hmxa}
    \varrho_{X\vert A} = \sum_{x} \ket{x}\bra{x} \otimes E_{x}^{A}, \\ \label{Ex.hepcx}
    \varrho_{C\vert X} = \sum_{x} \ket{x}\bra{x} \otimes \rho_{x}^{C}, \\ \label{Ex.hmyd}
    \varrho_{Y\vert D} = \sum_{y} \ket{y}\bra{y} \otimes E_{y}^{B}, \\ \label{Ex.hepyb}
    \varrho_{Y\vert B} = \sum_{y} \ket{y}\bra{y} \otimes \rho_{y}^{B}, 
\end{align}
where $y\in Out(Y)$ and $x \in Out(X)$.

Via Theorem \ref{teo:ChannelComposition}, we can find the conditional state resulting from the process of measuring a quantum system and then preparing an ensemble from the measurement outputs. More specifically, taking the hybrid states \eqref{Ex.hmxa} and \eqref{Ex.hepcx}, that corresponds to the left branch of the Fig.~\ref{fig_measure_and_prepare}, we have,
\begin{align}\label{Ex.apmeasureandpreparestate}
    \varrho_{C\vert A} &= \text{Tr}_{X}\left(\varrho_{C\vert X} \varrho_{X \vert A} \right) \nonumber \\ 
   &= \text{Tr}_{X}\left[\sum_{xx'} \left(\ket{x}\bra{x}\otimes \rho_{C}^{x}\right)\left(\ket{x'}\bra{x'}\otimes E_{x'}^{A}\right)\right] \nonumber \\ 
   &= \sum_{x} \rho_{C}^{x} \otimes  E_{x}^{A}.
\end{align}
The conditional state presented in the expression \eqref{Ex.apmeasureandpreparestate} is what we will call here the \textit{measure-and-prepare conditional state}. Again, by Theorem~\ref{teo:ChannelComposition} we have that its associated channel is $\mathcal{E}_{C\vert A}^{MP}$.

It is also direct to check that with $\{\rho_{C}^{X}\}$ being a set of normalized states and $\{E_{x}^{A}\}$ a POVM over $\mathcal{L}(\mathcal{H}_A)$, the conditional state as given in expression \eqref{Ex.measureandpreparestate} satisfies both Def.~\ref{def:acausalStates} and Def.~\ref{def:CausalState} of acausal and causal states.\footnote{This check can be carried out in a manner analogous to that employed for hybrid conditional states in ref.~\cite{LS13} and won't be done here.}

Considering that we have access to an initial state of the region $A$, namely $\rho_{A}$ we can, therefore, find ---via belief propagation ---the state $\rho_C$, that is,
\begin{align}\label{ex.bpmprhoc}
    \rho_{C} &= \text{Tr}_{A}\left( \varrho_{C\vert A}\rho_{A} \right) \nonumber \\ \nonumber
    &= \text{Tr}_{A}\left[\left(\sum_{x} \rho_{C}^{x} \otimes  E_{x}^{A}\right) \rho_{A}\right] \\ 
    &=\sum_{x} \text{Tr}\left(E_{x}^{A}\rho_{A}\right)\rho_{C}^{x}.
\end{align}

\subsubsection{Quantum Bayes Inversion}\label{SubSubSec.QuantumBInversion}

We show that it is possible to write a quantum version of the Bayes' rule within the conditional states formalism. To create a more transparent parallel between the classical and quantum versions of the inversion rule, we will first define an auxiliary, non-associative binary product \cite{LS13}
\begin{defi}[Star product]\label{AppDef.StarProduct}
    For any $M_{AB} \in \mathcal{L}(\mathcal{H}_{AB})$ and $N_A \in \text{Pos}(\mathcal{H}_{A})$, we define the $\star-$product as follows,
\begin{equation}
    M_{AB} \star N_A := (N_A^{\frac{1}{2}} \otimes \mathds{I}_B) M_{AB}(N_A^{\frac{1}{2}} \otimes \mathds{I}_B). 
\end{equation}
\end{defi}

Classically, recall that when $X$ and $Y$ are two random variables, the Bayes inversion formula for $P(X|Y)$ and $P(Y|X)$ is given by~\cite{james2004probabilidade}
\begin{equation}
   P(Y\vert X) = \frac{P(X \vert Y) P(Y)}{P(X)},
   \end{equation}
where $P(X\vert Y) = P(X,Y)P(Y)^{-1}$, given the joint distribution $P(X,Y)$. Quantumly, we can do essentially the same to derive a \textit{quantum Bayes' inversion} formula \cite{LS13} to obtain a conditional quantum state $\sigma_{A\vert B}$ of $A$ given $B$ from the conditional $\sigma_{B|A}$. To do so, it suffices to remind that the joint state of a pair of regions $AB$ is determined by $\sigma_{AB}=\sigma_{B|A} \star \rho_{A}^{-1}$, and that analogously $\sigma_{AB}=\sigma_{A|B} \star \rho_{B}^{-1}$; where $\rho_{A}:=\mbox{Tr}_{B}(\sigma_{AB})$ and $\rho_{B}:=\mbox{Tr}_{A}(\sigma_{AB})$. Consequently,
\begin{equation}\label{eq:QBayesTheorem}
    \sigma_{A\vert B} = \sigma_{B\vert A} \star (\rho_A \rho_B^{-1}). 
\end{equation} 

In the causal case, when $\varrho_{B\vert A}$ is Jamio\l{}kowski-isomorphic to a channel $\mathcal{E}_{B\vert A}$, if $\rho_A$ is the state that describes the agent's beliefs about a region A and $\rho_B$ is the result of the belief propagation from $A$ to $B$ determined by $\rho_B = \mathcal{E}_{B\vert A}(\rho_A) = \Tr_A (\varrho_{B|A} \rho_A)$, then the operator $\varrho_{A\vert B}$---dubbed the \textit{quantum Baeysian inversion} of $\varrho_{B\vert A}$---can be written as 
\begin{equation}
    \varrho_{A\vert B} = \varrho_{B\vert A} \star(\rho_A \rho_B^{-1}).
\end{equation}

While $\varrho_{B \vert A}$ is responsible for propagating beliefs from $A$ to its causal future $B$, its Bayesian inversion $\varrho_{A \vert B}$ retrospects beliefs from $B$ to $A$. Although naive, this digression is fundamental. It emphasises, in a natural way, the potential Bayesian character we can give to the coarse-graining formalism as a whole, as shown in more detail in the main text.

\subsubsection{Petz Recovery Map}\label{SubsubSec.PetzInversion}

Strikingly, given an initial state $\rho_{A}$ and a quantum channel $\mathcal{N}_{B|A}$, the Jamio\l{}kowski-inverse image of state associated with the quantum Bayesian inversion $\rho_{A|B}$ is exactly the well-known Petz recovery channel $\mathcal{R}_{A\vert B}$ \cite{LS13,Petz1986}, whose expression is given explicitly by: 
\begin{equation} \label{eq:PetzMap}
    \mathcal{R}_{A\vert B} (\cdot) = \rho_A^{\frac{1}{2}} \{\mathcal{N}^{\dagger}_{B\vert A} [\rho_B^{-\frac{1}{2}} (\cdot) \rho_B^{-\frac{1}{2}}]\}\rho_A^{\frac{1}{2}},
\end{equation}
where $\mathcal{N}^\dagger_{B\vert A} : \mathcal{L}(\mathcal{H}_B) \to \mathcal{L}(\mathcal{H}_A)$ is the adjoint of $\mathcal{N}_{B\vert A}$ and $\rho_{B}:=\mathcal{N}_{B|A}(\rho_{A})$. 

Such a map, which also appears in the literature on quantum computation, is recognised as the channel that achieves near-optimal quantum error correction when the initial state and quantum channel are both known~\cite{BK02}. What is remarkable is that in the conditional states formalism, it is nothing but a Bayes inversion from the original channel $\mathcal{N}_{B|A}$---which suggests that the quantum Bayes rule is the (quasi-)optimal way to propagate beliefs backwards from $B$ to $A$ creating a neat parallel between quantum (Bayesian) inference and quantum error correction~\cite{LS13}.


\subsubsection{Channel Composition Rule}\label{SubSubSec.CompositionRule}

To conclude, we present the channel composition rule viewed as an instance of belief propagation \cite{LS13}. Such a result, as seen in the main text, is one of the key ingredients we used to reframe the coarse-graining problem in the light of the conditional states paradigm.

\begin{thm}\label{teo:ChannelComposition}
Let $\mathcal{E}_{B\vert A}, \mathcal{E}_{C\vert B}$ and $\mathcal{E}_{C\vert A}$  be linear maps, and $\varrho_{B\vert A}, \varrho_{C\vert B}$ and $\varrho_{C\vert A}$, their respective Jamio\l{}kowski isomorphic operators. Then, it holds that $\mathcal{E}_{C\vert A} = \mathcal{E}_{C\vert B} \circ \mathcal{E}_{B\vert A}$ if, and only if, the Jamio\l{}kowski isomorphic operators satisfy  
\begin{equation}\label{eq:compositionRuleJami}
    \varrho_{C\vert A} = \Tr_B[(\mathds{I}_A\otimes\varrho_{C\vert B}) (\varrho_{B\vert A} \otimes \mathds{I}_C)].
\end{equation}
\end{thm}
The proof can be found in Appendix \ref{Sec.AppProof}. Obviously, an analogue of this Theorem can be directly derived when one considers the Choi-isomorphic operators $\rho_{C\vert A}, \rho_{C\vert B}$ and $\rho_{B\vert A}$, associated, respectively, with each of these maps rather than the Jamio\l{}kowski ones. The only difference lies in the necessity, on the right-hand side of eq. \eqref{eq:compositionRuleJami}, to partially transpose the Choi operator associated with the map $\mathcal{E}_{B\vert A}$ on the conditioned region $B$, yielding 
\begin{equation}\label{eq:compositionRuleChoi}
    \rho_{C\vert A} = \Tr_B[(\mathds{I}_A\otimes\rho_{C\vert B}) (\rho_{B\vert A}^{T_B} \otimes \mathds{I}_C)].
\end{equation}
We point out that this version, expressed in terms of the Choi-isomorphic operators, is extensively used in the SDP formulations presented in this work, particularly when representing the conventional compatibility relation in eq.~\eqref{eq:CommutationRelation} in terms of the appropriate conditional states. 


\section{Some Solutions to Variations of the CG Problem}\label{Sec.AppSupMaterial}

\subsection{Fully Classical Case}\label{SubSec.AppSupMaterialFullyClassical}

As we look critically at the conditional state obtained in expression~\eqref{Ex.classicalCSXYMainText}, we notice that the state $\rho_{R\vert X}$ is required in the composition in order to obtain the full form of the conditional state $\rho_{Y\vert X}$---similar to the fully quantum case where we used $\varrho_{A\vert C}$. Since we initially have access only to $\rho_{X\vert R}$, we must determine $\rho_{R\vert X}$ via the Bayes' rule.

Considering that we have access to a marginal state $\rho_{X}$ of region $X$, we can find the joint-state of the regions $X$ and $R$ as,
\begin{equation}
    \rho_{XR} = \rho_{X\vert R} \star \rho_R,
\end{equation}
and marginalize it in order to obtain,
\begin{equation}
    \rho_{X} = \text{Tr}_{R}\left(\rho_{XR}\right).
\end{equation}
Then, we do make use of the Bayesian inversion as presented in the expression \eqref{eq:QBayesTheorem}, obtaining,
\begin{equation}
    \rho_{R\vert X} = \rho_{X\vert R} \star \left(\rho_{R} \rho_{X}^{-1} \right).
\end{equation}
Hence, expression \eqref{Ex.classicalCSXYMainText} takes the form
\begin{align}\label{Ex.classicalXYfinalform}
    \rho_{Y\vert X} = \sum_{r,s,x,y}  P(Y\vert S)P(S\vert R)P(R\vert X) \ket{y}\bra{y}\otimes\ket{x}\bra{x},
\end{align}
where for simplicity we suppressed the indexes $r, s,x$ and $y$ in all terms such as $P(R=r)$ and the alike.

We can also verify that the state \eqref{Ex.classicalXYfinalform} respects the Def.~\ref{def:acausalStates}. In fact,
\begin{equation}
    \text{Tr}_{Y}\left( \rho_{Y\vert X} \right) = \sum_{r,s,x,y} P(Y\vert S)P(S\vert R)P(R\vert X)\ket{x}\bra{x},
\end{equation}
where fixing $S$ and summing over all $Y$ we have $\sum_{y}P(Y\vert S) =1$. Performing the analogous process for $P(S\vert R)$ and $P(R\vert X)$ , we obtain
\begin{equation}
     \text{Tr}_{Y}\left( \rho_{Y\vert X} \right) = \sum_{x} \ket{x}\bra{x} = \mathds{I}_X.
\end{equation}
Also, we can readily see that $\rho_{Y\vert X}$ is a acausal and causal state, satisfying both both definitions Def.~\ref{def:acausalStates} and Def.~\ref{def:CausalState}. 

\begin{figure}
    \centering
    \includegraphics[width=0.3\linewidth]{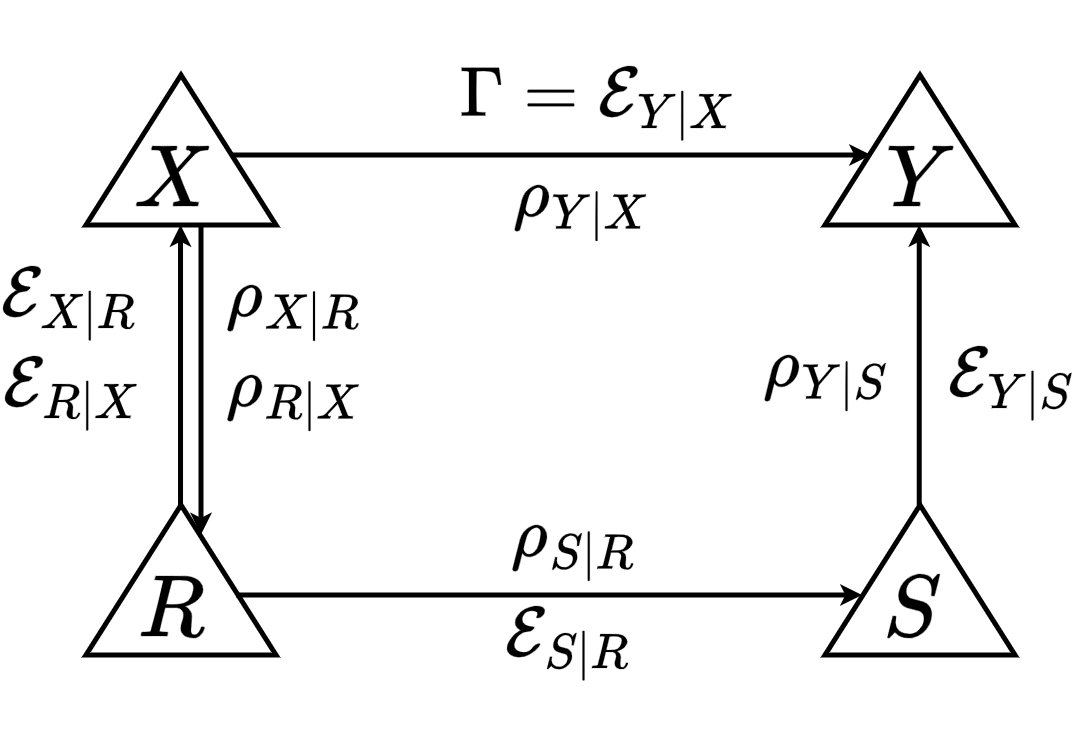}
    \caption{\textbf{Diagram representing the classical emergent dynamics}. In order to obtain the desired solution, we invert the left branch of the diagram via quantum Bayesian inference. In doing so, we were able to define $\Gamma = \mathcal{E}_{Y\vert X}$.}
    \label{fig_classical_emergent_dynamics}
\end{figure}

According to Theorem \ref{teo:ChannelComposition}, we stablish the isomorphism between conditional states and channels. Hence, the CPTP map in ~\eqref{Ex.ClassicalCPTPcomp}, namely
\begin{equation}
    \mathcal{E}_{Y\vert X} = \mathcal{E}_{Y\vert S} \circ \mathcal{E}_{S\vert R} \circ \mathcal{E}_{R\vert X}
\end{equation}
represents the emergent dynamics we are seeking for. This concatenation is pictured in Fig.~\ref{fig_classical_emergent_dynamics}.
Moreover, since our strategy for obtaining it combines Theorem~\eqref{teo:ChannelComposition} with the Bayesian inversion of one the arms of the diagram, and because we are dealing exclusively with classical regions, we encounter no fundamental limitations as we did in the full quantum case. 

The investigation carried out here, beyond merely establishing the emergent dynamics in classical scenarios, leads to an interesting result. Specifically, when we consider the scenario represented in Fig.~\ref{fig_classical_emergent_dynamics}, we find that, within the CSF framework, we are able to reproduce an instance of classical probability theory.

In fact, let us assume that we have access to the classical conditional probabilities $P(X\vert R)$, $P(S\vert R)$ and $P(Y\vert S)$, together with $P(R)$. We can find $P(X)$ by a belief propagation
\begin{equation}
    P(X) = \sum_r P(X\vert R) P(R)
 \end{equation}
as well as one can invert $P(X\vert R)$ via Bayes' theorem, obtaining $P(R\vert X)$. One can infer $P(Y)$ starting from $P(X)$. It is done by propagating $P(X)$ from region $X$ to $R$, then from $R$ to $S$ and, finally, from $S$ to $Y$, that is
\begin{equation}
    P(Y) = \sum_x P(Y\vert X)P(X),
\end{equation}
where
\begin{equation}\label{ex.classicalconditional}
    P(Y\vert X) = \sum_{r,s}P(Y \vert S)P(S\vert R) P(R \vert X),
\end{equation}
with expression \eqref{ex.classicalconditional}
being a valid conditional probability. 
The steps above are represented in 
Fig.~\ref{fig_4regclassical}.

Therefore, we can see that the conditional state \eqref{Ex.classicalXYfinalform} we achieved considering all the four regions to be classical recovers exactly the conditional probability \eqref{ex.classicalconditional} that we have obtained when doing belief propagation from $X$ to $Y$.
\begin{figure}
    \centering
    \includegraphics[width=0.6\linewidth]{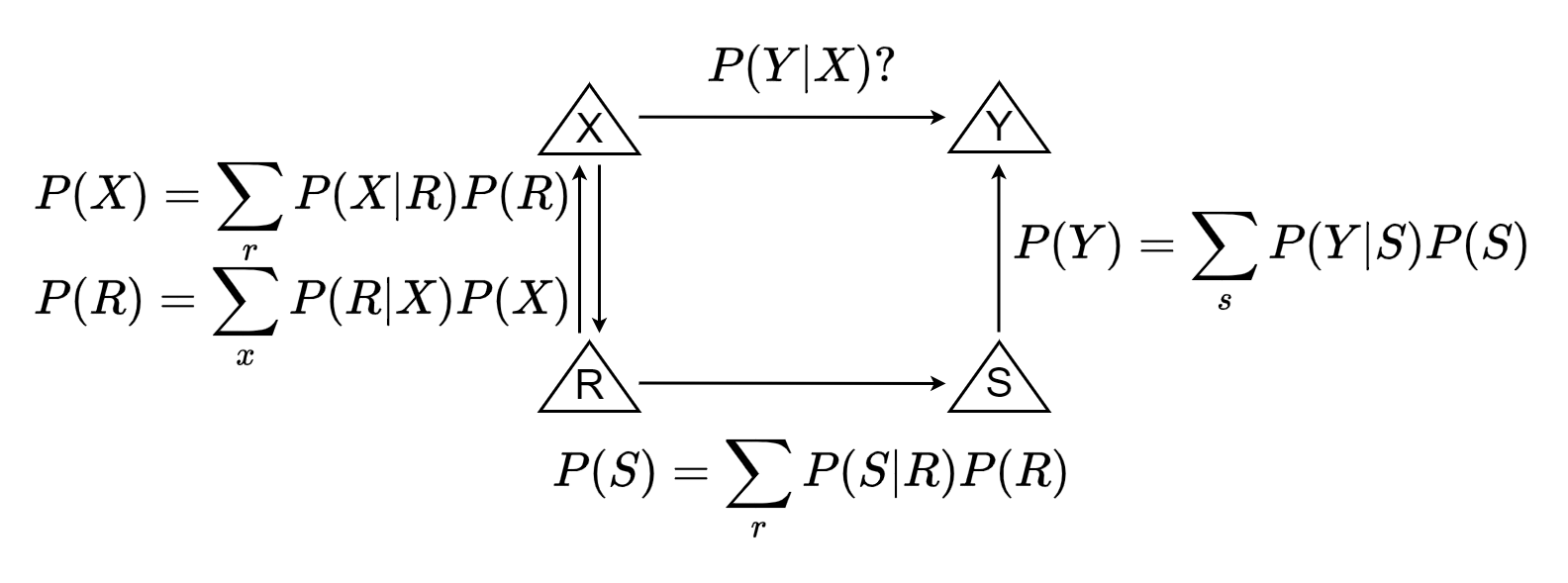}
    \caption{\textbf{Diagram representing the marginal and conditional probabilities together with their associated random variables.} Assuming that we have access to the conditional probabilities between the random variables and to the respective marginal probabilities distributions, we aim to find the belief propagation between the uppermost random variables, namely $X$ and $Y$.}
    \label{fig_4regclassical}
\end{figure}

\section{Hybrid Case - I:  Measurements}\label{SubSec.SolutionQuantumClassical}

As presented in section~\ref{Sec.CGAsBayesianInference}, in this work, we address coarse-graining as a mapping that reduces the dimensionality of the physical system under investigation. In this way, if we interpret a measurement as a process that maps a quantum region of dimension $D$ to a classical region of dimension $d < D$, we may regard this physical procedure as a form of coarse-graining. From this perspective, we are naturally led to ask what happens to the coarse-graining problem when the vertical arrows of the diagram of Fig.~\ref{fig:Coarse_graining_diagram} are replaced by measurement channels that emulate the coarse-graining, and, moreover, how the emergent dynamics behaves in this case---see Fig.~\ref{fig_diagram_hybrid_measurements_no_upper_line}.

The scenario investigated here can be described as follows. We consider two quantum regions, namely $A$ and $B$, connected by a channel $\mathcal{E}_{B\vert A}$, which is therefore associated with a causal conditional state $\varrho_{B\vert A}$. Subsequently, measurements are performed on both quantum regions, mapping region $A$ to a classical region $X$ and region $B$ to a classical region $Y$ and are, respectively, represented by the hybrid conditional states $\varrho_{X\vert A}$ and $\varrho_{Y\vert B}$---see Appendix~\ref{SubSubSec.HybridStatesMeasurements}. This concatenation of measurements is represented as follows,
\begin{equation}\label{Ex.classicaljointhybrid}
    \varrho_{XY} = \text{Tr}_{AB}\left( ( \varrho_{X\vert A} \varrho_{Y\vert B}) \varrho_{AB} \right),
\end{equation}
with expression \eqref{Ex.classicaljointhybrid} being the joint-state associated with the measurement outcomes. The joint-state of the composed region $AB$ can be obtained by the joint rule for conditional states, assuming that we have access to an initial state of the region $A$, namely $\rho_{A}$, we obtain $\varrho_{AB} = \varrho_{B \vert A} \star \rho_{A}$.

The reader may wonder why we resort to the joint state in expression~\eqref{Ex.classicaljointhybrid}, given that, \textit{(i)} in expression~\eqref{Ex.hybridjointstateXY} we again make use of the joint-state rule to recover the conditional state, and \textit{(ii)} considering that we are working within the coarse-graining problem, would it not be more expedient—as in previous sections—to work directly with the conditional state? The main reason for explicitly introducing $\varrho_{AB}$ here is that $\varrho_{B\vert A}$ is not, by itself, a physical state~\cite{LS13}.\footnote{Considering that a physical state of the composed region is represented by an operator $\rho \in \mathcal{D}(\mathcal{H}_{AB})$, we can see that $\rho_{B \vert A}$($\varrho_{B \vert A}$) might even not be a positive operator.} Consequently, since we are investigating measurement processes (and later, ensemble preparation), these operations must be performed on physically valid states.\footnote{Nonetheless, care must be taken when dealing with causal joint states, as these states might not be positive in general. However, since they yield the corrects marginal states of the composed region \cite{LS13}, therefore they remain valid to our analysis.}

In ref.~\cite{LS13}, a proposal similar to the one presented here can be found. However, in contrast to our approach, the authors in that work assume that the quantum regions are acausally connected, leading to a solution that differs slightly from ours. Here, we choose a causal connection between the two quantum regions because, as hinted earlier in this section, our goal is to investigate the coarse-graining problem when the coarse-graining is a measurement. 
\begin{figure}
    \centering
    \includegraphics[width=0.3\linewidth]{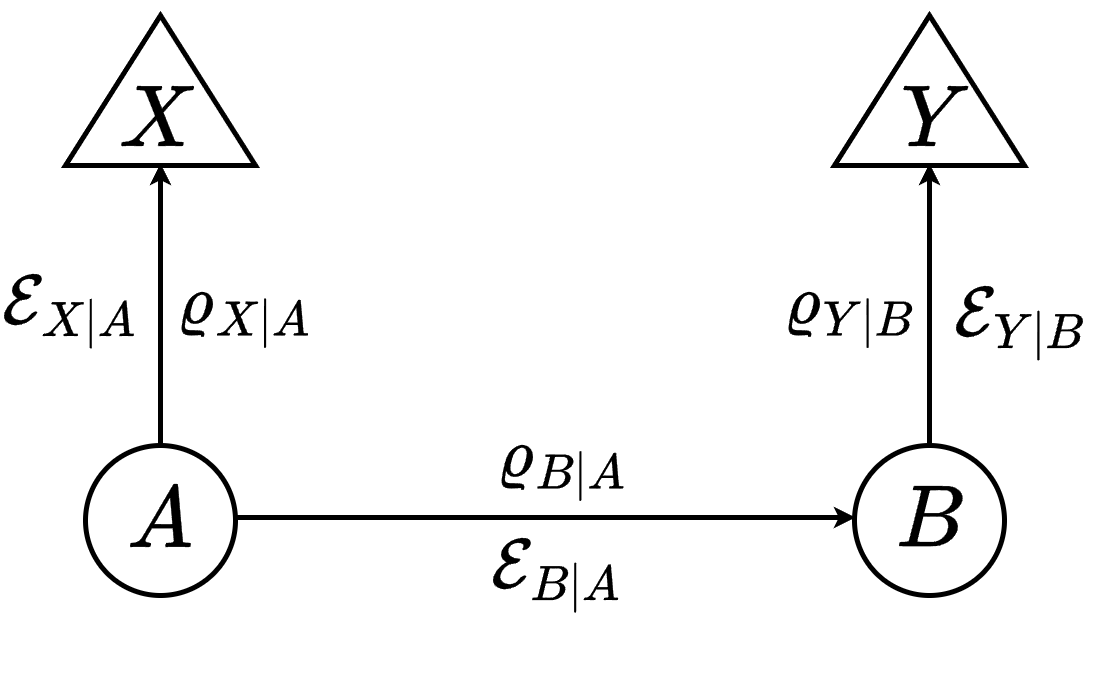}
    \caption{\textbf{Diagram representing the proposal of two measurements being made on the quantum regions $A$ and $B$.} In the bottom part of the diagram we have the two quantum regions $A$ and $B$ connected via an CPTP channel $\mathcal{E}_{B\vert A}$ and the associated conditional state $\varrho_{B\vert A}$. In the vertical branches of the diagram we have the two arrows, one that originates in the quantum region $A$ and terminates in the classical region $X$ and the second that originates in the quantum region $B$ and terminates in the classical region $Y$ both representing the measurement process in the respective quantum regions. The hybrid conditional states $\varrho_{X\vert A}$ and $\varrho_{Y\vert B}$ are the ones that represent the measurements in the CSF, isomorphic to the corresponding channels $\mathcal{E}_{X \vert A}$ and $\mathcal{E}_{Y \vert B}$.}
\label{fig_diagram_hybrid_measurements_no_upper_line}
\end{figure}

Starting from the expression \eqref{Ex.classicaljointhybrid}, and by the Bayesian inversion rule for conditional states, $\varrho_{X\vert A} = \varrho_{A\vert X} \star (\rho_{X} \rho_{A}^{-1})$, we obtain
\begin{align}\label{Ex.hybridjointstateXY}
    \varrho_{XY} &= \text{Tr}_{AB}\left[ (\varrho_{A\vert X} \star (\rho_{X} \rho_{A}^{-1})) \varrho_{Y\vert B} (\varrho_{B\vert A} \star \rho_{A}) \right] \nonumber \\ \nonumber
    &= \text{Tr}_{AB} \left[ \rho_{X}^{1/2} \rho_A^{-1/2} \varrho_{A\vert X} \rho_{X}^{1/2} \rho_{A}^{-1/2}  \varrho_{Y\vert B} \rho_{A}^{1/2} \varrho_{B\vert A} \rho_{A}^{1/2} \right] \\ \nonumber
    &= \text{Tr}_{AB} \left[ \rho_{X}^{1/2}\varrho_{A\vert X} \rho_{X}^{1/2} \varrho_{Y\vert B} \varrho_{B\vert A}\right] \\ 
    &= \text{Tr}_{AB} \left[ \varrho_{Y\vert B} \varrho_{B\vert A} \varrho_{A\vert X}\right] \star \rho_{X}.
\end{align}
By closely examining expression~\eqref{Ex.hybridjointstateXY} and once more applying the joint rule for conditional states, we are able to derive that
\begin{align}
\text{Tr}_{AB} \left( \varrho_{Y\vert B} \varrho_{B\vert A} \varrho_{A\vert X}\right) = \varrho_{Y\vert X},
\end{align}
and this conditional state is precisely the one associated with the propagation of beliefs between the classical region $X$ to the classical region $Y$. That is,
\begin{equation}\label{Ex.leftorightprop}
    \varrho_{Y\vert X} = \text{Tr}_{AB} \left( \varrho_{Y\vert B} \varrho_{B\vert A} \varrho_{A\vert X}\right),
\end{equation} 
where we used that hybrid, causal and acausal states are the same.

What is curious is that this problem is symmetrical. If we had initially access to the unitary channel $\mathcal{U}_{A\vert B}=\mathcal{E}_{A\vert B}$ instead of the unitary channel $\mathcal{U}_{B\vert A}$, in other words, if we had the lower arrow in the diagram of Fig.~\ref{fig_diagram_hybrid_measurements_no_upper_line} to be in the opposite direction, then we would work with the hybrid state $\varrho_{B\vert Y}$, obtaining the conditional state
\begin{equation}\label{Ex.righttoleftprop}
    \varrho_{X\vert Y} = \text{Tr}_{AB} \left( \varrho_{X\vert A} \varrho_{A\vert B} \varrho_{B\vert Y}\right),
\end{equation} 
associated with the belief propagation from $Y$ to $X$. It is important to emphasize that obtaining $\varrho_{A\vert B}$ through Bayesian inversion of $\varrho_{B\vert A}$, and consequently its associated channel, is not a trivial task within CSF, as argued in the section~\ref{SubSec.FullyQuantumPetz}. However, the Bayesian inversion of the hybrid conditional states that appears in the vertical branches of the diagram in Fig.~\ref{fig_diagram_hybrid_measurements_no_upper_line} shows no major limitations. The limitations directly related to more general scenarios with will deal with later on do not necessarily apply when the two regions at the bottom of the diagram are acausally related, and this picture is extensively discussed in ref.~\cite{LS13}.

Therefore, through the characterisation provided by the Theorem~\ref{teo:ChannelComposition}, we thus have that the expression \eqref{Ex.leftorightprop} is associated with the channel,
\begin{equation}\label{Ex.channellefttoright}
   \Gamma^{LR}_{M} = \mathcal{E}_{Y\vert X} = \mathcal{E}_{Y\vert B}^{M} \circ \mathcal{U}_{B\vert A} \circ \mathcal{E}_{A\vert X}^{M},
\end{equation}
and expression \eqref{Ex.righttoleftprop} is associated with the channel,
\begin{equation}\label{Ex.channelrighttoleft}
    \Gamma^{RL}_{M} = \mathcal{E}_{X\vert Y} = \mathcal{E}_{X\vert A}^{M} \circ \mathcal{U}_{A\vert B} \circ \mathcal{E}_{B\vert Y}^{M},
\end{equation}
where the channels presented above represent the emergent dynamics that connects the uppermost part of the diagram in~\ref{fig_diagram_hybrid_measurements_no_upper_line}, with $LR$ (left-to-right) and $RL$ (right-to-left) propagations. That is, starting from joint measurements in two quantum regions connected via a channel, under the umbrella of the coarse-graining problem, we construct a emergent dynamics that connects the two coarse-grained regions. Such elaboration is diagrammatically represented in the Fig.~\ref{Fig_twomeasurementshybrid}.
\begin{figure}[h]
\centering
\subfigure[]{
     \includegraphics[width=0.3\linewidth]{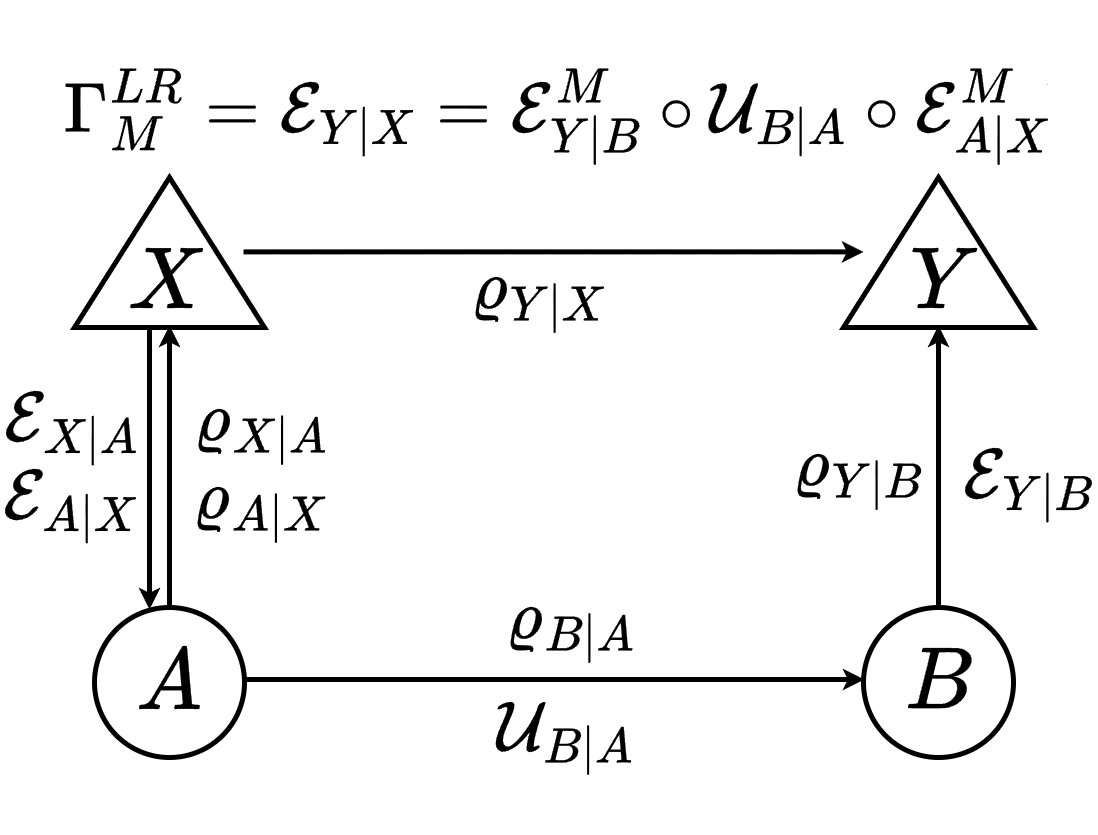}\label{subfig_measurements_leftotoright}}\hspace{0.5cm}
     \hspace{1.0cm}
\subfigure[]{
    \includegraphics[width=0.3\linewidth]{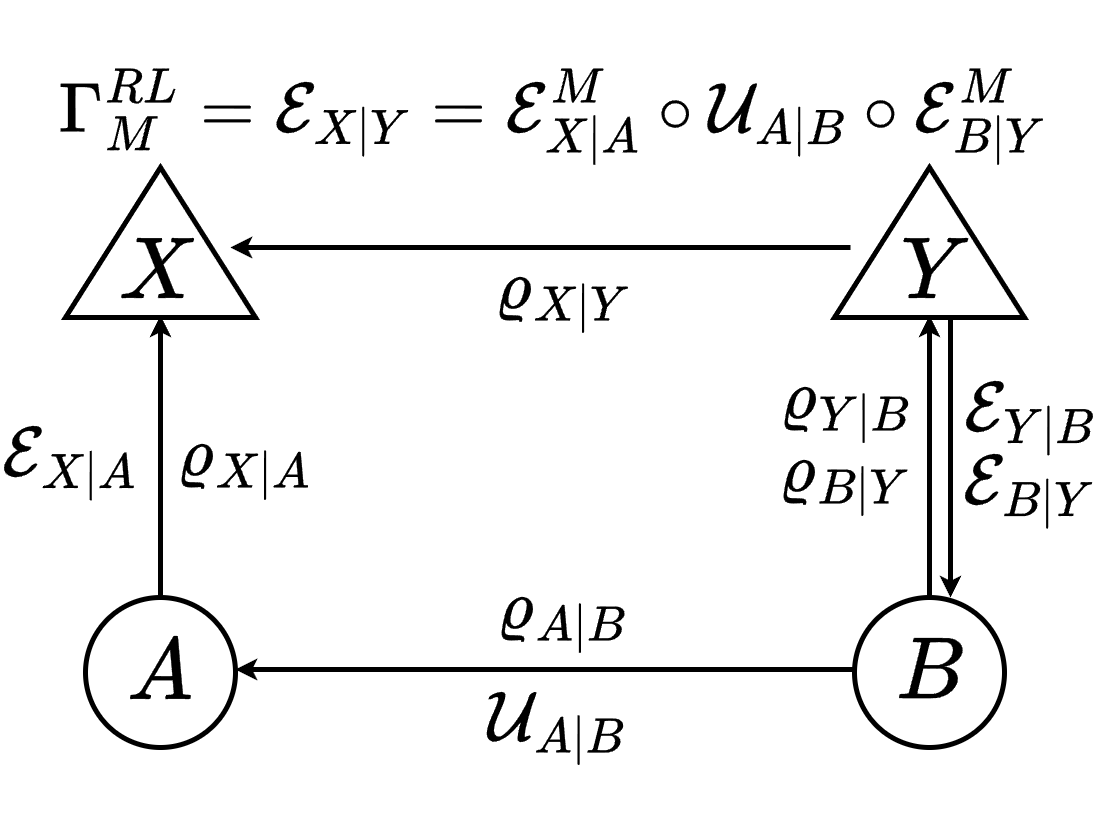}\label{subfig_measurements_righttoleft}}\\
\caption{\textbf{Diagram representing the two emergent dynamics obtained through the CSF}. In (a) the uppermost branch of the diagram shows the left-to-right belief propagation that was obtained through the inversion of the left branch of the diagram, that is, the channel $\mathcal{E}_{A\vert X}$ was obtained through the Bayesian inversion of the conditional state associated with the measurement procedure. In this way we were able to obtain the channel $\mathcal{E}_{Y\vert X}$. In (b) if instead of having access to the channel $\mathcal{U}_{B\vert A}$ we had access to a channel $\mathcal{U}_{A\vert B}$ we would be able to find the belief propagation right-to-left using a similar strategy of the previous case, that is, making the Bayesian inversion of the right branch of the diagram.} 
\label{Fig_twomeasurementshybrid}
\end{figure}

The expressions \eqref{Ex.channellefttoright} and \eqref{Ex.channelrighttoleft} could equally be obtained solely by using Theorem~\ref{teo:ChannelComposition}. However, unlike what we did here, this way we would only be thinking about belief propagation instead of simultaneous measurements in quantum regions. Thus, this approach constitutes one of the main differences from what we did in the previous sections.

We have shown that when dealing with the coarse-graining problem as a joint-measurement problem, we were able not only to describe the belief propagation between classical regions, but also to determine its direction based on the orientation of the channel connecting the quantum regions in the lower part of the diagram. Specifically, when the channel is directed from $A$ to $B$, the resulting emergent dynamics is directed from $X$ to $Y$. Conversely, when the channel runs from $B$ to $A$, the emergent dynamics correspondingly runs from $Y$ to $X$. Furthermore, the calculations above demonstrate once again the computational power of the CSF, since we can switch from the conditional state to the joint state and vice versa without major limitations and in a practical manner.

\section{Hybrid Case - II: Preparations}\label{SubSec.SolutionClassicalQuantum}

Taking a closer look at an ensemble preparation procedure, we assume that this process maps a classical region of dimension $n$ to a quantum region of dimension $d < n$. Under this assumption, we can also interpret this process within the coarse-graining problem setup. In accordance with the measurement case, our goal is also to replace the vertical arrows in the diagram of Fig.~\ref{fig:Coarse_graining_diagram}, with ensemble preparation channels. In this setting, since the process maps a classical region to a quantum one, the lower part of the diagram consists of two classical regions connected by a channel, in contrast to the measurement case, where the lower part contained two quantum regions.

We structure our investigation here as follows: we consider two classical regions $X$ and $Y$ connected via a unitary channel $\mathcal{U}_{Y\vert X}: \mathcal{L}(\mathcal{H}_{X}) \rightarrow \mathcal{L}(\mathcal{H}_{Y})$ and ensemble preparations are performed concomitantly in the classical regions, taking the classical region $X$ to the quantum region $A$ and the classical region $Y$ to the quantum region $B$. We picture this scenario in the diagram of the Fig.~\ref{fig_ensembleprep_noupperarrow}, and once more we set sail to seek for the emergence of an effective dynamics.
\begin{figure}
    \centering
    \includegraphics[width=0.3\linewidth]{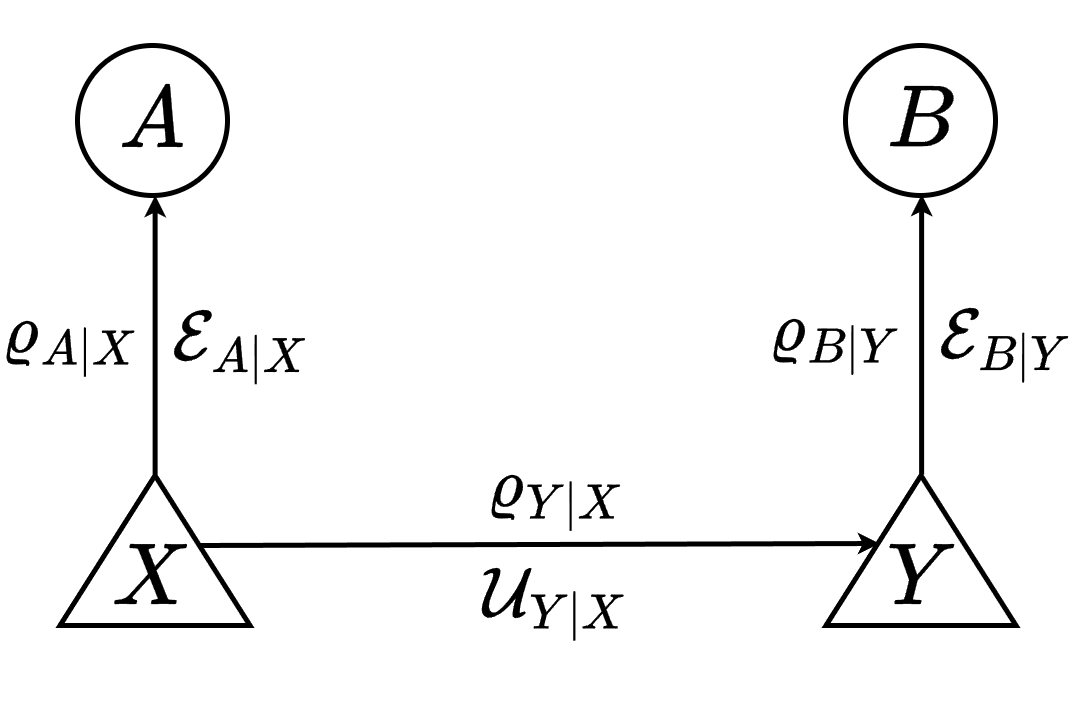}
    \caption{\textbf{Diagram representing the coarse-graining scenario when we have to ensemble preparations.} In this diagram, we have to classical regions $X$ and $Y$ connected via a CPTP channel $\mathcal{U}_{Y\vert X}$ and a ensemble preparation in made over the region $X$ leading to a quantum region $A$ and an ensemble preparation is made in the region $Y$ leading to the quantum region $B$. The CPTP channels associated with the ensemble preparation and their associated conditional states are represented in the vertical branches of the diagram.}
    \label{fig_ensembleprep_noupperarrow}
\end{figure}

Therefore, by an argument similar to the measurement case, we obtain here,
\begin{equation}
    \varrho_{AB} = \operatorname{Tr}_{XY}\left[(\varrho_{A\vert X}\varrho_{B\vert Y}) \varrho_{XY} \right],
\end{equation}
where $\varrho_{XY}$ can be obtained assuming that we have access to an initial state $\rho_{X}$ of the region $X$ and using the joint state rule yielding $\varrho_{XY} = \varrho_{Y\vert X} \star \rho_{X}$ with $\varrho_{Y\vert X}$ being the conditional state isomorphic to the channel $\mathcal{U}_{Y\vert X}$.

In the same way that we proceeded with the calculations in the expression \eqref{Ex.hybridjointstateXY}, we obtain
\begin{align} \label{Ex.jointprepcoarse}
    \varrho_{AB} &= \operatorname{Tr}_{XY}\left[(\varrho_{A\vert X}\varrho_{B\vert Y}) \varrho_{XY} \right] \nonumber \\ 
    & = \operatorname{Tr}_{XY}\left[\varrho_{B\vert Y}\varrho_{Y \vert X}\varrho_{X \vert A} \right] \star \rho_{A}.
\end{align}
Analyzing expression \eqref{Ex.jointprepcoarse} with the help of the joint state rule, we are able to infer that the term $\operatorname{Tr}_{XY}\left[\varrho_{B\vert Y}\varrho_{Y \vert X}\varrho_{X \vert A} \right] = \varrho_{B\vert A}$ is precisely the conditional state  associated with the belief propagation of the upper quantum regions. And by the characterization provided by Theorem~\ref{teo:ChannelComposition}, we have that resulting channel is given by,
\begin{equation}\label{Ex.channelprepLR}
    \Gamma_{EP}^{LR}=\mathcal{E}_{B\vert A} =  \mathcal{E}_{B\vert Y}^{EP} \circ  \mathcal{U}_{Y\vert X} \circ  \mathcal{E}_{X\vert A}^{EP},
\end{equation}
with expression \eqref{Ex.channelprepLR} being the emergent effective dynamics. 

Different from the measurement case, where we started from two connected quantum regions, here, because we depart from the classical ones, we does not have the need to require an opposite quantum dynamics in order to obtain the emergent dynamics to be in opposite direction. Rather, because we are dealing with two lower classical regions, we can therefore use the Bayesian inversion rule for conditional states and obtain $\varrho_{X\vert Y}$,\footnote{It's important to emphasize that with fully classical conditional states we do not have any problem in representing as a causal or acausal notation simply because these states respects both definitions.} and with similar calculations we would obtain
\begin{equation}\label{Ex.channelprepRL}
    \Gamma_{EP}^{RL} = \mathcal{E}_{A\vert B} =  \mathcal{E}_{A\vert X}^{EP} \circ  \mathcal{U}_{X\vert Y} \circ  \mathcal{E}_{Y\vert B}^{EP}.
\end{equation}

\begin{figure}[h]
\centering
\subfigure[]{
     \includegraphics[width=0.3\linewidth]{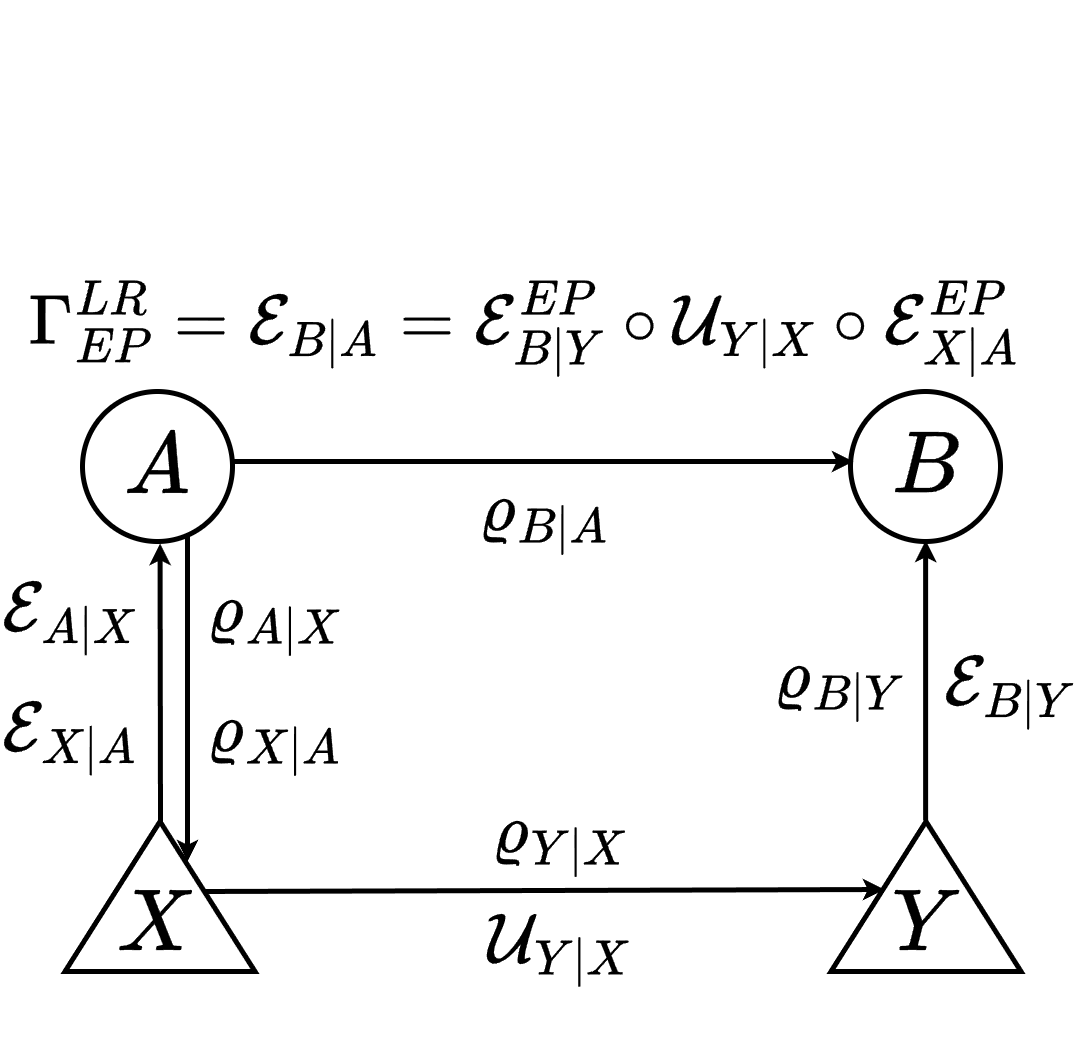}\label{subfig_prop_leftotoright}}\hspace{0.5cm}
     \hspace{1.0cm}
\subfigure[]{
    \includegraphics[width=0.3\linewidth]{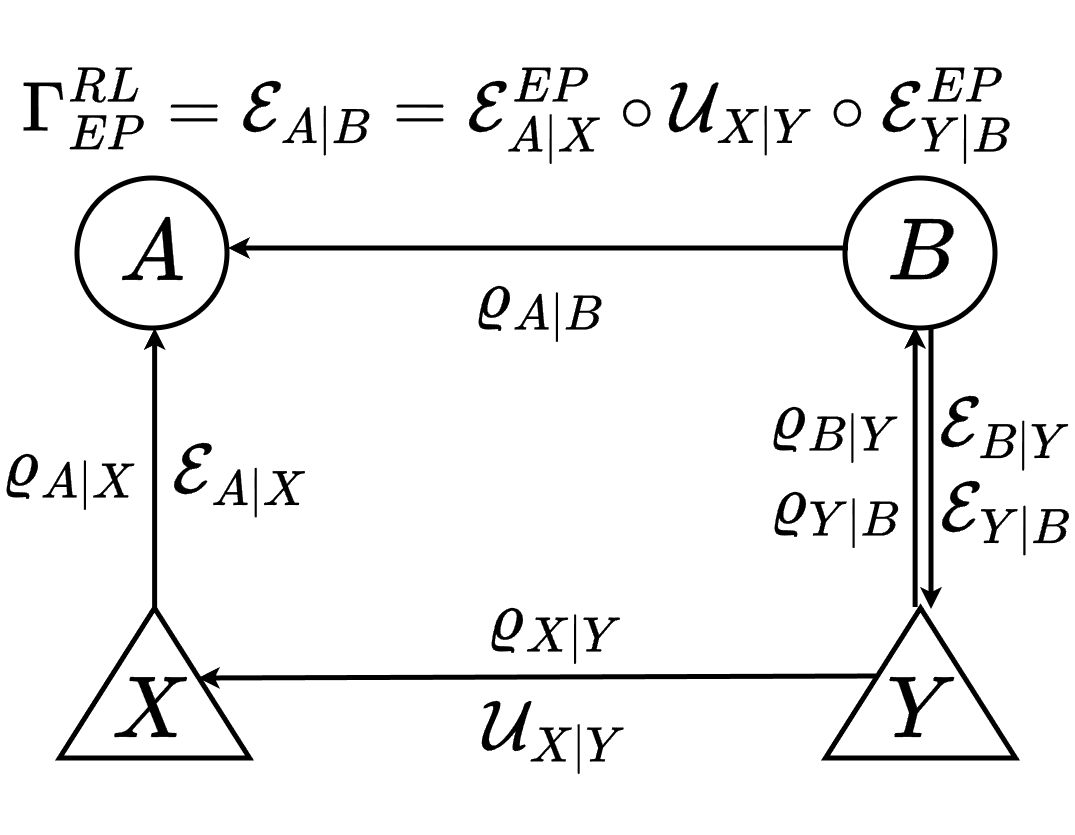}\label{subfig_prop_righttoleft}}
\caption{\textbf{Diagram representing the two possible solutions to the ensemble preparation case.} In the figure a) the uppermost branch of the diagram shows the left-to-right belief propagation that was obtained through the inversion of the left branch of the diagram, that is, the channel $\mathcal{E}_{X\vert A}$ was obtained through the Bayesian inversion of the conditional state associated with the ensemble preparation procedure, in this way we where able to obtain the channel $\mathcal{E}_{B\vert A}$. In the figure b) since we are dealing with two classical regions in the lower part of the diagram, we are able to do de Bayesian inversion of the conditional state $\varrho_{Y\vert X}$ to obtain $\varrho_{X\vert Y}$ and consequently find his isomorph channel to finally be able to obtain the channel $\mathcal{E}_{A\vert B}$.} 
\label{Fig_twoprepemergent_channels}
\end{figure}

The diagrams presented in the Figs.~\ref{subfig_prop_leftotoright} and \ref{subfig_prop_righttoleft} illustrates both cases. Thus, we observe that, when considering ensemble preparation processes, we are likewise led to an emergent effective dynamics that depend directly on the direction of the arrow connecting the lower regions. In other words, when the ensemble preparation scenario is treated as a particular instance of the coarse-graining problem, we are able to obtain, once more, a solution without major limitations.

\section{Petz Solution}\label{APX.PetzSolution}

The conditional state $\varrho_{A\vert C}$ can be obtained as follows. Consider that all we have access to is the state $\varrho_{C\vert A}$ and an initial state of region $A$, namely $\rho_{A}$. We can use the joint-state rule to obtain
\begin{equation}
    \varrho_{AC} = \varrho_{C \vert A} \star \rho_{A},
\end{equation}
and marginalize it in order to obtain,
\begin{equation}
    \rho_{C}  = \operatorname{Tr}_{A}(\varrho_{AC}).
\end{equation}
Therefore via the Bayesian inversion of operators, we achieve,
\begin{equation}\label{Ex.apxbayesinverfullyquantum}
    \varrho_{A\vert C} =  \varrho_{C\vert A} \star (\rho_{A}\rho_{C}^{-1}).
\end{equation}
That is, through the initial state $\rho_{A}$ and the marginal state $\rho_{C}$, we find the expression \eqref{Ex.bayesinverfullyquantum} to be the Bayesian inversion of expression \eqref{Ex.isomorphoCA}.

However, the above inversion falls into a limitation. To help understanding this limitation, let us briefly turn our attention to another scenario structured as follows: assume we have 3 quantum regions $A$, $B$, and $C$, and that these regions are connected to each other via the respective channels $\mathcal{E}_{B\vert A} : \mathcal{L}(\mathcal{H}_A) \rightarrow \mathcal{L}(\mathcal{H}_B)$ and $\mathcal{E}_{C\vert B} : \mathcal{L}(\mathcal{H}_B) \rightarrow \mathcal{L}(\mathcal{H}_C)$. We can additionally interpret that the region $B$ is in the causal future of region $A$ and that region $C$ is in the causal future of region $B$. Their corresponding isomorphic states are then $\varrho_{B\vert A}$ and $\varrho_{C\vert B}$. Assuming we have access to an initial state $\rho_{A}$, we could try to obtain the state $\rho_{C}$ as,
\begin{align}\label{ex.rho_cdynamics}
      \rho_C = \operatorname{Tr}_B(\varrho_{C\vert B}\operatorname{Tr}_A(\varrho_{B\vert A}\rho_A)).
\end{align}
On the other hand, we could try to find a causal joint-state for the three regions, that is, $\varrho_{ABC}$ such that, by taking the trace over $AB$, we would obtain a marginal state $\rho_C$ consistent with the one presented in the expression \eqref{ex.rho_cdynamics}. That is, we would expect that
\begin{equation}\label{ex.triequalitynonhold}
\operatorname{Tr}_{B}\bigl(\varrho_{C|B}\operatorname{Tr}_{A}(\varrho_{B \vert A}\rho_A)\bigr) = \rho_C = \operatorname{Tr}_{AB}(\varrho_{ABC}).
\end{equation}
Since $\varrho_{AB} = \varrho_{B \vert A}\star \rho_A$ is a joint-state of the region $AB$, we could guess that $\varrho_{ABC}$ would be of the form,
\begin{equation}\label{Ex.triregionjoint}
       \varrho_{ABC} = \varrho_{C|B} \star (\varrho_{B \vert A}\star \rho_A).
\end{equation}
However, since $\varrho_{B \vert A}$ is not a positive operator in general,\footnote{Since the causal conditional state defined through the partial transpose of an acausal conditional state, as highlighted in Def.~\ref{def:CausalState}, is a locally positive operator, but it is not a positive operator in general.} the expression \eqref{Ex.triregionjoint} does not hold true for states, because given the star product, we need to take the square root of that operator, that might assume negative eigenvalues. And therefore, we have that expression \eqref{ex.triequalitynonhold} is not achievable for states.
%

The scenario illustrated above is exactly what we have in the left arm of the Fig.~\ref{fig:CGandCSF}. We go from $A$ to $C$ and return to $A$. Because of this, we suffer these limitations in obtaining the marginal states, which impacts directly in the desired Bayesian inversion and consequently our proposed solution to the emergent dynamics $\Gamma^t$. Nonetheless, within the CSF itself, we have a way to overcome such limitations.

\section{Proof of Theorems}\label{Sec.AppProof}


\begin{lemma}
\label{lemm:ChoiActing} Let $\mathcal{H}_A$ and $\mathcal{H}_B$ be two finite-dimensional Hilbert spaces. Let
\begin{equation}
    \mathcal{E}_{B|A}: \mathcal{L}(\mathcal{H}_A) \to \mathcal{L}(\mathcal{H}_B) 
\end{equation}
be a linear map. If an operator $\rho_{B\vert A} \in \mathcal{L}(\mathcal{H}_A \otimes \mathcal{H}_B)$ is the Choi image of $\mathcal{E}_{B\vert A}$, then the action of $\mathcal{E}_{B\vert A}$ on any state $\sigma_A  \in \mathcal{L}(\mathcal{H}_A)$ can be written as
\begin{equation}
    \mathcal{E}_{B|A}(\sigma_{A}) = \Tr_A[\rho_{B|A}(\sigma_{A}^{T} \otimes \mathds{I}_B) ] \in \mathcal{L}{(\mathcal{H}_B)}.
\end{equation}

\end{lemma}
\begin{proof}
In other words, here we will check the validity of Eq. \eqref{eq:ChoiActing}. For this, suppose that $\rho_{B\vert A} \in \mathcal{L}(\mathcal{H}_A \otimes \mathcal{H}_B)$ is the Choi image of $\mathcal{E}_{B \vert A}$, and let $\sigma_A \in \mathcal{L}(\mathcal{H}_A)$. Then, it follows from Definition \ref{def:ChoiIso} that
\begin{align} \label{eq:MapAction}
     &\Tr_A[\rho_{B\vert A} (\sigma^T_{A} \otimes \mathds{I}_B)] =  \notag \\
     & = \Tr_A [ (id_A \otimes \mathcal{E}_{B \vert A}) (\ket{\Phi^+}\bra{\Phi^+}) (\sigma^T_{A} \otimes \mathds{I}_B) ] \notag\\
      & = \Tr_A [ (id_A \otimes \mathcal{E}_{B \vert A}) (\sum_{i,j=0}^{d-1} \ket{i}\bra{j} \otimes \ket{i}\bra{j}) (\sigma^T_{A} \otimes \mathds{I}_B) ] \notag \\
      & = \Tr_A [\sum_{i,j=0}^{d-1} \ket{i}\bra{j} \sigma^T_{A} \otimes \mathcal{E}_{B \vert A}(\ket{i}\bra{j})] = \sum_{i,j=0}^{d-1} \bra{j} \sigma^T_{A} \ket{i} \mathcal{E}_{B \vert A}(\ket{i}\bra{j}) \notag\\ 
      & = \sum_{i,j=0}^{d-1} (\sigma^T_{A})_{ji} \mathcal{E}_{B \vert A}( \ket{i} \bra{j}) = \sum_{i,j=0}^{d-1} (\sigma_A)_{ij} \mathcal{E}_{B \vert A}( \ket{i} \bra{j}) = \mathcal{E}_{B  \vert A}(\sigma_A).  
\end{align}
\end{proof}

\begin{thm*}[\ref{teo:ConnectsCPTPtoPSD}]
    Let $\mathcal{E}: \mathcal{L}(\mathcal{H}_A) \to \mathcal{L}(\mathcal{H}_B)$ be a linear map and let $\rho \in \mathcal{L}(\mathcal{H}_A \otimes \mathcal{H}_B)$ be its Choi-isomorphic operator. Then, it follows that $\rho$ satisfies: 
    \begin{enumerate}
        \item $\rho \geq0 ;$
        \item $\Tr_B(\rho) = \mathds{I}_A$,
    \end{enumerate}
    if, and only if, $\mathcal{E}$ is CPTP.
\end{thm*}
\begin{proof}
    Suppose that $\rho$ is positive and $\Tr_B(\rho) = \mathds{I}_A$. Let $N \in \mathcal{L}(\mathcal{H}_A)$ be and arbitrary operator on $\mathcal{H}_A$. Then, first, to show that $\mathcal{E}$ is trace preserving, it proceeds as follows: 
    \begin{align}
        \Tr_B[\mathcal{E}(N)] & =  \Tr_B\{\Tr_A [\rho (N^T \otimes \mathds{I}_B)] \} \notag \\
        & = \Tr_A[\Tr_B (\rho) N^T] \notag \\
        & = \Tr_A(\mathds{I}_A N^T) \notag \\
        & = \Tr_A(N). 
    \end{align}
To show that $\mathcal{E}$ is completely positive, considering a positive operator $P \in \mathcal{L}({\mathcal{H}_{C} \otimes\mathcal{H}_{A} })$, where $\mathcal{H}_{C}$ is any finite-dimensional Hilbert space, we need to show that 
\begin{equation}
    (\text{id}_C \otimes \mathcal{E})(P) \geq 0.
\end{equation}
To do this, first, as $P$ is positive, consider its spectral decomposition 
\begin{equation}
    P = \sum_{i=1}^{d_A \cdot d_C} p_i \ket{\psi_i}\bra{\psi_i}_{CA}, 
\end{equation}
 where $d_A := dim(\mathcal{H}_A)$, $d_C := dim(\mathcal{H}_C)$, $\{\ket{\psi_i}_{CA}\}_{i=1}^{d_A\cdot d_C}$ is some orthonormal basis for $\mathcal{H}_C \otimes \mathcal{H}_A$ and $p_i \in spec(P)$. Thus, 
\begin{equation}
    (\text{id}_C \otimes \mathcal{E})(P)  = \sum_{i=1}^{d_A d_C}p_i[\text{id}_C \otimes \mathcal{E}(\ket{\psi_i} \bra{\psi_i}_{CA})]. 
\end{equation}
 Since each eigenvalue $p_i$ of $P$ is non-negative, it's suffice to prove that  
 \begin{equation}
     \text{id}_C \otimes \mathcal{E}(\ket{\psi_i} \bra{\psi_i}_{CA}) \geq 0
 \end{equation}
 for a given $i \in \{1,..., d_Ad_C\}$. For this, letting $\alpha := \{\ket{j}_A\}_{j=0}^{d_A}$ and $\gamma : = \{\ket{k}_C\}_{k=0}^{d_C}$ be orthonormal basis of $\mathcal{H}_A$ and $\mathcal{H}_C$, respectively, we can expand each $\ket{\psi_i}_{CA}$ as: 
 \begin{equation}\label{eq:ketPsi}
     \ket{\psi_i}_{CA} = \sum_{k= 0}^{d_C-1} \sum_{j=0}^{d_A -1} (\Psi_{i})_{kj} \ket{k}_C \ket{j}_A.  
 \end{equation}
Then, one note that
\begin{align}\label{eq:A7}
      \text{id}_C \otimes \mathcal{E}(\ket{\psi_i} \bra{\psi_i}_{CA}) & =  \text{id}_C \otimes \mathcal{E}[\sum_{k,m =0}^{d_C -1} \sum_{j,n =0}^{d_A -1}(\Psi_{i})_{kj}(\Psi_{i})^*_{mn}\ket{k}_C \ket{j}_A \bra{m}_C \bra{n}_A]  \notag  \\
      & = \sum_{k,m =0}^{d_C -1} \sum_{j,n =0}^{d_A -1}(\Psi_{i})_{kj}(\Psi_{i})^*_{mn} \ket{k} \bra{m}_C \otimes \mathcal{E}(\ket{j}\bra{n}_A).
\end{align}
By the very Definition \ref{def:ChoiIso}, we have that  
\begin{align}\label{eq:A8}
    \mathcal{E}(\ket{j}\bra{n}_{A}) & = (\bra{j}_A \otimes \mathds{I}_B)[\sum_{j',n'=0}^{d_A -1} \ket{j'} \bra{n'}_A  \otimes \mathcal{E}(\ket{j'}\bra{n'}_A) ](\ket{n
    }_A \otimes \mathds{I}_B) \notag \\ 
    & = (\bra{j}_A \otimes \mathds{I}_B)(\text{id}_A \otimes \mathcal{E})(\ket{\Phi^+}\bra{\Phi^+})(\ket{n}_A \otimes \mathds{I}_B) \notag \\
    & = (\bra{j}_A \otimes \mathds{I}_B) (\rho) (\ket{n}_A \otimes \mathds{I}_B),
\end{align} 
where $\ket{\Phi^+}  = \sum_{i = 0}^{d_A -1} = \ket{ii}_{AA}$ is the unnormalized maximally entangled state with respect to the basis $\alpha$. Replacing \eqref{eq:A8} in \eqref{eq:A7}, we reach that 
\begin{equation}\label{eq:A9}
     \text{id}_C \otimes \mathcal{E}(\ket{\psi_i} \bra{\psi_i}_{CA}) = \sum_{k,m =0}^{d_C -1} \sum_{j,n =0}^{d_A -1}(\Psi_{i})_{kj}(\Psi_{i})^*_{mn} \ket{k} \bra{m}_C \otimes[ (\bra{j}_A \otimes \mathds{I}_B)(\rho)(\ket{n}_A \otimes \mathds{I}_B)]. 
\end{equation}
Now, since $\rho$ is positive, it can be decomposed as 
\begin{equation}\label{eq:A10}
    \rho = \sum_{l=1}^{r} \ket{\phi_l}\bra{\phi_l}_{AB},    
\end{equation}
where $r \leq d_A d_B$ is the rank of $\rho$, $d_B : = dim(\mathcal{H}_B)$ and $\{\ket{\phi_l}\}_{l=1}^{r}$ is a set of orthogonal vectors of $\mathcal{H}_A\otimes \mathcal{H}_B$. Also, for a given $l$, consider the following expansion of the vector $\ket{\phi_l}_{AB}$ in terms of the basis $\alpha$ of $\mathcal{H}_A$ and a given orthonormal basis $\beta:= \{ \ket{n}_B\}_{n=0}^{d_B -1}$ of $\mathcal{H}_B$ as: 
\begin{equation}
    \ket{\phi_l}_{AB} = \sum_{j=0}^{d_A -1} \sum_{n=0}^{d_B-1} (\Phi_l)_{jn}\ket{j}_A \otimes \ket{n}_B,
\end{equation}
where, for each $l\in\{1,...,r\}$, $\Phi_l$ denotes the $d_A \times d_B$ matrix whose entries are the coefficients of $\ket{\phi_l}_{AB}$ in the product basis $\{\ket{j}_A \otimes \ket{n}_B\}$.
Let $V_l: \mathcal{H}_A \to \mathcal{H}_B$ denote linear operators of the form: 
\begin{equation}\label{eq:AuxOperators}
    V_l  := \sum_{j=0}^{d_A-1}\sum_{n=0}^{d_B-1}(\Phi^T_l)_{nj}\ket{n}_B\bra{j}_A = \sum_{j=0}^{d_A-1}\sum_{n=0}^{d_B-1}(\Phi_l)_{jn}\ket{n}_B\bra{j}_A.
\end{equation}
From this definition, one can note that
\begin{align}
    (\mathds{I}_A \otimes V_l) \ket{\Phi^+} & = \sum_{k=0}^{d_A -1} \ket{k}_A \otimes V_l \ket{k}_A \notag \\
    & = \sum_{k,j=0}^{d_A -1}\sum_{n=0}^{d_B -1} (\Phi_l)_{jn} \ket{k}_A \otimes \ket{n}_B \braket{j}{k}_A \notag \\
    & = \sum_{k,j=0}^{d_A -1}\sum_{n=0}^{d_B -1} (\Phi_l)_{jn}\ket{k}_A \otimes \ket{n}_B \delta_{jk} \notag \\
    &  = \sum_{j=0}^{d_A -1} \sum_{n=0}^{d_B-1} (\Phi_l)_{jn}\ket{j}_A \otimes \ket{n}_B = \ket{\phi_l}_{AB}. 
\end{align}
In other words, for each $l$, it's possible to find a linear operator $V_l$, as given in eq.   \eqref{eq:AuxOperators}, such that every bipartite state $\ket{\phi_l}_{AB} $ can be written in the form $\ket{\phi_l}_{AB}=(\mathds{I}_A \otimes V_l )\ket{\Phi^+}$. Besides this, note that: 
\begin{align}
    (\bra{j}_A \otimes \mathds{I}_B) \ket{\phi_l}_{AB} & =
    \bra{j}_A \otimes \mathds{I}_B\left[ \sum_{k=0}^{d_A -1}\sum_{n=0}^{d_B -1} (\phi_l)_{kn} \ket{k}_A \otimes \ket{n}_B\right] \notag \\
    & =  \sum_{k=0}^{d_A -1}\sum_{n=0}^{d_B -1} (\phi_l)_{kn} \ket{n}_B \delta_{jk} \notag\\
    & =  \left(\sum_{k=0}^{d_A -1}\sum_{n=0}^{d_B -1} (\phi_l)_{kn} \ket{n}_B \bra{k}_A\right)\ket{j}_A = V_l \ket{j}_A. 
\end{align}
With the above observation and considering the decomposition shown in eq. \eqref{eq:A10}, we can rewrite the right-hand side of the eq. \eqref{eq:A8} as follows: 
\begin{align}\label{eq:A15}
    (\bra{j}_A \otimes \mathds{I}_B) (\rho) (\ket{n}_A \otimes \mathds{I}_B) & =  (\bra{j}_A \otimes\mathds{I}_B) \left( \sum_{l=1}^{r} \ket{\phi_l} \bra{\phi_l}_{AB}\right)(\ket{n}_A \otimes\mathds{I}_B) \notag \\
    & = \sum_{l=1}^{r} \left[(\bra{j}_A \otimes \mathds{I}_B)\ket{\phi_l}_{AB}\right]\left[ \bra{\phi_l}_{AB}(\ket{n}_A \otimes \mathds{I}_B) \right] \notag\\
    & = \sum_{l=1}^{r} V_l \ket{j} \bra{n}_A V_l^{\dagger}.
\end{align}
Thus, replacing the result of eq. \eqref{eq:A15} in eq. \eqref{eq:A9}, we have that, 
\begin{align}
    \text{id}_C \otimes \mathcal{E}(\ket{\psi_i} \bra{\psi_i}_{CA}) & = \sum_{k,m =0}^{d_C -1} \sum_{j,n =0}^{d_A -1}(\Psi_{i})_{kj}(\Psi_{i})^*_{mn} \ket{k} \bra{m}_C \otimes[ (\bra{j}_A \otimes \mathds{I}_B)\left( \sum_{l=1}^{r} \ket{\phi_l}\bra{\phi_l}_{AB}\right)(\ket{n}_A \otimes \mathds{I}_B)] \notag \\
    & = \sum_{k,m =0}^{d_C -1} \sum_{j,n =0}^{d_A -1}(\Psi_{i})_{kj}(\Psi_{i})^*_{mn} \ket{k} \bra{m}_C \otimes \left( \sum_{l=1}^{r} V_l \ket{j}\bra{n}_A V_l^{\dagger} \right) \notag\\
    & = \sum_{l=1}^{r}\left\{ \text{id}_C\otimes V_l \left[\sum^{d_C -1}_{k,m=0} \sum^{d_A-1}_{j,n=0} (\Psi_{i})_{kj}(\Psi_{i})^*_{mn} \ket{k} \bra{m}_C  \otimes \ket{j}\bra{n}_A\right] \text{id}_C \otimes V_l^{\dagger} \right\} \notag \\
    & = \sum_{l=1}^{r} \text{id}_C \otimes V_l (\ket{\psi_i} \bra{\psi_i}_{CA}) \text{id}_C \otimes V_l^{\dagger}.
\end{align}
Finally, defining $K_l : = \text{id}_C \otimes V_l$ for all $l\in \{1,...,r\}$, we can rewrite $\text{id}_C \otimes \mathcal{E}(\ket{\psi_i} \bra{\psi_i}_{CA})$ as 
\begin{equation}
    \text{id}_c \otimes \mathcal{E}(\ket{\psi_i} \bra{\psi_i}_{CA})   = \sum_{l=1}^{r} K_l(\ket{\psi_i} \bra{\psi_i}_{CA})K_l^\dagger \geq0,
\end{equation}
where this follows from the fact that, for any $\ket{\omega}_{CB} \in \mathcal{H}_C \otimes \mathcal{H}_B$, and for each $l$, we have 
\begin{equation}
    \bra{\omega}_{CB}K_l( \ket{\psi_i}\bra{\psi_i}_{CA})K_l^{\dagger} \ket{\omega}_{CB} = | \bra{\omega}_{CB}K_l \ket{\psi_i}_{CA}|^2 \geq0. 
\end{equation}
Concluding, in that way, the first half of the proof.

On the other side, suppose that $\mathcal{E}$ is CPTP. Then, it follows that, by the very definition of its Choi state $\rho$ and by the fact that $\mathcal{E}$ is trace preserving, 
\begin{align}
    \Tr_B(\rho) & = \Tr_B[\sum_{i,j = 0}^{d_A-1} \ket{i}\bra{j} \otimes \mathcal{E}(\ket{i}\bra{j})] \\
    & = \sum_{i,j = 0}^{d_A-1} \ket{i}\bra{j} \otimes\Tr_B[\mathcal{E}(\ket{i}\bra{j})] \notag \\
    & = \sum_{i,j = 0}^{d_A-1} \ket{i}\bra{j}  \otimes\Tr_A(\ket{i}\bra{j}) \notag \\
    & = \sum_{i,j = 0}^{d_A-1} \ket{i}\bra{j} \delta_{ij} = \sum_{i = 0}^{d_A-1} \ket{i}\bra{i} = \mathds{I}_A. \notag
\end{align}
Also, again by the very definition of the Choi state $\rho$, and using that $\mathcal{E}$ is completely positive, it follows that  
\begin{equation}
\rho =     (\text{id}_A \otimes \mathcal{E}) \ket{\Phi^+}\bra{\Phi^+} \geq 0, 
\end{equation}
since it is a completely positive map acting on a positive operator of $\mathcal{L}(\mathcal{H}_A \otimes \mathcal{H}_A)$.
\end{proof}

\begin{thm*}[\textbf{\ref{teo:ChannelComposition}}]
Let $\mathcal{E}_{B\vert A}, \mathcal{E}_{C\vert B}$ and $\mathcal{E}_{C\vert A}$  be linear maps, and $\varrho_{B\vert A}, \varrho_{C\vert B}$ and $\varrho_{C\vert A}$, their respective Jamio\l{}kowski isomorphic operators. Then, it holds that $\mathcal{E}_{C\vert A} = \mathcal{E}_{C\vert B} \circ \mathcal{E}_{B\vert A}$ if, and only if, the Jamio\l{}kowski isomorphic operators satisfy  
\begin{equation}
    \varrho_{C\vert A} = \Tr_B[(\mathds{I}_A\otimes\varrho_{C\vert B}) (\varrho_{B\vert A} \otimes \mathds{I}_C)].
\end{equation}
\end{thm*}

\begin{proof}
    Firstly, suppose that $\mathcal{E}_{C\vert A} = \mathcal{E}_{C\vert B} \circ \mathcal{E}_{B\vert A}$. Let $M_A \in \mathcal{L}(\mathcal{H}_A)$ be an arbitrary operator on $\mathcal{H}_A$. Then, 
    \begin{align}
         \mathcal{E}_{C\vert A}(M_A) & = \Tr_A[\varrho_{C\vert A} (M_A \otimes \mathds{I}_C)] \notag \\ &= \mathcal{E}_{C\vert B}[\mathcal{E}_{B\vert A}(M_A)] \notag \\
         & = \mathcal{E}_{C\vert B}\{\Tr_A[\varrho_{B\vert A} (M_A\otimes \mathds{I}_B)]\} \notag\\
         & = \Tr_{B}\{ \varrho_{C\vert B} [[ \Tr_A(\varrho_{B\vert A}(M_A \otimes\mathds{I}_B))]  \otimes \mathds{I}_C]\} \notag \\
         & = \Tr_A\{ \Tr_B[(\mathds{I}_A\otimes\varrho_{C\vert B}) (\varrho_{B\vert A}  \otimes \mathds{I}_C)] (M_A \otimes \mathds{I}_C)\} \notag \\
          &\implies \varrho_{C\vert A} =\Tr_B[(\mathds{I}_A\otimes\varrho_{C\vert B}) (\varrho_{B\vert A}  \otimes \mathds{I}_C)].
    \end{align}
Conversely, suppose that 
\begin{equation}
\varrho_{C\vert A} =\Tr_B[(\mathds{I}_A\otimes\varrho_{C\vert B}) (\varrho_{B\vert A}  \otimes \mathds{I}_C)].    
\end{equation}
Then, 
\begin{align}
     \mathcal{E}_{C\vert A}(M_A) & = \Tr_A[\varrho_{C\vert A}(M_A \otimes \mathds{I}_C)] \notag \\
    & = \Tr_A\{ \Tr_B[(\mathds{I}_A\otimes\varrho_{C\vert B}) (\varrho_{B\vert A}  \otimes \mathds{I}_C)](M_A \otimes \mathds{I}_C)\} \notag\\ 
    & = \Tr_B \{ \varrho_{C\vert B} [[\Tr_A(\varrho_{B\vert A}(M_A \otimes \mathds{I}_B ))]  \otimes \mathds{I}_C]\} \notag \\
    & = \Tr_B\{\varrho_{C\vert B} [[\mathcal{E}_{B\vert A}(M_A)]  \otimes \mathds{I}_C] \} \notag\\
    & = \mathcal{E}_{C\vert B}[\mathcal{E}_{B\vert A}(M_A)] \notag \\
     &\implies \mathcal{E}_{C\vert A} = \mathcal{E}_{C\vert B} \circ \mathcal{E}_{B\vert A}. 
\end{align}
\end{proof}

\end{widetext}

\end{document}